\documentclass[11pt]{article}

\usepackage{arxiv}

\usepackage[utf8]{inputenc} 
\usepackage[T1]{fontenc}    
\usepackage{hyperref}       
\usepackage{url}            
\usepackage{booktabs}       
\usepackage{amsfonts}       
\usepackage{nicefrac}       
\usepackage{microtype}      
\usepackage{lipsum}
\usepackage{graphicx}

\usepackage{ifthen}
\usepackage{color}
\usepackage{amsmath,amsthm,amssymb,amsfonts}
\usepackage{psfrag}
\usepackage{bbm} 
\usepackage{float} 
\usepackage{enumerate}
\usepackage{algorithm} 
\usepackage{algpseudocode}
\usepackage{listings}
\usepackage[title]{appendix}
\graphicspath{ {./images/} }

\newtheorem{theorem}{Theorem}[section]
\newtheorem{corollary}[theorem]{Corollary}
\newtheorem{lemma}[theorem]{Lemma}
\theoremstyle{definition}
\newtheorem{definition}[theorem]{Definition}
\theoremstyle{definition}
\newtheorem{assumption}[theorem]{Assumption}
\newtheorem{remark}[theorem]{Remark}
\newtheorem{example}[theorem]{Example}

\DeclareMathOperator*{\argmin}{arg\,min}

\newcommand\Tstrut{\rule{0pt}{2.6ex}}         
\newcommand\Bstrut{\rule[-1.3ex]{0pt}{2.6ex}}   

\title{An SMP-Based Algorithm for Solving the Constrained Utility Maximization Problem via Deep Learning}

\author{
 Kristof Wiedermann \vspace{0.2cm}\\
  Research Unit of Financial and Actuarial Mathematics\\
  Vienna University of Technology (TU Wien)\\
  Vienna, Austria \vspace{0.2cm}\\
  \texttt{kristof.wiedermann@tuwien.ac.at} \\
}

\begin{document}
\maketitle
\vspace{0.2cm}
\begin{abstract}
We consider the utility maximization problem under convex constraints with regard to theoretical results which allow the formulation of algorithmic solvers which make use of deep learning techniques.
In particular for the case of random coefficients, we prove a stochastic maximum principle (SMP), which also holds for utility functions $U$ with $\mathrm{id}_{\mathbb{R}^{+}} \cdot U'$ being not necessarily nonincreasing, like the power utility functions, thereby generalizing the SMP proved by Li and Zheng (2018). We use this SMP together with the strong duality property for defining a new algorithm, which we call deep primal SMP algorithm. Numerical examples illustrate the effectiveness of the proposed algorithm -- in particular for higher-dimensional problems and problems with random coefficients, which are either path dependent or satisfy their own SDEs. Moreover, our numerical experiments for constrained problems show that the novel deep primal SMP algorithm overcomes the deep SMP algorithm's (see Davey and Zheng (2021)) weakness of erroneously producing the value of the corresponding unconstrained problem. Furthermore, in contrast to the deep controlled 2BSDE algorithm from Davey and Zheng (2021), this algorithm is also applicable to problems with path dependent coefficients. As the deep primal SMP algorithm even yields the most accurate results in many of our studied problems, we can highly recommend its usage. Moreover, we propose a learning procedure based on epochs which improved the results of our algorithm even further. Implementing a semi-recurrent network architecture for the control process turned out to be also a valuable advancement.
\end{abstract}
\vspace{0.2cm}


\section{Introduction}
\label{section:introduction}
In mathematical finance, as well as in practical applications, it is of greatest interest to utilize the given freedom in a model, e.g. the (potentially constrained) choice of an investment strategy, in order to optimize the expected result. In the present paper we consider the utility maximization problem, where the quality of the terminal result, i.e. the portfolio value at terminal time, is measured by means of a strictly concave utility function $U$. Clearly, solving this infinite-dimensional problem is extremely difficult as one intends to maximize a functional over a space of progressively measurable processes. This is aggravated by the fact that, in practice, one usually has certain limitations in choosing the strategy. For example, selling stocks short is either limited or prohibited. We take this into consideration by requiring that the strategies have to take values in a closed, convex set $K$. However, the convexity of $K$ and the strict concavity of $U$ allow us to formulate a closely related problem by means of convex duality methods, the so-called dual problem (see \cite{labbe,zheng2}). This problem is potentially easier to solve. For example, the dual control is forced to be the zero process, if the primal problem is unconstrained. In the context of the formulation of our novel algorithm, the dual problem will mainly be used for deriving a high quality upper estimate for the true value of the primal problem. \vspace{0.3cm}\newline
There are two popular theoretical concepts which can be used for solving the aforementioned control problems: The dynamic programming approach (DP) in a Markovian setting and the stochastic maximum principle~(SMP), which is also applicable to problems with random coefficients. The first approach is centered around a nonlinear PDE, the so-called Hamilton-Jacobi-Bellman (HJB) equation, which is derived by varying the initial condition of the control problem, for which the Markovian nature of the problem is essential. In some cases, the value function can then be obtained by solving the HJB equation. The second concept is based on a necessary condition for an optimal control. It states that if a control is optimal, then it necessarily has to maximize a Hamiltonian-related function almost everywhere in the control argument, where the other arguments are given by the solution to an FBSDE system. Since we intend that our algorithm is also capable of solving non-Markovian problems, we choose the second path. The main idea, which then leads to our algorithmic solver, is to model all processes which only appear in the integrands of the FBSDE system associated with the stochastic maximum principle by means of neural networks for each point in a time discretization. This concept was first studied in \cite{jentzen2bsde} and \cite{jentzenbsde}, where the authors aim at solving FBSDE systems via deep learning. In \cite{zheng1}, this method was extended to controlled FBSDE systems which appear in the study of stochastic control problems. Hence, also the control process has to be modeled by means of neural networks. Since a neural network can be described by means of a finite-dimensional parameter vector, the algorithm, therefore, is required to solve a problem which is only finite-dimensional. The choice of neural networks is insofar appealing, as they have desirable approximation properties. Moreover, this puts the problem into the context of deep learning.\vspace{0.3cm}\newline
The paper is structured as follows. In Section \ref{section:UMaxandDual}, we introduce the constrained utility maximization problem and motivate the derivation of its dual problem. In Section \ref{section:chapter4}, we prove an SMP for the primal problem and cite the corresponding theorem for the dual problem. These results are then used for defining the novel deep primal SMP algorithm. Moreover, we provide numerical examples which contrast the performance of our novel algorithm with both algorithms defined in \cite{zheng1} for several specific utility maximization problems (see Section \ref{section:NumericalExperiments}). This includes high-dimensional problems and problems with random coefficients which are either path dependent or satisfy their own SDEs. Furthermore, Section \ref{section:refining} is devoted to a discussion of two potential refinements of the deep primal SMP algorithm. Finally, we review our results in Section \ref{section:conclusion}.\vspace{0.3cm}\newline
\textbf{Note.} This paper aims at summarizing and refining the main findings of the author's master's thesis, which can be accessed via the link provided in \cite{KW2021deep}. Hence, Sections 2-4 correspond to a revised and condensed version of the respective passages from \cite{KW2021deep}.

\section{The Utility Maximization Problem and Its Dual Problem}
\label{section:UMaxandDual}

The aim of the first section is to present the abstract market model we are going to use in the following. In contrast to the Markovian setting in Chapter 3 of \cite{pham}, we choose a rather general formulation here in order to allow the price processes to be non-Markovian. Moreover, we present the utility maximization problem and the associated dual problem while working with a definition of a utility function $U$ which is slightly stricter than the classical one. This restriction implies that the Legendre-Fenchel transform of $U$ satisfies several important properties (see Subsection \ref{subsection:dualproblem} below).

\subsection{The Underlying Market Model}
\label{subsection:marketmodel}

Let us fix a finite time horizon $T \in \mathbb{R}^{+}$. We consider a standard $m$-dimensional, $m \in \mathbb{N}^{+}$, Brownian motion~$B$ on a filtered probability space $(\Omega,\mathcal{F},\mathbb{F}\hspace{-0.1cm}=\hspace{-0.1cm}(\mathcal{F}_{t})_{t \in [0,T]},\mathbb{P})$, where we choose the filtration $\mathbb{F}$ as the natural filtration generated by $B$, completed with all subsets of null sets of $(\Omega,\mathcal{F},\mathbb{P})$. By construction, we obtain a complete filtration. Hence, as Brownian motion has independent future increments and continuous paths, we conclude the right-continuity of our filtration from a version of Blumenthal's zero-one law, which states that, under these assumptions, suitable sets from $\mathcal{F}_{t}$ and $\mathcal{F}_{t}^{+}$ only differ by a set of measure zero. Therefore, the filtration fulfills the so-called usual conditions. Since we have a Brownian filtration, we are allowed to apply the martingale representation theorem to local $\mathbb{F}$-martingales, which is of greatest importance for the derivation of the SMP in Theorem \ref{theorem:PrimalnonMarkov}.\vspace{0.3cm}\newline
Analogously to \cite{zheng1} and \cite{zheng2}, we consider a market consisting of $m$ stocks and one risk-free bond. Let $r:\Omega \times [0,T] \rightarrow \mathbb{R}$, $\mu:\Omega \times [0,T] \rightarrow \mathbb{R}^{m}$ and $\sigma:\Omega \times [0,T] \rightarrow \mathbb{R}^{m \times m}$ be $\mathbb{F}$-progressively measurable processes. For notational convenience, we are going to omit the argument $\omega$ in the following. Moreover, we have to formulate conditions which ensure the solvability of (B)SDEs that we are going to encounter in the following (e.g. \eqref{eq:stockmodel}, \eqref{eq:portfoliosde}, \eqref{eq:dualsde} and the FBSDE systems arising from both SMPs). On the one hand, we can assume that $r$, $\mu$ and $\sigma$ are uniformly bounded processes and $\sigma$ satisfies the strong non-degeneracy condition:
\begin{displaymath}
\exists k \in \mathbb{R}^{+},\hspace{0.05cm} \forall (y,t) \in \mathbb{R}^{m}\times [0,T]: \hspace{0.3cm} y^{\intercal}\hspace{0.05cm} \sigma(t)\hspace{0.05cm}\sigma^{\intercal}(t)\hspace{0.05cm} y \geq k \hspace{0.05cm}|y|^{2}, 
\end{displaymath}
where we denote the Euclidean norm by $|\cdot|$. According to Section $5.8$ of \cite{shreveBM}, this ensures that $\sigma(t)$ and its transpose are invertible for all $t\in [0,T]$ with the inverse matrices also being uniformly bounded. Alternatively, requiring the continuity of $r$, $\mu$ and $\sigma$ and the existence of $\sigma^{-1}$ is also sufficient for this purpose. Hence, the results of this paper are also applicable to the Vasicek model and Heston's stochastic volatility model, provided that the Feller condition is satisfied. \vspace{0.3cm}\newline
We can now define the dynamics of the risk-free bond $S_{0}$, i.e. the bank account, and the stocks $S_{i}$, $i\in~\hspace{-0.1cm}\{1,\dots,m\}$, via  
\begin{equation} \label{eq:stockmodel}
    \begin{aligned}
    dS_{0}(t) =& \hspace{0.1cm} S_{0}(t)r(t)\hspace{0.05cm}dt, \hspace{0.3cm} t\in [0,T],\hspace{0.3cm} S_{0}(0)=1, \hspace{0.3cm}\mbox{and} \\
    dS_{i}(t) =& \hspace{0.1cm} S_{i}(t) \big(\mu_{i}(t)\hspace{0.05cm}dt + \sigma_{i\cdot}(t)\hspace{0.05cm}dB(t) \big),\hspace{0.3cm} t\in [0,T], \hspace{0.3cm} S_{i}(0)>0.
    \end{aligned}
\end{equation}
Due to our assumptions on $r$, $\mu$ and $\sigma$, the SDEs in \eqref{eq:stockmodel} admit unique strong solutions, namely the corresponding stochastic exponentials.

\subsection{The Utility Maximization Problem in Portfolio Optimization}
\label{subsection:utilitymax}

Suppose that a market model as discussed above is given. At first, we aim at constructing a portfolio for a small investor with initial capital $x_{0}\in\mathbb{R}^{+}$. The notion of a ``small'' investor is insofar important as we can assume as a consequence that trades which are executed by our investor do not affect the stock prices. \vspace{0.3cm}\newline
Consider a progressively measurable process $\pi:\Omega \times [0,T] \rightarrow \mathbb{R}^{m}$. For notational convenience, we call the set of these processes $\mathcal{A}_{prog}$. We denote by $\pi_{i}(t)$, $i\in\{1,\dots,m\}$, the portion of wealth invested into the $i$th stock at time $t$. Since we focus on the utility maximization problem under constraints, we restrict ourselves to processes whose images lie in a closed, convex subset $K\subseteq\mathbb{R}^{m}$. The set $K$ is a priori given and we suppose that it contains the zero vector, which is equivalent to the admissibility of only holding the risk-free bond. As we shall see later, when checking the solvability of the SDE \eqref{eq:portfoliosde}, it is sufficient to require that $\pi$ is square-integrable with respect to the product measure $\mathbb{P}\otimes\lambda|_{[0,T]}$. Therefore, we choose to define the set of all admissible portfolio strategies like in \cite{zheng2} as
\begin{equation}\label{eq:admissiblestrategies}
    \mathcal{A} := \bigg\{ \pi\in\mathcal{A}_{prog} \hspace{0.2cm} \bigg| \hspace{0.2cm} \pi(t)\in K \hspace{0.2cm}\mbox{a.s. for a.e.} \hspace{0.2cm} t\in[0,T], \hspace{0.2cm}\mathbb{E}\bigg[ \int_{0}^{T} |\pi(t)|^2\hspace{0.05cm}dt \bigg]<\infty \bigg\}. 
\end{equation}
Clearly, requiring only almost sure finiteness for the integral in \eqref{eq:admissiblestrategies} is also sufficient with regards to the solvability of the SDEs below. However, considering only processes $\pi\in H^{2}(0,T;\mathbb{R}^{m})$ is quite common in control theory since in the proof of the dynamic programming principle for Markovian problems one wants to equip the space of admissible control processes with a metric. By requiring the portfolio to be self-financing, we can now define the dynamics of the associated wealth process $X^{\pi}$ for any given portfolio process $\pi\in\mathcal{A}$ as
\begin{displaymath}
    dX^{\pi}(t)= \sum_{i=1}^{m}\frac{X^{\pi}(t)\pi_{i}(t)}{S_{i}(t)}\hspace{0.05cm} dS_{i}(t) + \frac{X^{\pi}(t)\big(1-\sum_{j=1}^{m}\pi_{j}(t)\big)}{S_{0}(t)}\hspace{0.05cm} dS_{0}(t), \hspace{0.3cm} t\in [0,T].
\end{displaymath}
By plugging in according to \eqref{eq:stockmodel} and simplifying we obtain for every $t\in [0,T]$:
\begin{equation}\label{eq:portfoliosde}
  dX^{\pi}(t)=X^{\pi}(t)\big[ \big(r(t)+\pi^{\intercal}(t)\sigma(t)\theta(t) \big)\hspace{0.05cm}dt + \pi^{\intercal}(t)\sigma(t) \hspace{0.05cm}dB(t)\big], 
\end{equation}
with initial condition $X^{\pi}(0)=x_{0}$. Here, we defined $\theta(t)$ as the market price of risk
\begin{displaymath}
\theta(t):=\sigma^{-1}(t)\big(\mu(t)- (r(t),\dots,r(t))^{\intercal} \big).
\end{displaymath}
Of course, this process is also progressively measurable and either uniformly bounded or continuous thanks to the respective conditions on $r$, $\mu$ and $\sigma$. Due to the choice of $\mathcal{A}$ and the aforementioned properties, we obtain that there exists a unique strong solution $X^{\pi}$ to this SDE, namely the corresponding stochastic exponential. Note that the structure of \eqref{eq:portfoliosde} implies that the wealth process remains positive after starting in $x_{0}\in\mathbb{R}^{+}$. Hence, negative wealth levels, i.e. the investor's ruin, are a priori excluded, which simplifies our considerations below. \vspace{0.3cm}\newline
In the following, we give a definition of a utility function which requires stronger differentiability properties than the classical definition.
\begin{definition}\label{definition:utilityfunction}
    Let $U\in C^{2}(\mathbb{R}^{+})$ be a real-valued function which is strictly increasing and strictly concave. We then call $U$ a utility function. Furthermore, if the identities
    \begin{displaymath}
        \lim_{x\searrow 0} U'(x) = \infty \hspace{0.3cm}\mbox{and}\hspace{0.3cm} \lim_{x\nearrow \infty} U'(x) = 0
    \end{displaymath}
    hold, we say that $U$ satisfies the Inada conditions.
\end{definition}
\noindent Due to the strict concavity of $U$, we obtain from the inverse function theorem that $I:=(U')^{-1}$ is continuously differentiable, as well. This is an important observation with regard to Lemma \ref{lemma:transformierteallg} below. Furthermore, the Inada conditions guarantee that the domain of $I$ is $\mathbb{R}^{+}$. \vspace{0.3cm}\newline
The following additional assumption (cf. \cite{shreveMM}) proves to be helpful for deriving an SMP for the dual problem, which will be introduced in Subsection \ref{subsection:dualproblem}.
\begin{assumption} \label{assumption:assumpU}
  Let $U$ be a utility function according to Definition \ref{definition:utilityfunction} which satisfies the Inada conditions. We require that there exist constants $\beta\in (0,1)$ and $\gamma\in (1,\infty)$ such that 
  \begin{equation}\label{eq:assumpU}
\forall x \in \mathbb{R}^{+}: \hspace{0.3cm} \beta U'(x)\ge U'(\gamma x),
  \end{equation}
  and $\mathrm{id}_{\mathbb{R}^{+}} \cdot U'$ is nondecreasing, where we denote the identity function of $\mathbb{R}^{+}$ by $\mathrm{id}_{\mathbb{R}^{+}}$.
\end{assumption}
\noindent An application of the product rule shows that the Arrow-Pratt index of relative risk aversion of $U$ is, under the premise of the previous statement, bounded above by $1$. 
\begin{example}\label{example:utility}
    Consider $U_{1} :=\log$ and $U_{2,p} :=p^{-1} \cdot (\mathrm{id}_{\mathbb{R}^{+}})^{p}$ as two prominent examples of utility functions, namely the log utility function and the power utility function with parameter $p\in (0,1)$. Clearly, these functions satisfy Definition \ref{definition:utilityfunction}. Moreover, it is obvious that in both cases the Inada conditions are satisfied. Furthermore, we record the fact that $\mathbbm{1}_{\mathbb{R}^{+}}$ and $(\mathrm{id}_{\mathbb{R}^{+}})^{p}$, $p\in (0,1)$, are nondecreasing functions. Due to the simple form of $U_{1}'$, we can even choose for any $\beta\in (0,1)$ a constant $\gamma$ with the desired properties, namely $\beta^{-1}$. In the case of power utility, we obtain by $\beta^{q}$, where $q:=(p-1)^{-1} \in (-\infty,-1)$, likewise a constant~$\gamma$ for any $\beta \in (0,1)$, i.e. an even stronger property than required in \eqref{eq:assumpU}. Hence, the functions $U_{1}$ and $U_{2,p}$, $p\in (0,1)$, satisfy Assumption \ref{assumption:assumpU}. 
\end{example}
\noindent We are now in position to formulate the constrained utility maximization problem.
\begin{definition} \label{definition:Umaxproblem}
    Let $U$ be a utility function and $X^{\pi}$ the solution to the SDE \eqref{eq:portfoliosde} for $\pi \in \mathcal{A}$, where $\mathcal{A}$ is defined in \eqref{eq:admissiblestrategies}. We define the gain function, which maps every strategy $\pi \in \mathcal{A}$ to the expected utility of the portfolio value at time $T$, by
    \begin{displaymath}
    J(\pi):= \mathbb{E} [U(X^{\pi}(T))].
    \end{displaymath}
    Maximizing the expected utility corresponds to finding
    \begin{equation}\label{eq:Umaxproblem}
        V:= \sup_{\pi\in\mathcal{A}} J(\pi).
    \end{equation}
    A control $\pi^{*}\in\mathcal{A}$ is called optimal, if it attains the supremum in \eqref{eq:Umaxproblem}, i.e. $V=J(\pi^{*})$.
\end{definition}
\noindent Quite naturally, we are interested in the case where $V$ is a real number. Initially, one has to ensure that the expectation in the definition of $J(\pi)$ is well-defined. For example, we are going to require that even $U(X^{\pi}(T)) \in L^{1}$ holds for all admissible strategies in the setting of the SMP for the utility maximization problem (cf. Assumption \ref{assumption:primalintegrability} below).

\subsection{The Legendre-Fenchel Transform Leading to the Dual Problem}
\label{subsection:dualproblem}

The following considerations are similar to \cite{shrevemartdual91} and \cite{shreveMM} as we are interested in properties of the Legendre-Fenchel transform under the special circumstances of Definition \ref{definition:utilityfunction}. We refer to \cite{rockafellar} for a general theory under milder assumptions. 
\begin{definition}\label{definition:legendrefenchel}
    Let $U$ be a utility function as defined in Definition \ref{definition:utilityfunction}. Then the Legendre-Fenchel transform, $\widetilde{U}:\mathbb{R}^{+}\rightarrow\mathbb{R}$, is defined by 
    \begin{equation}\label{eq:legendrefenchel}
        \widetilde{U}(y):= \sup_{x\in\mathbb{R}^{+}} \big\{ U(x)-xy\big\}, \hspace{0.3cm} y\in\mathbb{R}^{+}.
    \end{equation}
\end{definition}
\noindent Hence, $\widetilde{U}$ maps every $y\in\mathbb{R}^{+}$ to the maximum signed distance between $U$ and a linear function starting in zero whose first derivative is equal to $y$. In the literature, the Legendre-Fenchel transform is usually defined for convex functions $f$ first via $\widetilde{f}(y):=\sup_{x\in D} \{ xy-f(x)\}$, where $D$ denotes the domain of $f$. However, applying this definition to the strictly convex, strictly increasing function $-U(-x)$, $x\in\mathbb{R}^{-}$, leads precisely to \eqref{eq:legendrefenchel}. The next lemma shows some general properties of $\widetilde{U}$.
\begin{lemma}\label{lemma:transformierteallg}
    Let $U$ be a utility function satisfying the Inada conditions according to Definition \ref{definition:utilityfunction} and $\widetilde{U}$ the corresponding Legendre-Fenchel transform. Then the following properties hold:
    \begin{enumerate}[(i)]
    \item $\widetilde{U}(y)=U(I(y))-yI(y)$, \hspace{0.3cm} $y\in\mathbb{R}^{+}$,
    \item $\widetilde{U} \in C^{2}(\mathbb{R}^{+})$, strictly decreasing, strictly convex,
    \item $U(x)=\widetilde{U}(U'(x))+xU'(x)=\inf_{y\in\mathbb{R}^{+}} \big\{ \widetilde{U}(y)+xy\big\}$,\hspace{0.3cm} $x\in\mathbb{R}^{+}$,
    \item $\widetilde{U}'(y)=-I(y)$, \hspace{0.3cm} $y\in\mathbb{R}^{+}$,
    \item $\widetilde{U}(0):= \lim_{y\searrow 0}\widetilde{U}(y)=\lim_{x\nearrow \infty} U(x) =: U(\infty)$ \hspace{0.2cm} and
    \item $\widetilde{U}(\infty):= \lim_{y\nearrow\infty}\widetilde{U}(y)=\lim_{x\searrow 0} U(x) =: U(0)$.
    \end{enumerate}
\end{lemma}
\begin{proof}
     We refer to \cite[Section 4]{shrevemartdual91} and \cite[Section 3.4]{shreveMM} for proofs of the above statements.
\end{proof}
\noindent In the following lemma, we summarize the implications of Assumption \ref{assumption:assumpU} for $\widetilde{U}$ and $I=-\widetilde{U}'$, which is important for Theorem \ref{theorem:DualnonMarkov}.
\begin{lemma}\label{lemma:transformierteassumption}
    In addition to the assumptions in Lemma \ref{lemma:transformierteallg}, we require that $U$ satisfies Assumption \ref{assumption:assumpU}. Then we additionally obtain
    \begin{enumerate}[(i)]
        \item $x\mapsto xI(x)$ is nonincreasing on $\mathbb{R}^{+}$,
        \item $x\mapsto \widetilde{U}(\exp(x))$ is convex on $\mathbb{R}$ and
        \item $\exists \beta\in (0,1), \gamma\in (1,\infty):$ \hspace{0.1cm} \big($\forall x \in\mathbb{R}^{+}:$\hspace{0.2cm} $I(\beta x)\le \gamma I(x)$\big).
    \end{enumerate}
\end{lemma}
\begin{proof}
    We refer to \cite[Section 4]{shrevemartdual91} and \cite[Section 3.4]{shreveMM} for proofs of the above statements.
\end{proof}
\noindent By multiplying in Lemma \ref{lemma:transformierteassumption} \textit{(i)} and \textit{(iii)} with $-1$ and Lemma \ref{lemma:transformierteallg} \textit{(iv)}, we obtain the reverse results for $\widetilde{U}'$. Finally, we briefly consider the utility functions discussed in Example \ref{example:utility} in the light of Lemma \ref{lemma:transformierteassumption}.
\begin{example}\label{example:utilitydual}
    Let $U_{1}$ and $U_{2,p}$, $p\in (0,1)$, be given as in Example \ref{example:utility}. One can easily see that
    \begin{displaymath}
        I_{1}(x)=\frac{1}{x} \hspace{0.3cm}\mbox{and}\hspace{0.3cm} I_{2,p}(x)=x^{\frac{1}{p-1}}, \hspace{0.3cm} x\in\mathbb{R}^{+},
    \end{displaymath}
    holds. Hence, by Lemma \ref{lemma:transformierteallg} \textit{(i)} we can write $\widetilde{U}_{1}$ and $\widetilde{U}_{2,p}$ explicitly as
    \begin{displaymath}
        \widetilde{U}_{1}(y)=-\log(y)-1 \hspace{0.3cm}\mbox{and}\hspace{0.3cm} \widetilde{U}_{2,p}(y)=\frac{1-p}{p} y^{\frac{p}{p-1}}, \hspace{0.3cm} y\in\mathbb{R}^{+}.
    \end{displaymath}
    We know from Example \ref{example:utility} that these utility functions satisfy Assumption \ref{assumption:assumpU}. Therefore, Lemma \ref{lemma:transformierteassumption} is applicable. Explicit calculation shows that the same constants $\beta$, $\gamma$ as in Example \ref{example:utility} work for the statement of Lemma \ref{lemma:transformierteassumption} \textit{(iii)}, as indicated by its proof.
\end{example}

Solving the constrained utility maximization problem explicitly can become very difficult in many instances. Therefore, formulating an accompanying problem, that can potentially be solved with less computational effort, would be highly favorable. We aim at obtaining an upper bound for the value of the original problem while using the Legendre-Fenchel transform $\widetilde{U}$ in the formulation of this new problem. Similarly to \cite{schachermayer}, we are interested in positive semimartingales $Y$ such that $X^{\pi} Y$ is a supermartingale for every $\pi \in \mathcal{A}$. By the definition of $\widetilde{U}(y)$, it is for every $x \in \mathbb{R}^{+}$ an upper bound for $U(x)-xy$. The supermartingale property then enables us to derive an upper bound which is independent of $\pi$. In the following, we are going to summarize the approach presented in \cite{zheng2}. \vspace{0.3cm}\newline
In accordance with the aforementioned idea, we choose the following ansatz to describe the dynamics of the desired one-dimensional process $Y$:
\begin{equation}\label{eq:ydualansatz}
dY(t)=Y(t) \big[ a(t)\hspace{0.05cm} dt + b^{\intercal}(t) \hspace{0.05cm} dB(t) \big], \hspace{0.3cm}  t\in [0,T],
\end{equation}
with initial condition $Y(0)=y$, where $y\in \mathbb{R}^{+}$. Here, the progressively measurable processes $a:\Omega~\hspace{-0.095cm}\times~\hspace{-0.095cm}[0,T] \rightarrow \mathbb{R}$ and $b:\Omega \times [0,T] \rightarrow \mathbb{R}^{m}$ have to be chosen such that $Y$ is uniquely given by the corresponding stochastic exponential and $X^{\pi} Y$ is a supermartingale for every $\pi \in \mathcal{A}$. The latter can be achieved by applying the integration by parts formula for continuous semimartingales to $X^{\pi} Y$ and ensuring that the resulting drift term is nonpositive. This condition leads to an upper bound for $a$ containing the support function $\delta_{K}$, which turns out to be the optimal choice for $a$ with regard to our intention of finding a minimal upper bound.
\begin{remark}\label{remark:supportfct}
We remember that the support function of $-K$, where $K$ is a closed, convex set, is defined via
\begin{equation}\label{eq:supportfunction}
    \delta_{K}(z)= \sup_{\pi\in K} \{ -\pi^{\intercal}z\}, \hspace{0.3cm} z\in \mathbb{R}^{m}.
\end{equation}
We observe that $\delta_{K}$ is positive homogeneous and subadditive since the respective properties also hold for the supremum. Therefore, it is a convex function, which is an important observation with regard to Theorem~\ref{theorem:DualnonMarkov} below. Furthermore, $\delta_{K}$ is nonnegative because we assumed $0\in K$. Hence, it is finite on the set $\widetilde{K}:=\{z\in\mathbb{R}^{m} \hspace{0.1cm}|\hspace{0.1cm} \delta_{K}(z)<\infty\}$.
\end{remark}
\noindent The processes arising from the ansatz given in \eqref{eq:ydualansatz} which satisfy the desired properties and whose drift is optimal with regard to finding a small upper bound are, therefore, precisely given by
\begin{equation}\label{eq:dualsde}
    dY^{(y,v)}(t)= -Y^{(y,v)}(t) \big[ \big(r(t)+\delta_{K}(v(t)) \big) \hspace{0.05cm} dt + \big( \theta(t)+\sigma^{-1}(t)v(t)\big)^{\intercal}\hspace{0.05cm} dB(t)\big], \hspace{0.3cm}  t\in [0,T],
\end{equation}
with initial condition $Y^{(y,v)}(0)=y$ (see \cite{zheng2} for details). Note that we hereby obtain a family of SDEs which can be described by means of the inital condition $y$ and a progressively measurable process $v$, as indicated by the notation. The process $v$, which will be called dual control process from now on, has to be chosen such that the above SDE admits a unique strong solution. This is clearly satisfied, if we introduce the set of all admissible pairs of initial values and dual control processes, similarly to \eqref{eq:admissiblestrategies}, via
\begin{equation}\label{eq:dualadmissible}
    \mathcal{D}:= \bigg\{ (y,v)\in \mathbb{R}^{+}\times\mathcal{A}_{prog} \hspace{0.2cm} \bigg| \hspace{0.2cm} \mathbb{E}\bigg[ \int_{0}^{T} \big[ |v(t)|^2+\delta_{K}(v(t))\big] \hspace{0.05cm}dt \bigg]<\infty \bigg\}.
\end{equation}
By the definition of $\widetilde{U}$ and the supermartingale property of $X^{\pi} Y^{(y,v)}$, we obtain
\begin{equation}\label{eq:dualisupperbound1}
    \mathbb{E}\big[ U(X^{\pi}(T))\big] \le \mathbb{E}\big[\widetilde{U}(Y^{(y,v)}(T)) \big] + \mathbb{E}\big[ X^{\pi}(T) Y^{(y,v)}(T) \big] \le \mathbb{E}\big[ \widetilde{U}(Y^{(y,v)}(T))\big] + x_{0}y.
\end{equation}
Since this is true for any $\pi\in\mathcal{A}$ and $(y,v)\in \mathcal{D}$, we can take the supremum on the left-hand side and the infimum on the right-hand side. This leads to
\begin{equation}\label{eq:dualisupperbound2}
    V= \sup_{\pi\in\mathcal{A}}\mathbb{E}\big[ U(X^{\pi}(T))\big] \le \inf_{(y,v)\in\mathcal{D}} \big\{\mathbb{E}\big[\widetilde{U}(Y^{(y,v)}(T))\big] + x_{0}y \big\},
\end{equation}
where $V$ is exactly the same as in \eqref{eq:Umaxproblem}. The right-hand side corresponds to the value of the dual problem:
\begin{definition}\label{definition:dualproblem}
    Let a constrained utility maximization problem according to Definition \ref{definition:Umaxproblem} be given. The value of the associated dual problem is defined by
    \begin{equation}\label{eq:dualproblem}
        \widetilde{V} := \inf_{(y,v)\in\mathcal{D}} \big\{\mathbb{E}\big[\widetilde{U}(Y^{(y,v)}(T))\big] + x_{0}y \big\},
    \end{equation}
    where $\mathcal{D}$ is given by \eqref{eq:dualadmissible} and $Y^{(y,v)}$ denotes the unique solution to the SDE \eqref{eq:dualsde} for a fixed tuple $(y,v)\in \mathcal{D}$. A pair $(y^{*},v^{*})\in \mathcal{D}$ is called optimal, if it attains the infimum on the right-hand side of \eqref{eq:dualproblem}.
\end{definition}
\noindent Note that we also have to optimize with respect to the initial condition here, whereas $x_{0}$ is a priori given in the original problem.

\section{A Stochastic Maximum Principle for Our General Non-Markovian Setting Paving the Way for the Novel Deep Primal SMP Algorithm}
\label{section:chapter4}
We consider the utility maximization problem and its dual problem, as introduced in Section \ref{section:UMaxandDual}, in its full generality here. Hence, the processes $r$, $\mu$ and $\sigma$ are allowed to be even path dependent. This implies that the dynamic programming approach, for which a controlled Markovian structure of the problem is essential, is not necessarily applicable. However, considering whether we can prove stochastic maximum principles for both problems might be worthwhile as there are several such results in the literature which also work for non-Markovian problems. For example, \cite{SMPkaratzas} finds an SMP for problems with random coefficients, where the drift and the diffusion coefficient of the state process $X^{\pi}$ are affine in $\pi$ and $X^{\pi}$. Unfortunately, this does not hold for our problem, as \eqref{eq:portfoliosde} shows. However, \cite{zheng2} shows that it is still possible to prove stochastic maximum principles for the utility maximization problem and its dual problem, respectively. Unfortunately, it turns out that, in contrast to the statement of \cite[Theorem 3.5]{zheng2}, the proof concept therein is not applicable to certain utility functions including the power utility functions (cf. Remark \ref{remark:primalZhengassumptionswrong}). We will show in Subsection \ref{subsection:nonmarkovprimalSMP} that proving an SMP for these utility functions is also possible, if one introduces alternative assumptions which differ from Assumption \ref{assumption:assumpU}.\vspace{0.3cm}\newline
As the generalized Hamiltonian is an essential ingredient with regards to the formulation of SMPs, we start this section by writing it down explicitly for both problems. The generalized Hamiltonian associated with the primal problem is for every $(t,x,\pi,y,z)\in [0,T]\times\mathbb{R}^{+}\times K \times\mathbb{R}\times \mathbb{R}^{m}$ given by
\begin{equation}\label{eq:generalizedHamutility}
    \mathcal{H}_{1}(t,x,\pi,y,z):= x\hspace{0.05cm}(r(t)+\pi^{\intercal}\sigma(t)\theta(t))\hspace{0.05cm} y + x\hspace{0.05cm}\pi^{\intercal}\sigma(t)\hspace{0.05cm} z.
\end{equation}
In the case of the dual problem, it is for every $(t,x,v,y,z)\in [0,T]\times\mathbb{R}^{+}\times \widetilde{K} \times\mathbb{R}\times \mathbb{R}^{m}$ given by
\begin{equation}\label{eq:generalizedHamdual}
    \mathcal{H}_{2}(t,x,v,y,z):= -x\hspace{0.05cm}(r(t)+\delta_{K}(v))\hspace{0.05cm} y  -x\hspace{0.05cm} \big(\theta(t)+ \sigma^{-1}(t)\hspace{0.05cm}v\big)^{\intercal}\hspace{0.05cm} z.
\end{equation}
Hence, it follows that the corresponding first-order adjoint equations are given by
\begin{equation}\label{eq:primaladjointeq}
    dp_{1}(t)= -\big[\big(r(t) + \pi^{\intercal}(t)\hspace{0.05cm}\sigma(t)\hspace{0.05cm}\theta(t) \big)\hspace{0.05cm} p_{1}(t) + \pi^{\intercal}(t)\hspace{0.05cm} \sigma(t)\hspace{0.05cm}q_{1}(t) \big]\hspace{0.05cm}dt + q_{1}^{\intercal}(t)\hspace{0.05cm}dB(t),\hspace{0.3cm} t\in [0,T],
\end{equation}
with associated terminal condition $p_{1}(T)= -U'\big(X^{\pi}(T) \big)$ for the utility maximization problem and
\begin{equation}\label{eq:dualadjointeq}
    dp_{2}(t) = \big[\big(r(t)+\delta_{K}(v(t))\big)\hspace{0.05cm} p_{2}(t)  + \big(\theta(t)+ \sigma^{-1}(t)\hspace{0.05cm}v(t)\big)^{\intercal}\hspace{0.05cm} q_{2}(t) \hspace{0.05cm}\big] dt + q_{2}^{\intercal}(t)\hspace{0.05cm} dB(t),\hspace{0.3cm} t\in [0,T],
\end{equation}
with terminal condition $p_{2}(T)= -\widetilde{U}'\big(Y^{(y,v)}(T) \big)$ for the accompanying dual problem. Note that we are going to require from a potential solution $(Y,Z)$ to a BSDE only that the integral processes occurring in the corresponding stochastic integral equation are well-defined, which is very common in this field (see also \cite{jentzen2bsde,jentzenbsde,horst,trivellato}, for example).\vspace{0.3cm}\newline
Finally, we want to motivate how a simplified version of the classical SMP for Markovian problems (see \cite[Section 3.3]{yongzhou}) can serve as the formal motivation for both SMPs to be discussed below. Note that this result also includes a second-order adjoint equation and the maximization condition concerns the concavity-adjusted generalized Hamiltonian. However, it is argued that the purpose of the adjustment of $\mathcal{H}$ in the maximum condition is ensuring that the adjusted function is concave in the control argument. Moreover, it is pointed out that if $\mathcal{H}$ is already concave in the control argument, then the second-order adjoint equation is superfluous. Clearly, this is the case for the utility maximization problem and its dual problem as \eqref{eq:generalizedHamutility}, \eqref{eq:generalizedHamdual} and the concavity of $-\delta_{K}$ show. The maximization condition will, therefore, correspond to maximizing $\mathcal{H}$ with respect to the control space. Hence, it seems worthwhile to study the feasibility of proving SMPs which are conceptually similar to the one in \cite[Section 3.3]{yongzhou} with $P\equiv 0$ while also allowing random coefficients. Furthermore, we will see that also the reverse implication of each SMP holds. We refer to Remarks \ref{remark:connectionclassSmpprimal} and \ref{remark:connectionclassSmpdual} for a discussion, why Theorems \ref{theorem:PrimalnonMarkov} and \ref{theorem:DualnonMarkov} correspond indeed to stochastic maximum principles.

\subsection{A Stochastic Maximum Principle for the Utility Maximization Problem}
\label{subsection:nonmarkovprimalSMP}
The aim of this section is deriving a result for characterizing the optimality of a control $\pi^{*}$ for the primal problem which corresponds conceptually to a stochastic maximum principle. As argued in the previous subsection and in Remark \ref{remark:primalZhengassumptionswrong} below, this essentially means adjusting and generalizing \cite[Theorem 3.5]{zheng2} such that also utility functions $U$ with $\mathrm{id}_{\mathbb{R}^{+}} \cdot U'$ being strictly increasing (like the power utility functions) are a priori not excluded from the SMP. This result will also allow us to formulate an algorithm made-to-measure for solving the primal problem directly (see Subsection \ref{subsection:primalDeepSMPAlgorithm} below). At first, we formulate technical assumptions which we are going to need in the following.
\begin{assumption}\label{assumption:primalintegrability}
    Suppose that we have $U\big( X^{\pi}(T)\big)\in L^{1}$ and $X^{\pi}(T)\hspace{0.05cm}U'\big( X^{\pi}(T)\big)\in L^{2}$ for every admissible strategy $\pi\in\mathcal{A}$, where $X^{\pi}$ denotes the unique solution to \eqref{eq:portfoliosde}.
\end{assumption}
\noindent Note that the second condition is trivially satisfied by $U=\log$. This assumption guarantees that there exists a solution to \eqref{eq:primaladjointeq}  which satisfies the prerequisites to be specified below. Hence, it is guaranteed that Theorem~\ref{theorem:PrimalnonMarkov} does not correspond to a statement on the empty set. We cite the following lemma (cf. \cite[Lemma 3.4]{zheng2}) including its proof as the explicitly constructed solution to the primal adjoint BSDE \eqref{eq:primaladjointeq} will play an important role in the proof of Theorem \ref{theorem:PrimalnonMarkov}.
\begin{lemma}\label{lemma:existencesolutionstoBSDEprimal}
    Consider $\pi\in\mathcal{A}$ and suppose that Assumption \ref{assumption:primalintegrability} is in place. Then there exists a pair $(p_{1},q_{1})$ solving \eqref{eq:primaladjointeq} such that $p_{1}\hspace{0.03cm}X^{\pi}$ is a martingale.
\end{lemma}
\begin{proof}
    We obtain by means of the martingale representation theorem that there exists a continuous version $V$ of the square-integrable martingale
    \begin{displaymath}
        M:=\big(\mathbb{E}\big[ -X^{\pi}(T)\hspace{0.05cm}U'\big( X^{\pi}(T)\big) \big|\mathcal{F}_{t}\big]\big)_{t\in[0,T]},
    \end{displaymath}
     which can be written as $M_{0}+W\bullet B$, where the process $W\in H^{2}(0,T;\mathbb{R}^{m})$ is unique. Clearly, $p_{1}:=V/X^{\pi}$ is well-defined and satisfies the terminal condition associated with \eqref{eq:primaladjointeq}. Hence, an application of It\^{o}'s formula to $p_{1}$ shows
    \begin{equation}\label{eq:existencesolutionstoBSDEprimalproof1}
        \begin{aligned}
            dp_{1}(t)=& \bigg[ -p_{1}(t) \big(r(t) + \pi^{\intercal}(t)\hspace{0.05cm}\sigma(t)\hspace{0.05cm}\theta(t) - \big| \pi^{\intercal}(t)\hspace{0.05cm}\sigma(t)\big|^{2} \big)  
            - \frac{\pi^{\intercal}(t)\hspace{0.05cm}\sigma(t)\hspace{0.05cm}W(t)}{X^{\pi}(t)}\bigg] \hspace{0.05cm}dt \\
            &+ \bigg[ \frac{W^{\intercal}(t)}{X^{\pi}(t)} - p_{1}(t)\hspace{0.05cm} \pi^{\intercal}(t)\hspace{0.05cm} \sigma(t) \bigg]\hspace{0.05cm}dB(t),
        \end{aligned}
    \end{equation}
    $t\in[0,T]$. Therefore, defining a process $q_{1}$ for $t\in [0,T]$ by
    \begin{equation*}
        q_{1}(t):=\frac{W(t)}{X^{\pi}(t)} - p_{1}(t)\hspace{0.05cm}  \sigma^{\intercal}(t)\hspace{0.05cm}\pi(t)
    \end{equation*}
    reduces \eqref{eq:existencesolutionstoBSDEprimalproof1} to \eqref{eq:primaladjointeq}. Hence, $(p_{1},q_{1})$ solves the primal adjoint equation and $p_{1}\hspace{0.03cm}X^{\pi}=V$ is indeed a martingale, which concludes the proof.
\end{proof}
\noindent Furthermore, we are going to need an easy, albeit important, result on the concavity/con\-vexity of the composition of concave and convex functions, if certain monotonicity properties hold. 
\begin{lemma}\label{lemma:compositionconvexconcave}
    Let $A,B\subseteq \mathbb{R}$ be convex sets and consider a convex, nonincreasing function $f:B\rightarrow \mathbb{R}$ and a concave function $g:A\rightarrow \mathbb{R}$, such that $g(A)\subseteq B$ holds. Then $f\circ g$ is a convex function on $A$. Moreover, if $f$ is nondecreasing and concave and $g$ is again concave, then $f\circ g$ is concave.
\end{lemma}
\begin{proof}
    Fix arbitrary numbers $a_{1}, a_{2}\in A$ and $\lambda\in [0,1]$. By using the concavity of $g$ and the monotonicity of~$f$ in $(1)$ and the convexity of $f$ in $(2)$ we obtain
    \begin{displaymath}
        f\big(g\big(\lambda a_{1} + (1-\lambda)a_{2}\big)\big) \stackrel{(1)}{\le}f\big( \lambda g ( a_{1}) + (1-\lambda) g(a_{2})\big) \stackrel{(2)}{\le} \lambda f(g(a_{1})) + (1-\lambda) f(g(a_{2})).
    \end{displaymath}
    The second claim follows as the relations $(1)$ and $(2)$ are precisely reversed in this case.
\end{proof}
\noindent Now we are in position to formulate the main result of this section, namely the SMP for the primal problem. 
\begin{theorem}\label{theorem:PrimalnonMarkov}
    Let a utility maximization problem in the setting of Definition \ref{definition:Umaxproblem} be given and suppose that Assumption \ref{assumption:primalintegrability} is satisfied. Consider an admissible control $\pi^{*}\in\mathcal{A}$. Let $X^{\pi^{*}}$, $p_{1}^{*}$ and $q_{1}^{*}$ denote processes which satisfy $X^{\pi^{*}}(0)=x_{0}$, $p_{1}^{*}(T)= -U'\big( X^{\pi^{*}}(T)\big)$ and for $t\in [0,T]$:
    \begin{equation}\label{eq:SDEsystemprimalnonMarkov}
        \begin{aligned}
        &dX^{\pi^{*}}(t)=X^{\pi^{*}}(t)\big[ \big(r(t)+\pi^{* \hspace{0.045cm}\intercal}(t)\hspace{0.05cm}\sigma(t)\hspace{0.05cm}\theta(t) \big)\hspace{0.05cm}dt + \pi^{* \hspace{0.045cm}\intercal}(t)\hspace{0.05cm}\sigma(t) \hspace{0.05cm}dB(t)\big], \\
        &dp_{1}^{*}(t)= -\big[\big(r(t) + \pi^{* \hspace{0.045cm}\intercal}(t)\hspace{0.05cm}\sigma(t)\hspace{0.05cm}\theta(t) \big)\hspace{0.05cm} p_{1}^{*}(t) + \pi^{* \hspace{0.045cm}\intercal}(t)\hspace{0.05cm} \sigma(t)\hspace{0.05cm}q_{1}^{*}(t) \big]\hspace{0.05cm}dt + q_{1}^{* \hspace{0.045cm}\intercal}(t)\hspace{0.05cm}dB(t),
        \end{aligned}
    \end{equation}
    such that $p_{1}^{*}\hspace{0.03cm}X^{\pi^{*}}$ is even a martingale. Moreover, suppose that at least one of the following statements is true:
    \begin{enumerate}[(i)]
        \item $\mathrm{id}_{\mathbb{R}^{+}} \cdot U'$ is a nonincreasing function on $\mathbb{R}^{+}$.
        \item Define $\Theta := \big\{(\overline{\pi}-\pi^{*})\mathbbm{1}_{C}\hspace{0.05cm} \big | \hspace{0.05cm} \overline{\pi}\in K, C \in \Sigma_{p}^{\mathbb{F},[0,T]}\big\}$, where $\Sigma_{p}^{\mathbb{F},[0,T]}$ denotes the progressive $\sigma$-algebra accompanying our filtered probability space and $\overline{\pi}$ also stands for the constant control process mapping to $\overline{\pi}$, for notational convenience. For every $\theta\in\Theta$, we have the uniform integrability of the family 
        \begin{displaymath}
            \big(\Delta_{\varepsilon}^{\theta}\big)_{\varepsilon\in (0,1)}:=\bigg( \frac{U\big(X^{\pi^{*}+\varepsilon \theta}(T)\big) - U\big( X^{\pi^{*}}(T)\big) }{\varepsilon}\bigg)_{\varepsilon\in (0,1)}.
        \end{displaymath}
        \item For every $\theta\in\Theta$, there exists a random variable $\xi_{\theta}$ with $(\xi_{\theta})^{-}\in L^{1}$ such that $\Delta_{\varepsilon}^{\theta}$ is (a.s.) bounded below by $\xi_{\theta}$ for every $\varepsilon\in (0,1)$.
    \end{enumerate}
    Then $\pi^{*}$ is optimal if and only if it holds almost surely for almost every $t\in [0,T]$:
    \begin{equation}\label{eq:PrimalnonMarkovcondition1}
        \pi^{*\hspace{0.045cm}\intercal}(t)\hspace{0.05cm} \big[-\sigma(t)\hspace{0.05cm}\big(p_{1}^{*}(t)\hspace{0.05cm}\theta(t)+q_{1}^{*}(t)\big) \big] = \sup_{\pi\in K} \Big\{\pi^{\intercal}\hspace{0.05cm} \big[-\sigma(t)\hspace{0.05cm}\big(p_{1}^{*}(t)\hspace{0.05cm}\theta(t)+q_{1}^{*}(t)\big) \big]\Big\}.
    \end{equation}
\end{theorem}
\begin{proof}
    At first, we recall that Lemma \ref{lemma:existencesolutionstoBSDEprimal} guarantees the existence of processes $X^{\pi^{*}}$, $p_{1}^{*}$ and $q_{1}^{*}$ which meet the requirements. For reasons of clarity, we will subdivide the proof into five steps. At first, we show that \eqref{eq:PrimalnonMarkovcondition1} is necessarily satisfied by an optimal control. \vspace{0.2cm}\newline
    \textit{Step 1. Necessary Condition: \eqref{eq:PrimalnonMarkovcondition1} is satisfied: $\Phi_{\pi}(\varepsilon):=U\big(X^{\pi^{*}+ \varepsilon(\pi-\pi^{*})}(T) \big)$, $\varepsilon\in [0,1]$, is right differen-\\ 
    \noindent\hspace*{1.1cm} tiable at $\varepsilon=0$ and the corresponding derivative is for every $\pi\in\mathcal{A}$ a.s. given by the random variable\\
    \noindent\hspace*{1.1cm} $U'\big(X^{\pi^{*}}(T)\big) \hspace{0.05cm} X^{\pi^{*}}(T)\hspace{0.05cm} H_{\pi}(T)$.}\vspace{0.1cm}\newline
    Let $\pi^{*}\in\mathcal{A}$ be an optimal control and consider an arbitrary, but fixed control $\pi\in\mathcal{A}$. Since $K$ is a convex set, it follows that the convex combination $\pi^{*}+ \varepsilon(\pi-\pi^{*})$ is for every $\varepsilon\in [0,1]$ again an admissible control. Hence, $\Phi_{\pi}$ is well-defined for almost every $\omega\in\Omega$. Moreover, due to the structure of \eqref{eq:portfoliosde}, $X^{\pi^{*}+ \varepsilon(\pi-\pi^{*})}(T)$ is explicitly given by
    \begin{equation}\label{eq:primalnonmarkovproof1}
        \begin{aligned}
         X^{\pi^{*}+ \varepsilon(\pi-\pi^{*})}(T)=x_{0}\hspace{0.03cm} \exp\bigg( &\int_{0}^{T}  \big(r(t)+ \big(\pi^{*}(t)+ \varepsilon(\pi(t)-\pi^{*}(t)) \big)^{\intercal}\hspace{0.05cm} \sigma(t)\hspace{0.05cm}\theta(t)\big) \hspace{0.05cm} dt \\
            &- \frac{1}{2}\int_{0}^{T} \big| \big( \pi^{*}(t)+ \varepsilon(\pi(t)-\pi^{*}(t))\big)^{\intercal}\hspace{0.05cm}\sigma(t) \big|^{2} \hspace{0.05cm} dt \\
            &+ \int_{0}^{T}\big( \pi^{*}(t)+ \varepsilon(\pi(t)-\pi^{*}(t))\big)^{\intercal}\hspace{0.05cm}\sigma(t)\hspace{0.05cm} dB(t) \bigg).   
        \end{aligned}
    \end{equation}
    Fix $\delta\in\mathbb{R}^{-}$. We consider a function $g_{\pi}$ on $(\delta,1]$ which is defined by $\log\big( X^{\pi^{*}+ \varepsilon(\pi-\pi^{*})}(T)\big)=:g_{\pi}(\varepsilon)$ for $\varepsilon\in [0,1]$ and for the remaining points of the domain by the logarithm of the right-hand side of \eqref{eq:primalnonmarkovproof1}, which is also meaningful for negative $\varepsilon$. As $g_{\pi}$ is a polynomial in $\varepsilon$, it follows immediately that it is differentiable with respect to $\varepsilon$. Its derivative at $\varepsilon=0$ is given by $H_{\pi}(T)$, where the process $H_{\pi}$ is defined as
    \begin{equation}\label{eq:primalnonmarkovproof2}
        \begin{aligned}
            H_{\pi}:=&\int_{0}^{\cdot} \big[\big( \pi(t)-\pi^{*}(t)\big)^{\intercal}\hspace{0.05cm}\sigma(t)\hspace{0.05cm}\theta(t) - \big( \pi(t)-\pi^{*}(t)\big)^{\intercal}\hspace{0.05cm}\sigma(t)\hspace{0.05cm}\sigma^{\intercal}(t)\hspace{0.05cm}\pi^{*}(t)\big] \hspace{0.05cm}dt \\
            &+ \int_{0}^{\cdot} \big( \pi(t)-\pi^{*}(t)\big)^{\intercal}\hspace{0.05cm}\sigma(t) \hspace{0.05cm}dB(t).
        \end{aligned}
    \end{equation}
    Therefore, we obtain from the chain rule and $\Phi_{\pi}=(U\circ \exp \circ \hspace{0.05cm}g_{\pi})\big|_{[0,1]}$:
    \begin{equation}\label{eq:primalnonmarkovproof3}
        \lim_{\varepsilon\searrow 0}  \Delta_{\varepsilon}\Phi_{\pi}= U'\big(X^{\pi^{*}}(T)\big) \hspace{0.05cm} X^{\pi^{*}}(T)\hspace{0.05cm} H_{\pi}(T), \hspace{0.3cm}\mbox{a.s.},
    \end{equation}
    where $\Delta_{\varepsilon}\Phi_{\pi}$ denotes the difference quotient of $\Phi_{\pi}$ over $[0,\varepsilon]$. It is important to place emphasis on the fact that, in contrast to the situation for the dual problem, where $\log\big( Y^{(y^{*},v^{*}+ \varepsilon(v-v^{*}))}(T) \big)$ contains $\delta_{K}(v^{*}+\varepsilon(v-v^{*}))$ in the integrand, $g_{\pi}$ is already differentiable with respect to $\varepsilon$. Hence, there is no necessity for searching for a converging estimate as in the proof of Theorem \ref{theorem:DualnonMarkov}. \vspace{0.2cm}\newline
    \textit{Step 2. Necessary Condition: \eqref{eq:PrimalnonMarkovcondition1} is satisfied: Localization by appropriate stopping times.}\vspace{0.1cm}\newline
    The aim of this step is finding a sequence of stopping times converging almost surely to $T$ which guarantees on the one hand that $H_{\pi}$ is bounded and on the other hand that a stochastic integral process appearing in \textit{Step~4} is even a true martingale up to each of these stopping times. For every $\pi\in\mathcal{A}$, we define a sequence of stopping times $(\tau_{n}^{\pi})_{n\in\mathbb{N}}$ via
    \begin{equation}\label{eq:primalnonmarkovproof4}
    \begin{aligned}
        \tau_{n}^{\pi}:=& \inf\bigg\{ t\ge 0 \hspace{0.1cm}:\hspace{0.1cm} \bigg| \int_{0}^{t} p_{1}^{*}(s)\hspace{0.05cm}X^{\pi^{*}}(s) \hspace{0.05cm} \big( \pi(s)-\pi^{*}(s)\big)^{\intercal}\hspace{0.05cm}\sigma(s)\hspace{0.05cm} dB(s) \bigg|\ge n  \bigg\}\\
        &\wedge\hspace{0.1cm} \inf\big\{ t\ge 0 \hspace{0.1cm}:\hspace{0.1cm}  |H_{\pi}(t)|\ge n\big\} \hspace{0.1cm}\wedge\hspace{0.1cm} T,
    \end{aligned}
    \end{equation}
    for $n\in\mathbb{N}$. As the processes within the norm $|\cdot |$ are continuous semimartingales under both sets of assumptions on $r$, $\mu$ and $\sigma$ proposed in Subsection \ref{subsection:marketmodel} (i.e. uniform boundedness and strong non-degeneracy as well as continuity and existence of $\sigma^{-1}$), we can conclude that $\tau_{n}^{\pi}\nearrow T$ holds a.s. as $n\rightarrow\infty$.\newline
    Fix $n\in\mathbb{N}$. Clearly, $\pi^{*}+\mathbbm{1}_{[0,\tau_{n}^{\pi}]}\hspace{0.05cm}\varepsilon\hspace{0.05cm} (\pi-\pi^{*})$ defines an admissible control as well, since it corresponds pointwise to convex combinations of elements of $K$. Hence, the arguments presented in \textit{Step 1} are also applicable, if we replace $\pi^{*}+ \varepsilon(\pi-\pi^{*})$ with the stopped control process. Note that the derivative at $\varepsilon=0$ of the logarithm of the adjusted right-hand side of \eqref{eq:primalnonmarkovproof1} is exactly given by $H_{\pi}^{\tau_{n}^{\pi}}(T)$ since the occurring integrands are, in comparison with the original result, precisely multiplied by $\mathbbm{1}_{[0,\tau_{n}^{\pi}]}$. Moreover, we define the function $\Phi_{\pi}^{n}(\varepsilon) := U\big(X^{\pi^{*}+\mathbbm{1}_{[0,\tau_{n}^{\pi}]}\hspace{0.05cm}\varepsilon\hspace{0.05cm} (\pi-\pi^{*})}(T)\big)$, $\varepsilon\in [0,1]$, and
    \begin{equation}\label{eq:primalnonmarkovproof5}
        \Delta_{\varepsilon}\Phi_{\pi}^{n}:= \frac{U\big(X^{\pi^{*}+\mathbbm{1}_{[0,\tau_{n}^{\pi}]}\hspace{0.05cm}\varepsilon\hspace{0.05cm} (\pi-\pi^{*})}(T)\big)- U\big(X^{\pi^{*}}(T)\big)}{\varepsilon},
    \end{equation}
    for each $\varepsilon\in (0,1]$. Hence, it follows analogously to our previous considerations for any $n\in\mathbb{N}$:
    \begin{equation}\label{eq:primalnonmarkovproof6}
        \lim_{\varepsilon\searrow 0}  \Delta_{\varepsilon}\Phi_{\pi}^{n}= U'\big(X^{\pi^{*}}(T)\big) \hspace{0.05cm} X^{\pi^{*}}(T)\hspace{0.05cm} H_{\pi}^{\tau_{n}^{\pi}}(T), \hspace{0.3cm}\mbox{a.s.}
    \end{equation}
    \textit{Step 3. Necessary Condition: \eqref{eq:PrimalnonMarkovcondition1} is satisfied: Proving $\mathbb{E}\big[ U'\big(X^{\pi^{*}}(T)\big) \hspace{0.05cm} X^{\pi^{*}}(T)\hspace{0.05cm} H_{\pi}^{\tau_{n}^{\pi}}(T) \big]\le 0$ under the\\
    \noindent\hspace*{1.1cm} premise that at least one of the Conditions (i), (ii) and (iii) holds.}\vspace{0.1cm}\newline
    \textit{(i)}: At first, we notice that $g_{\pi}^{n}(\varepsilon):=\log \big( X^{\pi^{*}+\mathbbm{1}_{[0,\tau_{n}^{\pi}]}\hspace{0.05cm}\varepsilon\hspace{0.05cm} (\pi-\pi^{*})}(T)\big)$, $\varepsilon\in [0,1]$, is a concave function for almost every $\omega\in\Omega$ (cf. an adjusted version of \eqref{eq:primalnonmarkovproof1}). Furthermore, it is an immediate consequence of Condition~\textit{(i)} that $\exp\hspace{0.05cm}\cdot\hspace{0.05cm} (U'\circ \exp)$ is nonincreasing as well. Additionally, we have that $U\circ \exp$ is a nondecreasing function on $\mathbb{R}$. Hence, it follows that $U\circ \exp$ is a concave and nondecreasing function. Lemma \ref{lemma:compositionconvexconcave}, therefore, guarantees that $\Phi_{\pi}^{n}=U\circ \exp \circ \hspace{0.05cm}g_{\pi}^{n}$ is a concave function. This observation shows that $(\Delta_{\varepsilon}\Phi_{\pi}^{n})_{\varepsilon\in (0,1]}$ is nonincreasing. Hence, this family is almost surely bounded from below by the integrable (cf. Assumption~\ref{assumption:primalintegrability}) random variable $U\big(X^{\pi^{*}+\mathbbm{1}_{[0,\tau_{n}^{\pi}]}\hspace{0.05cm} (\pi-\pi^{*})}(T)\big)- U\big(X^{\pi^{*}}(T)\big)$. Therefore, we obtain from the optimality of $\pi^{*}$, the monotone convergence theorem and \eqref{eq:primalnonmarkovproof6}:
    \begin{equation}\label{eq:primalnonmarkovproof31}
        0\ge \lim_{\varepsilon\searrow 0}  \mathbb{E}\big[\Delta_{\varepsilon}\Phi_{\pi}^{n}\big] =\mathbb{E}\big[ U'\big(X^{\pi^{*}}(T)\big) \hspace{0.05cm} X^{\pi^{*}}(T)\hspace{0.05cm} H_{\pi}^{\tau_{n}^{\pi}}(T) \big],
    \end{equation}
    for every $n\in\mathbb{N}$. Clearly, this implies
    \begin{equation}\label{eq:primalnonmarkovproof3last}
        \mathbb{E}\big[ U'\big(X^{\pi^{*}}(T)\big) \hspace{0.05cm} X^{\pi^{*}}(T)\hspace{0.05cm} H_{\pi}^{\tau_{n}^{\pi}}(T) \big] \le 0,
    \end{equation}
    where the expectation on the left-hand side is finite due to \eqref{eq:primalnonmarkovproof4} and Assumption \ref{assumption:primalintegrability}.\newline
    \textit{(ii)}: As we shall see in \textit{Step 4} below, it is sufficient to consider controls of the form $\pi = \overline{\pi}\hspace{0.05cm}\mathbbm{1}_{C} + \pi^{*}\hspace{0.05cm}\mathbbm{1}_{C^{c}}$ with a constant control $\overline{\pi}$ and $C\in \Sigma_{p}^{\mathbb{F},[0,T]}$. Hence, we restrict ourselves to these controls in the following. For a control $\pi$ with the above structure, we obtain $\mathbbm{1}_{[0,\tau_{n}^{\pi}]}\hspace{0.05cm}(\pi-\pi^{*})=\mathbbm{1}_{[0,\tau_{n}^{\pi}]}\hspace{0.05cm}\mathbbm{1}_{C}\hspace{0.05cm}(\overline{\pi}-\pi^{*})$. We notice that $\mathbbm{1}_{[0,\tau_{n}^{\pi}]}$ is progressively measurable because its paths are left-continuous and $\{\omega\in\Omega \hspace{0.1cm}|\hspace{0.1cm} \mathbbm{1}_{[0,\tau_{n}^{\pi}(\omega)]}(t)=0\}=\{\tau_{n}^{\pi}<t\}\in\mathcal{F}_{t}$ holds for every $t\in [0,T]$ as $\tau_{n}^{\pi}$ is in particular a weak stopping time. Hence, it follows from Condition \textit{(ii)} that $(\Delta_{\varepsilon}\Phi_{\pi}^{n})_{\varepsilon\in (0,1)}$ is uniformly integrable for every $n\in\mathbb{N}$. As $\mathbb{P}$ is in particular a finite measure, the convergence in \eqref{eq:primalnonmarkovproof6} also holds in probability. Combining this with the optimality of $\pi^{*}$ and Vitali's convergence theorem proves the validity of \eqref{eq:primalnonmarkovproof31} and \eqref{eq:primalnonmarkovproof3last} also in this case. \newline
    \textit{(iii)}: We recall from the previous paragraph that $\mathbbm{1}_{[0,\tau_{n}^{\pi}]}\hspace{0.05cm}(\pi-\pi^{*})\in\Theta$ holds for controls of the form $\pi = \overline{\pi}\hspace{0.05cm}\mathbbm{1}_{C} + \pi^{*}\hspace{0.05cm}\mathbbm{1}_{C^{c}}$ as specified above. Hence, Condition \textit{(iii)} guarantees the existence of a random variable~$\xi$ with $\xi^{-}\in L^{1}$ which is (a.s.) a lower bound for the elements of $(\Delta_{\varepsilon}\Phi_{\pi}^{n})_{\varepsilon\in (0,1)}$. Therefore, we obtain from the optimality of $\pi^{*}$, Fatou's lemma and \eqref{eq:primalnonmarkovproof6}:
    \begin{equation}\label{eq:primalnonmarkovproof32}
        0\ge \liminf_{\varepsilon\searrow 0}  \mathbb{E}\big[\Delta_{\varepsilon}\Phi_{\pi}^{n}\big] \ge   \mathbb{E}\Big[\liminf_{\varepsilon\searrow 0}\Delta_{\varepsilon}\Phi_{\pi}^{n}\Big] =\mathbb{E}\big[ U'\big(X^{\pi^{*}}(T)\big) \hspace{0.05cm} X^{\pi^{*}}(T)\hspace{0.05cm} H_{\pi}^{\tau_{n}^{\pi}}(T) \big],
    \end{equation}
    which implies again \eqref{eq:primalnonmarkovproof3last}.
    \vspace{0.2cm}\newline
    \textit{Step 4. Necessary Condition: \eqref{eq:PrimalnonMarkovcondition1} is   satisfied: Applying It\^{o}'s formula to the left-hand side of \eqref{eq:primalnonmarkovproof3last} and\\ 
    \noindent\hspace*{1.1cm} concluding the proof by considering certain strategies $\pi\in\mathcal{A}$.} \vspace{0.1cm}\newline
    At first, we recall that the continuous process $p_{1}^{*}\hspace{0.05cm} X^{\pi^{*}}$ is even a martingale by assumption. Therefore, as $p_{1}^{*}$ satisfies the terminal condition associated with \eqref{eq:primaladjointeq}, we can conclude that $p_{1}^{*}\hspace{0.05cm} X^{\pi^{*}}$ is a modification of the continuous process $V^{*}$ discussed in the proof of Lemma \ref{lemma:existencesolutionstoBSDEprimal}. The continuity of both processes implies by means of a standard result that they are even indistinguishable. On the one hand, this observation facilitates the application of It\^{o}'s formula below and on the other hand it implies that $p_{1}^{*}\hspace{0.05cm} X^{\pi^{*}}$ (and, therefore, also $p_{1}^{*}$) is a strictly negative process. The latter will play a decisive role below, especially in \textit{Step 5}, where this property is also applicable since we did not use the optimality of $\pi^{*}$ for its proof.
    We obtain for every $n\in\mathbb{N}$ by means of It\^{o}'s formula
    \begin{equation*}
        \begin{aligned}
        d\big(p_{1}^{*}\hspace{0.05cm}X^{\pi^{*}}\hspace{0.05cm}H_{\pi}^{\tau_{n}^{\pi}}\big)(t)=\hspace{0.1cm} &p_{1}^{*}(t)\hspace{0.05cm}X^{\pi^{*}}(t)\hspace{0.05cm}\mathbbm{1}_{[0,\tau_{n}^{\pi}]}(t)\hspace{0.05cm}\big( \pi(t)-\pi^{*}(t)\big)^{\intercal}\hspace{0.03cm}\sigma(t)\hspace{0.03cm}\big(\theta(t) - \sigma^{\intercal}(t)\hspace{0.05cm} \pi^{*}(t)\big)\hspace{0.05cm}dt\\
        &+\mathbbm{1}_{[0,\tau_{n}^{\pi}]}(t)\hspace{0.05cm}\big( \pi(t)-\pi^{*}(t)\big)^{\intercal}\hspace{0.03cm}\sigma(t)\hspace{0.03cm}\big(  p_{1}^{*}(t)\hspace{0.05cm}X^{\pi^{*}}(t)\hspace{0.05cm} \sigma^{\intercal}(t) \hspace{0.05cm} \pi^{*}(t) + X^{\pi^{*}}(t)\hspace{0.05cm}
        q_{1}^{*}(t) \big)\hspace{0.05cm}dt \\
        &+  p_{1}^{*}(t)\hspace{0.05cm}X^{\pi^{*}}(t)\hspace{0.05cm}\mathbbm{1}_{[0,\tau_{n}^{\pi}]}(t)\hspace{0.05cm}\big( \pi(t)-\pi^{*}(t)\big)^{\intercal}\hspace{0.03cm}\sigma(t)\hspace{0.05cm}dB(t)\\
        &+ \big(  p_{1}^{*}(t)\hspace{0.05cm} X^{\pi^{*}}(t)\hspace{0.05cm} \pi^{*\hspace{0.045cm}\intercal}(t)\hspace{0.05cm}\sigma(t) + X^{\pi^{*}}(t)\hspace{0.05cm}q_{1}^{*\hspace{0.045cm}\intercal}(t) \big)\hspace{0.05cm} H_{\pi}^{\tau_{n}^{\pi}}(t) \hspace{0.05cm}dB(t),
        \end{aligned}
    \end{equation*}
    $t\in [0,T]$. We observe that the integrand of the finite variation part can be simplified to
    \begin{equation}\label{eq:primalnonmarkovproof41}
        X^{\pi^{*}}(t)\hspace{0.05cm}\mathbbm{1}_{[0,\tau_{n}^{\pi}]}(t)\hspace{0.05cm}\big( \pi(t)-\pi^{*}(t)\big)^{\intercal}\hspace{0.03cm}\sigma(t)\hspace{0.03cm}\big(p_{1}^{*}(t)\hspace{0.05cm}\theta(t) + q_{1}^{*}(t)\big), \hspace{0.3cm} t\in [0,T].
    \end{equation}
    Note that the process which is multiplied by $H_{\pi}^{\tau_{n}^{\pi}}$ in the dynamics above corresponds exactly to $W^{*\hspace{0.045cm}\intercal}$ from the proof of Lemma \ref{lemma:existencesolutionstoBSDEprimal}, which lies in $H^{2}(0,T;\mathbb{R}^{m})$. Hence, it follows from \eqref{eq:primalnonmarkovproof4} that the local martingale part defines in fact a true martingale. Combining this with $H_{\pi}^{\tau_{n}^{\pi}}(0)=0$ and \eqref{eq:primalnonmarkovproof3last} shows for every $\pi\in\mathcal{A}$ for which the conclusion of \textit{Step 3} holds:
    \begin{equation}\label{eq:primalnonmarkovproof42}
        0 \ge \mathbb{E}\bigg[\int_{0}^{\tau_{n}^{\pi}} - X^{\pi^{*}}(t)\hspace{0.03cm}\big( \pi(t)-\pi^{*}(t)\big)^{\intercal}\hspace{0.03cm}\sigma(t)\hspace{0.03cm}\big(p_{1}^{*}(t)\hspace{0.05cm}\theta(t) + q_{1}^{*}(t)\big)\hspace{0.05cm}dt\bigg],
    \end{equation}
    for every $n\in\mathbb{N}$. We conclude this segment of the proof by arguing that \eqref{eq:primalnonmarkovproof42} necessarily implies \eqref{eq:PrimalnonMarkovcondition1}. For this purpose, we apply the same argument as given in \cite{zheng2}. Fix $\overline{\pi}\in K$ and consider the set
    \begin{equation}\label{eq:primalnonmarkovproof43}
        N^{\overline{\pi}}:=\big\{(\omega,t)\in\Omega\times [0,T]: -\big(\overline{\pi}- \pi^{*}(t) \big)^{\intercal}\hspace{0.03cm}\sigma(t)\hspace{0.03cm}\big(p_{1}^{*}(t)\hspace{0.05cm}\theta(t)+q_{1}^{*}(t) \big) >0\big\}.
    \end{equation}
    Moreover, we define the control process $\pi^{\overline{\pi}}:=\overline{\pi}\hspace{0.05cm}\mathbbm{1}_{N^{\overline{\pi}}} + \pi^{*}\hspace{0.05cm}\mathbbm{1}_{(N^{\overline{\pi}})^{c}}$. The admissibility of $\pi^{\overline{\pi}}$ results from ${\pi^{*}}\in\mathcal{A}$ and $N^{\overline{\pi}}$ being measurable with respect to the progressive $\sigma$-algebra. It is an immediate consequence of the special structure of $\pi^{\overline{\pi}}$ that \textit{Step 3} (and, therefore, also \eqref{eq:primalnonmarkovproof42}) is applicable to $\pi^{\overline{\pi}}$ under each of the three Conditions \textit{(i)}, \textit{(ii)} and \textit{(iii)}, respectively. As $X^{\pi^{*}}$ is a strictly positive process, we necessarily obtain from \eqref{eq:primalnonmarkovproof42} applied to $\pi^{\overline{\pi}}$ that, for every $n\in\mathbb{N}$, the set
    \begin{equation}\label{eq:primalnonmarkovproof44}
        N^{\overline{\pi}} \cap \{(\omega,t)\in\Omega\times [0,T]:\mathbbm{1}_{[0,\tau_{n}^{\pi^{\overline{\pi}}}]}(t)=1\}
    \end{equation}
    has to be a $\mathbb{P}\otimes\lambda|_{[0,T]}$-null set. Hence, as $\mathbbm{1}_{[0,\tau_{n}^{\pi^{\overline{\pi}}}]}\nearrow \mathbbm{1}_{[0,T]}$ holds almost surely for $n\rightarrow\infty$ and $\mathbb{P}\otimes\lambda|_{[0,T]}$ is continuous from below, we can conclude that also $N^{\overline{\pi}}$ is a $\mathbb{P}\otimes\lambda|_{[0,T]}$-null set. Moreover, we obtain from the subadditivity property that even
    \begin{displaymath}
        N:=\bigcup_{\overline{\pi}\in K\cap \mathbb{Q}^{m}} N^{\overline{\pi}}
    \end{displaymath}
    is a $\mathbb{P}\otimes\lambda|_{[0,T]}$-null set. It follows from $K\cap \mathbb{Q}^{m}$ being dense in $K$ and the continuity of the Euclidean inner product in every component that we have for $\mathbb{P}\otimes\lambda|_{[0,T]}$-almost every $(\omega,t)\in\Omega\times [0,T]$ and every $\overline{\pi}\in K$:
    \begin{equation}\label{eq:primalnonmarkovproof45}
        -\big(\overline{\pi}- \pi^{*}(t) \big)^{\intercal}\hspace{0.05cm}\sigma(t)\hspace{0.05cm}\big(p_{1}^{*}(t)\hspace{0.05cm}\theta(t)+q_{1}^{*}(t) \big) \le 0,
    \end{equation}
    which is equivalent to \eqref{eq:PrimalnonMarkovcondition1}.
    \vspace{0.2cm}\newline
    \textit{Step 5. Sufficient Condition: $\pi^{*}$ is optimal: Proving that $p_{1}^{*}\hspace{0.05cm}X^{\pi}$ is a submartingale for every $\pi\in\mathcal{A}$ and\\ 
    \noindent\hspace*{1.1cm} applying the concavity of $U$.}\vspace{0.1cm}\newline
    Finally, we show that \eqref{eq:PrimalnonMarkovcondition1} is also sufficient for the optimality of $\pi^{*}$. Let $\pi^{*}\in\mathcal{A}$ be given and assume that \eqref{eq:PrimalnonMarkovcondition1} is satisfied. We obtain for every $\pi\in\mathcal{A}$ by means of the integration by parts formula:
    \begin{equation}\label{eq:primalnonmarkovproof51}
        \begin{aligned}
            p_{1}^{*}\hspace{0.05cm}X^{\pi} =\hspace{0.2cm} &p_{1}^{*}(0)\hspace{0.05cm}x_{0} - \int_{0}^{\cdot} X^{\pi}(t)\hspace{0.05cm}\big[\big(r(t) + \pi^{* \hspace{0.045cm}\intercal}(t)\hspace{0.05cm}\sigma(t)\hspace{0.05cm}\theta(t) \big)\hspace{0.05cm} p_{1}^{*}(t) + \pi^{* \hspace{0.045cm}\intercal}(t)\hspace{0.05cm} \sigma(t)\hspace{0.05cm}q_{1}^{*}(t) \big]\hspace{0.05cm} dt\\
            &+\int_{0}^{\cdot}X^{\pi}(t)\hspace{0.05cm}\big[ \big(r(t)+\pi^{\intercal}(t)\hspace{0.05cm}\sigma(t)\hspace{0.05cm} \theta(t) \big)\hspace{0.05cm}p_{1}^{*}(t) + \pi^{\intercal}(t)\hspace{0.05cm}\sigma(t)\hspace{0.05cm}q_{1}^{*}(t)\big] dt\\
            &+ \int_{0}^{\cdot} X^{\pi}(t)\hspace{0.05cm}\big(p_{1}^{*}(t)\hspace{0.05cm}\pi^{\intercal}(t)\hspace{0.05cm}\sigma(t) + q_{1}^{* \hspace{0.045cm}\intercal}(t)\big)\hspace{0.05cm} dB(t).
        \end{aligned}
    \end{equation}
    As the process $p_{1}^{*}\hspace{0.05cm}X^{\pi}$ is strictly negative (see \textit{Step 4} for the strict negativity of $p_{1}^{*}$), it follows that it can be written as $p_{1}^{*}(0)\hspace{0.05cm}x_{0}\hspace{0.05cm}\mathcal{E}(Z)$, where the continuous semimartingale $Z$ is given by
    \begin{equation}\label{eq:primalnonmarkovproof52}
       \begin{aligned}
            Z = & \int_{0}^{\cdot} \big(\pi^{\intercal}(t) - \pi^{* \hspace{0.045cm}\intercal}(t)\big)\hspace{0.03cm}\sigma(t)\hspace{0.03cm}\big( \theta(t) + p_{1}^{*\hspace{0.045cm}-1}(t)\hspace{0.05cm}
            q_{1}^{*}(t) \big)\hspace{0.05cm} dt\\
            &+ \int_{0}^{\cdot} \big[\pi^{\intercal}(t)\hspace{0.05cm}\sigma(t) + p_{1}^{*\hspace{0.045cm}-1}(t)\hspace{0.05cm}q_{1}^{* \hspace{0.045cm}\intercal}(t)\big]\hspace{0.05cm} dB(t) =:A+M.
        \end{aligned} 
    \end{equation}
    Therefore, we obtain the explicit representation $p_{1}^{*}(0)\hspace{0.05cm}x_{0}\hspace{0.05cm}\mathcal{E}(M)\exp(A)$. Clearly, $\mathcal{E}(M)$ is a supermartingale as it is a nonnegative local martingale with integrable starting value. Hence, it is an immediate consequence of $p_{1}^{*}(0)<0$ that $p_{1}^{*}(0)\hspace{0.05cm}x_{0}\hspace{0.05cm}\mathcal{E}(M)$ is a submartingale. Moreover, it follows from $p_{1}^{*}$ being strictly negative and \eqref{eq:PrimalnonMarkovcondition1} that the integrand of $A$ is nonpositive, i.e. the process $\exp(A)$ is decreasing. Obviously, $p_{1}^{*}\hspace{0.05cm}X^{\pi}$ is adapted and $L^{1}$. The second property essentially results from the above representation and $\exp(A)$ being bounded by $1$. Fix $s,t\in [0,T]$ with $s\le t$. The submartingale property follows from the monotonicity of $\exp(A)$ and $p_{1}^{*}(0)\hspace{0.05cm}x_{0}\hspace{0.05cm}\mathcal{E}(M)$ being a negative submartingale:
    \begin{equation}\label{eq:primalnonmarkovproof53}
        \mathbb{E}\big[p_{1}^{*}(t)\hspace{0.05cm}X^{\pi}(t)|\mathcal{F}_{s}\big]\ge \exp(A(s))\hspace{0.05cm}\mathbb{E}\big[ p_{1}^{*}(0)\hspace{0.05cm}x_{0}\hspace{0.05cm}\mathcal{E}(M)_{t} |\mathcal{F}_{s}\big]\ge p_{1}^{*}(s)\hspace{0.05cm}X^{\pi}(s).
    \end{equation}
    Finally, as $p_{1}^{*}\hspace{0.05cm}X^{\pi^{*}}$ is even a martingale, we obtain from the concavity of $U$:
    \begin{equation}\label{eq:primalnonmarkovproof54}
        \begin{aligned}
        0=p_{1}^{*}(0)\hspace{0.05cm}x_{0}-p_{1}^{*}(0)\hspace{0.05cm}x_{0} &\ge \mathbb{E}\big[ p_{1}^{*}(T)\hspace{0.05cm}X^{\pi^{*}}(T)\big] - \mathbb{E}\big[ p_{1}^{*}(T)\hspace{0.05cm}X^{\pi}(T)\big]\\ 
        &= \mathbb{E}\big[ U'\big( X^{\pi^{*}}(T)\big)\hspace{0.05cm}\big(X^{\pi}(T) - X^{\pi^{*}}(T) \big)\big]\\
        &\ge \mathbb{E}\big[ U\big(X^{\pi}(T) \big) - U\big(X^{\pi^{*}}(T) \big) \big].
        \end{aligned}
    \end{equation}
    Since $\pi\in\mathcal{A}$ was arbitrary, we can conclude that $\pi^{*}$ is indeed optimal.
\end{proof}
\begin{remark}\label{remark:primalZhengassumptionswrong}
    Note that our set of assumptions differs fundamentally from \cite{zheng2}. There, the proof of \eqref{eq:PrimalnonMarkovcondition1} being a necessary optimality condition is carried out under Assumption \ref{assumption:assumpU}. However, in contrast to the corresponding part in the proof of Theorem \ref{theorem:DualnonMarkov}, this is not sufficient for concluding that the family of difference quotients $(\Delta_{\varepsilon}\Phi_{\pi}^{n})_{\varepsilon\in (0,1]}$ is monotone. The key difference is that $\widetilde{U}\circ \exp$ is decreasing, whereas $U\circ\exp$ is increasing. Note that both functions are convex under Assumption \ref{assumption:assumpU} (cf. Lemma \ref{lemma:transformierteassumption}). Furthermore, $g_{\pi}^{n}$ as defined in \textit{Step 3} is always concave. In the case of $U\circ\exp$ it is, therefore, in general not true that $\Phi_{\pi}^{n}$ is either concave or convex as the inequalities $(1)$ and $(2)$ from the proof of Lemma \ref{lemma:compositionconvexconcave} cannot be combined by the transitive property anymore. However, as the second statement of Lemma \ref{lemma:compositionconvexconcave} shows, it is possible to resolve this issue, if $U\circ\exp$ is concave. Hence, in contrast to Assumption \ref{assumption:assumpU} (an important assumption for Theorem \ref{theorem:DualnonMarkov}), we have to require that $\mathrm{id}_{\mathbb{R}^{+}} \cdot U'$ is a nonincreasing function. Note that this assumption is also formulated in \cite{trivellato}. Since the only utility functions which satisfy the aforementioned condition and Assumption \ref{assumption:assumpU} at the same time are given by $c\cdot\log + d$, where $c\in\mathbb{R}^{+}$ and $d\in \mathbb{R}$ are fixed, we formulate by Conditions \textit{(ii)} and \textit{(iii)} alternatives which justify interchanging limit and expectation operator in \textit{Step 3}. Hence, other utility functions $U$, in particular those where $\mathrm{id}_{\mathbb{R}^{+}} \cdot U'$ is strictly increasing (like the power utility functions), are a priori not excluded from Theorem \ref{theorem:PrimalnonMarkov}.
\end{remark}
\begin{remark}\label{remark:changeserrorsinoriginal}
    Besides that, we implemented the following changes in comparison with \cite{zheng2}: We optimized the sequence of stopping times $(\tau_{n}^{\pi})_{n\in\mathbb{N}}$ insofar as our definition in \eqref{eq:primalnonmarkovproof4} makes it straightforward that the stochastic integrals in \textit{Step 4} are, in fact, true martingales. In \textit{Step 5}, we placed great emphasis on showing that $p_{1}^{*}\hspace{0.05cm}X^{\pi}$ is a submartingale for every $\pi\in\mathcal{A}$. Hence, this reminds us of a common concept in optimization theory: If a process from a family of supermartingales (which is either obtained by using the general adjoint equation for maximization problems specified to the primal problem instead of \eqref{eq:primaladjointeq} in the formulation of Theorem \ref{theorem:PrimalnonMarkov}, see Remark \ref{remark:connectionclassSmpprimal}, or by simply multiplying each element of the family by $-1$), which is indexed by the set of the admissible strategies, is even a true martingale, then the corresponding control is optimal. Except for \textit{Step 3}, the proof concept is very similar to the original one given in \cite{zheng2}.
\end{remark}
\begin{remark}\label{remark:normalconeequivalence}
    Note that \eqref{eq:PrimalnonMarkovcondition1} is equivalent to $-\sigma(t)\hspace{0.05cm}\big(p_{1}^{*}(t)\hspace{0.05cm}\theta(t)+q_{1}^{*}(t)\big)$ being an element of the normal cone of $K$ at $\pi^{*}(t)$ almost surely for almost every $t\in [0,T]$, i.e.
    \begin{equation}\label{eq:normalconeequivalence}
        \forall\pi\in K: \hspace{0.3cm} \big(\pi-\pi^{*}(t)\big)^{\intercal} \hspace{0.05cm} \big[-\sigma(t)\hspace{0.05cm}\big(p_{1}^{*}(t)\hspace{0.05cm}\theta(t)+q_{1}^{*}(t)\big) \big] \le 0.
    \end{equation}
\end{remark}
\begin{remark}\label{remark:connectionclassSmpprimal}
    It is an immediate consequence of \eqref{eq:generalizedHamutility} and the strict positivity of $X^{\pi^{*}}$ that \eqref{eq:PrimalnonMarkovcondition1} is equivalent to
    \begin{equation}\label{eq:Hammaxforprimal}
        \mathcal{H}_{1}\big(t,X^{\pi^{*}}(t),\pi^{*}(t),p_{1}^{*}(t),q_{1}^{*}(t) \big) = \inf_{\pi\in K} \mathcal{H}_{1}\big( t,X^{\pi^{*}}(t),\pi,p_{1}^{*}(t),q_{1}^{*}(t)\big),
    \end{equation}
    almost surely for almost every $t\in [0,T]$, i.e. $\pi^{*}$ minimizes $\mathcal{H}_{1}$. Hence, as the primal problem is a maximization problem and the result serving as the formal starting point (see \cite[Section 3.3]{yongzhou}) is devoted to minimization problems, it is hereby justified to view Theorem \ref{theorem:PrimalnonMarkov} as a stochastic maximum principle. Note that \eqref{eq:primaladjointeq} actually corresponds to the adjoint equation for the equivalent minimization problem with loss function $-U$. Since $\mathcal{H}_{1}$ is linear in $y$ and $z$, it follows that $(-p_{1}^{*},-q_{1}^{*})$ solves the classical adjoint equation for maximization problems \cite[(6.25)]{pham}, which, for the primal problem, is given by \eqref{eq:primaladjointeq} combined with the terminal condition $p_{1}^{*}(T)=U'(X^{\pi^{*}}(T))$. Hence, \eqref{eq:PrimalnonMarkovcondition1} would become a maximization condition for $\mathcal{H}_{1}$, if we replaced \eqref{eq:primaladjointeq} with its maximization setting analogue in the formulation of Theorem \ref{theorem:PrimalnonMarkov}.
\end{remark}
\begin{remark}\label{remark:forsufficientweakerrequirements}
    We observe that neither of the Conditions \textit{(i)}, \textit{(ii)} and \textit{(iii)} is necessary for performing the arguments in \textit{Step 5} above. These were just required in order to ensure that we can interchange limit and expectation operator in \textit{Step 3}. This is an important observation with regard to Theorem \ref{theorem:fromdualtoprimal} below. Clearly, the same applies to the requirement that $\mathrm{id}_{\mathbb{R}^{+}} \cdot U'$ is nondecreasing in the light of the sufficiency of the three optimality conditions given in Theorem \ref{theorem:DualnonMarkov}.
\end{remark}

\subsection{A Stochastic Maximum Principle for the Dual Problem}
\label{subsection:theoreticalbackgrounddualnonmarkovian}
In this subsection, we cite the stochastic maximum principle for the dual problem found in \cite{zheng2} as combining it with the results from the previous subsection allows our novel algorithm to efficiently calculate a natural upper bound for the value of the primal problem. Moreover, it is the key to the deep SMP algorithm from \cite{zheng1}. We are going to need the following technical assumption:
\begin{assumption}\label{assumption:dualintegrability}
    Suppose that $\widetilde{U}\big( Y^{(y,v)}(T)\big)\in L^{2}$ holds for any admissible pair $(y,v)\in\mathcal{D}$.
\end{assumption}
\begin{remark}\label{remark:existenceoptlabbe}
    The motivation for formulating Assumption \ref{assumption:dualintegrability} as above is a result from \cite{labbe} which guarantees that there exists an optimal pair $(y^{*},v^{*})$ for the dual problem, if, in addition, Assumption \ref{assumption:assumpU}, $U(0)>-\infty$, $U(\infty)=\infty$ and $V\in\mathbb{R}$ (cf. \eqref{eq:Umaxproblem}) hold. Alternatively, we could directly require the integrability condition which is needed for the second part of Lemma \ref{lemma:existencesolutionstoBSDEdual}, as it is done for Lemma \ref{lemma:existencesolutionstoBSDEprimal} in Subsection \ref{subsection:nonmarkovprimalSMP}.
\end{remark}
\noindent The first auxiliary result ensures that, in the setting of Theorem \ref{theorem:DualnonMarkov}, there always exists a pair solving \eqref{eq:dualadjointeq} which satisfies the prerequisites.
\begin{lemma}\label{lemma:existencesolutionstoBSDEdual}
Consider a fixed pair $(y,v)\in\mathcal{D}$ and suppose that Assumptions \ref{assumption:assumpU} and \ref{assumption:dualintegrability} and either, using the notation of Lemma \ref{lemma:transformierteallg}, $U(0)>-\infty$ or $U=\log$ hold. Then $Y^{(y,v)}(T) \hspace{0.05cm}\widetilde{U}'\big(Y^{(y,v)}(T) \big)\in L^{2}$ and there exists a solution $(p_{2},q_{2})$ to \eqref{eq:dualadjointeq} such that $p_{2}\hspace{0.03cm}Y^{(y,v)}$ is a martingale. 
\end{lemma}
\begin{proof}
    We refer to the proof of \cite[Lemma 3.7]{zheng2} for details. The proof of the first part essentially relies on the results gathered in Lemmata \ref{lemma:transformierteallg} and \ref{lemma:transformierteassumption}. As in Lemma \ref{lemma:existencesolutionstoBSDEprimal}, solutions to \eqref{eq:dualadjointeq} with the desired properties can be explicitly constructed using the martingale representation theorem and It\^{o}'s formula.
\end{proof}
\noindent Now that we have this preliminary result at our disposal, we are in position to cite the stochastic maximum principle for the dual problem which was found in \cite{zheng2}.
\begin{theorem}\label{theorem:DualnonMarkov}
    Let a dual problem be given such that Assumptions \ref{assumption:assumpU} and \ref{assumption:dualintegrability} and either $U(0)>-\infty$ or $U=\log$ hold and consider an admissible pair $(y^{*},v^{*})\in\mathcal{D}$. Moreover, let $Y^{(y^{*},v^{*})}$, $p_{2}^{*}$ and $q_{2}^{*}$ denote processes which solve the following FBSDE system and assume that $p_{2}^{*}\hspace{0.05cm} Y^{(y^{*},v^{*})}$ is even a martingale: For~$t\in [0,T]$: 
    \begin{equation}\label{eq:SDEsystemdualnonMarkov}
        \begin{aligned}
        &dY^{(y^{*},v^{*})}(t) = -Y^{(y^{*},v^{*})}(t) \big[ \big(r(t)+\delta_{K}(v^{*}(t)) \big) \hspace{0.05cm} dt + \big( \theta(t)+\sigma^{-1}(t)v^{*}(t)\big)^{\intercal}\hspace{0.05cm} dB(t)\big], \hspace{0.3cm} \\
        &dp_{2}^{*}(t) = \big[\big(r(t)+\delta_{K}(v^{*}(t))\big)\hspace{0.05cm} p_{2}^{*}(t)  + \big(\theta(t)+ \sigma^{-1}(t)\hspace{0.05cm}v^{*}(t)\big)^{\intercal}\hspace{0.05cm} q_{2}^{*}(t) \hspace{0.05cm}\big]\hspace{0.05cm} dt + q_{2}^{* \hspace{0.045cm}\intercal}(t)\hspace{0.05cm} dB(t),
        \end{aligned}
    \end{equation}
    with initial condition $Y^{(y^{*},v^{*})}(0)=y^{*}$ and terminal condition $p_{2}^{*}(T)= -\widetilde{U}'\big(Y^{(y^{*},v^{*})}(T) \big)$, respectively. Then $(y^{*},v^{*})$ is optimal in the sense of Definition \ref{definition:dualproblem} if and only if the following properties are satisfied almost surely for almost every $t\in [0,T]$:
    \begin{eqnarray}
    &&p_{2}^{*}(0)=x_{0},\label{eq:DualnonMarkovcondition1}\\
    &&p_{2}^{*\hspace{0.045cm} -1}(t) \big[\sigma^{-1}\big]^{\intercal}(t)\hspace{0.05cm} q_{2}^{*}(t) \in K \hspace{0.3cm}\mbox{and}\label{eq:DualnonMarkovcondition2}\\
    &&p_{2}^{*}(t)\hspace{0.05cm}\delta_{K}(v^{*}(t)) + \big(\sigma^{-1}(t)\hspace{0.05cm}v^{*}(t)\big)^{\intercal}\hspace{0.05cm} q_{2}^{*}(t)=0. \label{eq:DualnonMarkovcondition3}
    \end{eqnarray}
\end{theorem}
\begin{proof}
     The original proof of the above statement can be found in \cite[Theorem 3.9]{zheng2}. Moreover, we refer to the author's diploma thesis (see \cite[Theorem 4.6]{KW2021deep}) for a very detailed exposition of the proof removing minor mistakes from the original presentation (see also Remarks 4.7 and 4.8 therein). Even though the original result is formulated under the premise that the processes $r$, $\mu$ and $\sigma$ are uniformly bounded and $\sigma$ satisfies the strong non-degeneracy condition, it can easily be seen that the arguments are also applicable, if the processes $r$, $\mu$ and $\sigma$ are merely continuous and $\sigma^{-1}$ exists.
\end{proof}
\begin{remark}\label{remark:connectionclassSmpdual}
   It follows from \eqref{eq:generalizedHamdual} and the strict positivity of $Y^{(y^{*},v^{*})}$ that the Hamiltonian maximization condition
   \begin{equation}\label{eq:Hammaxfordual1}
        \mathcal{H}_{2}\big(t,Y^{(y^{*},v^{*})}(t),v^{*}(t),p_{2}^{*}(t),q_{2}^{*}(t) \big) = \sup_{v\in \widetilde{K}} \mathcal{H}_{2}\big( t,Y^{(y^{*},v^{*})}(t),v,p_{2}^{*}(t),q_{2}^{*}(t)\big)
    \end{equation}
   is equivalent to
   \begin{equation}\label{eq:Hammaxfordual2}
        p_{2}^{*}(t)\hspace{0.05cm}\delta_{K}\big(v^{*}(t)\big) + \big(\sigma^{-1}(t)\hspace{0.05cm}v^{*}(t)\big)^{\intercal}\hspace{0.05cm} q_{2}^{*}(t)= \inf_{v\in\widetilde{K}}\big\{p_{2}^{*}(t)\hspace{0.05cm}\delta_{K}(v) + \big(\sigma^{-1}(t)\hspace{0.05cm}v\big)^{\intercal}\hspace{0.05cm} q_{2}^{*}(t)\big\}
    \end{equation}
    being satisfied almost surely for almost every $t\in [0,T]$. Clearly, \eqref{eq:DualnonMarkovcondition3} yields that the left-hand side is $0$. Moreover, we obtain from \eqref{eq:DualnonMarkovcondition2}, $p_{2}^{*}$ being strictly positive and the definition of $\delta_{K}$ that the expression within the curly brackets is nonnegative for every $v\in\widetilde{K}$. Hence, \eqref{eq:Hammaxfordual1} holds indeed for an optimal dual control process $v^{*}$. Furthermore, as the proof of Theorem \ref{theorem:DualnonMarkov} shows, \eqref{eq:DualnonMarkovcondition1} results from considering the minimization problem $\inf_{y\in\mathbb{R}^{+}} \big\{x_{0} y + \mathbb{E}\big[ \widetilde{U}\big( Y^{(y,v^{*})}(T)\big)\big] \big\}$ with $v^{*}$ being fixed. Consequently, Theorem \ref{theorem:DualnonMarkov} can be regarded as a stochastic maximum principle for the dual problem which also takes the optimization with respect to $y$ into account.
\end{remark}

\subsection{Constructing a Solution to the Utility Maximization Problem by Means of a Solution to Its Dual Problem and Vice Versa}\label{subsection:connectionprimaldual}
It is natural to ask how good of an upper estimate the value of the dual problem actually is. 
As a first step, the following theorem found in \cite{zheng2} proves by means of Theorems \ref{theorem:PrimalnonMarkov} and \ref{theorem:DualnonMarkov} and the technical assumption \eqref{eq:fromdualtoprimal1} that we can construct a solution to the primal problem, if there exists an optimal pair for the associated dual problem.
\begin{theorem}\label{theorem:fromdualtoprimal}
    Let a dual problem be given such that Assumptions \ref{assumption:assumpU} and \ref{assumption:dualintegrability} and either $U(0)>-\infty$ or $U=\log$ hold. Assume that there exists an optimal pair $(y^{*},v^{*})\in\mathcal{D}$. Let $Y^{(y^{*},v^{*})}$, $p_{2}^{*}$ and $q_{2}^{*}$ denote processes which solve \eqref{eq:SDEsystemdualnonMarkov} and suppose that $p_{2}^{*}\hspace{0.05cm} Y^{(y^{*},v^{*})}$ is a martingale. Moreover, assume
    \begin{equation}\label{eq:fromdualtoprimal1}
        \pi^{*}:=p_{2}^{*\hspace{0.045cm}-1}\hspace{0.03cm} \big[\sigma^{-1}\big]^{\intercal}\hspace{0.03cm}q_{2}^{*} \in H^{2}(0,T;\mathbb{R}^{m}).
    \end{equation}
    Then it follows that $\pi^{*}$ is an optimal control for the corresponding primal problem. Furthermore, a triple of processes solving \eqref{eq:SDEsystemprimalnonMarkov} for $\pi^{*}$ is explicitly given by 
    \begin{equation}\label{eq:fromdualtoprimal2}
       X^{\pi^{*}}:=p_{2}^{*},\hspace{0.3cm} p_{1}^{*}:=-Y^{(y^{*},v^{*})} \hspace{0.3cm} \mbox{and}\hspace{0.3cm} q_{1}^{*}:= Y^{(y^{*},v^{*})}\hspace{0.03cm}(\sigma^{-1}\hspace{0.05cm} v^{*}+\theta). 
    \end{equation}
\end{theorem}
\begin{proof}
    For a proof of the above statement, we refer to \cite[Theorem 3.10]{zheng2}.
\end{proof}
\noindent Moreover, also the ``converse'' holds (cf. \cite{zheng2}):
\begin{theorem}\label{theorem:fromprimaltodual}
    Assume that there exists an optimal control $\pi^{*}$ for the utility maximization problem with initial wealth $x_{0}$. Moreover, suppose that the requirements for Theorem \ref{theorem:PrimalnonMarkov} (including at least one of the Conditions \textit{(i)}, \textit{(ii)} and \textit{(iii)}) are met and let $X^{\pi^{*}}$, $p_{1}^{*}$ and $q_{1}^{*}$ denote processes solving \eqref{eq:SDEsystemprimalnonMarkov} such that $p_{1}^{*}\hspace{0.03cm}X^{\pi^{*}}$ is even a martingale. If the process $-\sigma\hspace{0.03cm}(p_{1}^{*\hspace{0.045cm}-1}\hspace{0.03cm}q_{1}^{*}+\theta)$ satisfies the technical integrability property \eqref{eq:dualadmissible}, then 
    \begin{equation}\label{eq:fromprimaltodual1}
        (y^{*},v^{*}):=\big(-p_{1}^{*}(0), -\sigma\hspace{0.03cm}(p_{1}^{*\hspace{0.045cm}-1} \hspace{0.03cm}q_{1}^{*}+\theta)\big)
    \end{equation}
    is an optimal pair. A triple of processes solving \eqref{eq:SDEsystemdualnonMarkov} for $(y^{*},v^{*})$ is explicitly given by
    \begin{equation}\label{eq:fromprimaltodual2}
        Y^{(y^{*},v^{*})}:= -p_{1}^{*}, \hspace{0.3cm} p_{2}^{*}:= X^{\pi^{*}}\hspace{0.3cm}\mbox{and}\hspace{0.3cm} q_{2}^{*}:= \sigma^{\intercal}\hspace{0.03cm}\pi^{*}\hspace{0.03cm}X^{\pi^{*}}.
    \end{equation}
\end{theorem}
\begin{proof}
    For a proof of the above statement, we refer to \cite[Theorem 3.11]{zheng2}. Note, however, that requiring at least one of the Conditions \textit{(i)}, \textit{(ii)} and \textit{(iii)} is essential in order to ensure that \eqref{eq:PrimalnonMarkovcondition1} becomes indeed a necessary optimality condition (cf. Remark \ref{remark:primalZhengassumptionswrong}).
\end{proof}
\noindent By means of the above results, we can quite easily prove that the so-called strong duality property, i.e. $V=\widetilde{V}$, holds. This means that the value of the primal problem agrees with the one of the associated dual problem.
\begin{corollary}\label{corollary:strongduality}
    Consider a utility maximization problem and its dual problem and suppose that the requirements of either Theorem \ref{theorem:fromdualtoprimal} or Theorem \ref{theorem:fromprimaltodual} are met. Then the strong duality property holds.
\end{corollary}
\begin{proof}
    Due to our choice of the prerequisites, we obtain that there also exists an optimal control for the respective other problem and the desired FBSDE solution can be constructed by means of either \eqref{eq:fromdualtoprimal2} or \eqref{eq:fromprimaltodual2}. Hence, we are required to prove
    \begin{equation}\label{eq:strongduality1}
        \mathbb{E}\big[U\big(p_{2}^{*}(T) \big) \big] = x_{0}y^{*} + \mathbb{E}\big[ \widetilde{U}\big(Y^{(y^{*},v^{*})}(T) \big)\big].
    \end{equation}
    Combining $p_{2}^{*}(T)=-\widetilde{U}'\big(Y^{(y^{*},v^{*})}(T) \big)$ with Lemma \ref{lemma:transformierteallg} \textit{(iii)} and \textit{(iv)} on the left-hand side of \eqref{eq:strongduality1} shows the equivalence with
    \begin{equation}\label{eq:strongduality2}
        \mathbb{E}\big[\widetilde{U}\big(Y^{(y^{*},v^{*})}(T) \big) + p_{2}^{*}(T)\hspace{0.05cm}Y^{(y^{*},v^{*})}(T)\big] = x_{0}y^{*} + \mathbb{E}\big[ \widetilde{U}\big(Y^{(y^{*},v^{*})}(T) \big)\big].
    \end{equation}
    As $p_{2}^{*}\hspace{0.05cm}Y^{(y^{*},v^{*})}$ is a martingale and $p_{2}^{*}(0)=x_{0}$ according to either \eqref{eq:DualnonMarkovcondition1} or \eqref{eq:fromprimaltodual2}, we obtain
    \begin{equation}\label{eq:strongduality3}
        \mathbb{E}\big[ p_{2}^{*}(T)\hspace{0.05cm}Y^{(y^{*},v^{*})}(T)\big] = x_{0}y^{*},
    \end{equation}
    which concludes the proof as \eqref{eq:strongduality3} is equivalent to \eqref{eq:strongduality2}.
\end{proof}
\noindent Corollary \ref{corollary:strongduality} in connection with \eqref{eq:fromprimaltodual2} suggests a method for obtaining high quality upper estimates for the true value of the utility maximization problem (see \eqref{eq:boundsprimaldeepSMP} below). This is important as, in contrast to the deep controlled 2BSDE algorithm and like the deep SMP algorithm (see \cite{zheng1}), the deep primal SMP algorithm will not provide a natural estimate for the value of the problem by means of the initial values of the processes from \eqref{eq:SDEsystemprimalnonMarkov}. Instead, lower and upper bounds have to be calculated using the Monte Carlo method.

\subsection{Formulation of the Novel Deep Primal SMP Algorithm}
\label{subsection:primalDeepSMPAlgorithm}

It is natural to ask whether Theorem \ref{theorem:PrimalnonMarkov} can be used for formulating an algorithm which solves the primal problem directly. In \cite[Remark 4.2]{zheng1}, the authors dismiss this idea by arguing that the implementation of the maximization condition \eqref{eq:PrimalnonMarkovcondition1} would require a computationally costly, grid-based method. However, as \eqref{eq:PrimalnonMarkovcondition1} corresponds to minimizing the generalized Hamiltonian $\mathcal{H}$ in the control argument (cf. Remark \ref{remark:connectionclassSmpprimal}), we aim at integrating this condition into a novel algorithm in a way similar to how the maximization condition for the classical Hamiltonian of Markovian problems (see \cite[Theorem 3.2]{zheng1}) was utilized in the formulation of the deep controlled 2BSDE algorithm. Our numerical experiments in Section \ref{section:NumericalExperiments} below will show that this is a highly worthwhile approach since it overcomes the deep SMP algorithm's difficulties with constrained problems (see Example \ref{example:highdimlogutconstr} for details). This novel method, which we will call deep primal SMP algorithm in the following, combines essentially the BSDE solver from \cite{jentzenbsde} with the optimality condition for the control process given in Theorem \ref{theorem:PrimalnonMarkov}. \vspace{0.3cm}\newline
Let us fix an equidistant time discretization $(t_{i})_{i\in\{0,\dots,N\}}$ of $[0,T]$ with step size $\Delta t := T/N$, where $N\in\mathbb{N}$. We are going to simulate the system \eqref{eq:SDEsystemprimalnonMarkov} from Theorem \ref{theorem:PrimalnonMarkov} by means of a discrete-time forward scheme motivated by the Euler-Maruyama method which yields two processes $(p_{i})_{i\in\{0,\dots,N\}}$ and $(X_{i})_{i\in\{0,\dots,N\}}$. For every $i\in\{0,\dots,N-1\}$, we substitute $\pi(t_{i})$ and $q_{1}(t_{i})$ with
\begin{equation}\label{eq:networkprimalSMP}
    \pi_{i}:= \mathcal{N}_{\theta_{i,\pi}}(X_{i}) \hspace{0.3cm}\mbox{and}\hspace{0.3cm} q_{i}:= \mathcal{N}_{\theta_{i,q}}(X_{i}),
\end{equation}
respectively, where the neural networks $\mathcal{N}_{\theta_{i,\pi}}: \mathbb{R}\rightarrow\mathbb{R}^{m}$ and $\mathcal{N}_{\theta_{i,q}}: \mathbb{R}\rightarrow\mathbb{R}^{m}$ are parameterized by appropriate vectors $\theta_{i,\pi}$ and $\theta_{i,q}$. Moreover, hard constraints guarantee via the final layer of each $\mathcal{N}_{\theta_{i,\pi}}$ by means of a surjective, a.e. differentiable function mapping to $K$ that every $\pi_{i}$ is $K$-valued. For $i=0$, we can simplify the network architectures according to the following remark. 
\begin{remark}\label{remark:F0measurablemodelling}
    By the definition of our filtration, random variables which are measurable with respect to~$\mathcal{F}_{0}$ have to be almost surely constant. Hence, it is sufficient to model $\pi_{0}$ and $q_{0}$ by means of trivial neural networks which merely consist of a bias vector. This becomes even clearer, if we combine the ansatz \eqref{eq:networkprimalSMP} with $X_{0}$ being deterministic. Hence, a more complex network architecture would not improve the result while still increasing the dimensionality of the problem. 
\end{remark}
While the Markovian network architectures in \eqref{eq:networkprimalSMP} are certainly not optimal for non-Markovian problems, it is a reasonable choice for the base variant of our algorithm. Depending on the problem, more complex architectures might lead to even better results (see Remark \ref{remark:newqSMP} and Subsection \ref{subsection:rekurrentverallg} below). We start our scheme with
\begin{equation}\label{eq:primalSMPinitialvalues}
    X_{0}:=x_{0} \hspace{0.3cm}\mbox{and}\hspace{0.3cm} p_{0}:=p,
\end{equation}
where $p$ is a trainable variable. Let $(\Delta B_{i})_{i\in\{0,\dots,N-1\}}$ be a family of i.i.d. $m$-dimensional, centered normally distributed random vectors with covariance matrix $\Delta t \cdot I_{m}$. Hence, we can inductively define for every $i\in\{0,\dots,N-1\}$:
\begin{equation}\label{eq:primalSMPdiscretesystem}
    \begin{aligned}
      X_{i+1} :=\hspace{0.12cm} &X_{i} +X_{i}\hspace{0.05cm}\big(r(t_{i})+\pi_{i}^{\intercal}\hspace{0.05cm}\sigma(t_{i})\hspace{0.05cm}\theta(t_{i})\big)
      \hspace{0.05cm}\Delta t +X_{i}\hspace{0.05cm} \pi_{i}^{\intercal}\hspace{0.05cm}\sigma(t_{i}) \hspace{0.05cm}\Delta B_{i},\\
      p_{i+1} :=\hspace{0.12cm} &p_{i} -\big[ \big(r(t_{i}) + \pi_{i}^{\intercal}\hspace{0.05cm}\sigma(t_{i})\hspace{0.05cm} \theta(t_{i}) \big)\hspace{0.05cm}p_{i} + \pi_{i}^{\intercal}\hspace{0.05cm}\sigma(t_{i})\hspace{0.05cm}q_{i}\big]\hspace{0.05cm}\Delta t + q_{i}^{\intercal}\hspace{0.05cm}\Delta B_{i},
    \end{aligned}
\end{equation}
which corresponds to a discretized version of \eqref{eq:SDEsystemprimalnonMarkov}. As $(p_{i})_{i\in\{0,\dots,N\}}$ is supposed to fulfill a terminal condition, we aim at finding parameters which minimize
\begin{equation}\label{eq:primalSMPBSDEloss}
    \mathcal{L}_{BSDE}(p,\theta_{0,q},\dots,\theta_{N-1,q}):= \mathbb{E}\big[  \big| p_{N} + U'(X_{N}) \big|^{2} \big].
\end{equation}
Furthermore, also the parameters for the control process have to be optimized. This is achieved by minimizing 
\begin{equation}\label{eq:primalSMPcontrolloss}
    \mathcal{L}_{control}^{i}(\theta_{i,\pi}):= \mathbb{E}\big[ \pi_{i}^{\intercal}\hspace{0.03cm}\sigma(t_{i})\hspace{0.03cm}\big(p_{i}\hspace{0.03cm}\theta(t_{i}) + q_{i} \big) \big],
\end{equation}
for every $i\in\{0,\dots,N-1\}$. In contrast to \eqref{eq:PrimalnonMarkovcondition1}, we use the expectation operator here as the algorithm only considers a finite number of realizations. Moreover, as specified in the next paragraph, it works with the same realization almost surely at most once. Subsection \ref{subsection:epochenverallg} investigates the effects of studying each element in an a priori drawn set of trajectories several times in the course of the entire training procedure. Clearly, this leads on the other hand to fewer trajectories studied overall. \vspace{0.3cm}\newline
Like the deep SMP algorithm, the deep primal SMP algorithm starts with randomly initialized parameters. It then repeats the training step procedure given by Algorithm \ref{alg:deepprimalSMP} until the updates become sufficiently small and the algorithm seems to have converged. In each training step, the updates are performed by means of one step of an SGD algorithm with respect to each loss function given above. Note that the simulations performed in Substep 2 of Algorithm \ref{alg:deepprimalSMP} already use the updated parameters from Substep 1, which should accelerate convergence. For every training step, the algorithm studies $b_{size}$ realizations of $(\Delta B_{i})_{i\in\{0,\dots,N-1\}}$ (i.e. it samples a batch $\{\omega_{j}\}_{j\in\{1,\dots,b_{size}\}}\subseteq\Omega$), which are then discarded. \vspace{0.3cm}\newline
\noindent As in the deep SMP algorithm and in contrast to the deep controlled 2BSDE algorithm (see \cite{zheng1}), there is no parameter which naturally describes the value of the control problem. However, motivated by Corollary \ref{corollary:strongduality} and \eqref{eq:fromprimaltodual2}, we can formulate a lower and an upper bound in a similar manner:
\begin{equation}\label{eq:boundsprimaldeepSMP}
    V_{l}:= \mathbb{E}\big[ U(X_{N})\big] \hspace{0.3cm}\mbox{and}\hspace{0.3cm} V_{u}:= \mathbb{E}\big[\widetilde{U}(-p_{N}) \big] - x_{0} p.
\end{equation}
Corollary \ref{corollary:strongduality} suggests that $V_{u}-V_{l}$ becomes small, if the technical requirements are satisfied. In Section \ref{section:NumericalExperiments}, we calculate $V_{l}$ and $V_{u}$ by drawing a large sample of size $N_{MC} >\!\!> b_{size}$ from $\Omega$, simulating the processes according to the above scheme and calculating the corresponding sample means. However, as this procedure is computationally expensive, the bounds should not be calculated after each training step. In Section \ref{section:NumericalExperiments}, they are computed every 200 training steps.

\begin{algorithm}
    \caption{One training step of the deep primal SMP algorithm}
    \label{alg:deepprimalSMP}
    \begin{algorithmic}[1]
        \State Generate $b_{size}$ realizations of $(\Delta B_{i})_{i\in\{0,\dots,N-1\}}$, i.e. $\{\omega_{j}\}_{j\in\{1,\dots,b_{size}\}}\subseteq\Omega$;
        \State // Substep 1: Minimizing $\mathcal{L}_{BSDE}$
        \State Initialize according to \eqref{eq:primalSMPinitialvalues} for every $j\in\{1,\dots,b_{size}\}$; 
        \For {$i=0,1,\ldots,N-1$}
            \For {$j=1,2,\ldots,b_{size}$}
                \State Calculate $\pi_{i}^{j}$ and $q_{i}^{j}$ by means of \eqref{eq:networkprimalSMP};
                \State Use \eqref{eq:primalSMPdiscretesystem} in order to obtain $X_{i+1}^{j}$ and $p_{i+1}^{j}$;
            \EndFor
        \EndFor
        \State $loss1 \gets \frac{1}{b_{size}} \sum_{j=1}^{b_{size}} \big| p_{N}^{j} + U'(X_{N}^{j}) \big|^{2}$;
        \State Update $p,\theta_{0,q},\dots,\theta_{N-1,q}$ with one step of an SGD algorithm w.r.t. $loss1$;
        \State // Substep 2: Minimizing $\mathcal{L}_{control}^{i}$ for every $i\in\{0,\dots,N-1\}$
        \For {$i=0,1,\ldots,N-1$}
            \For {$j=1,2,\ldots,b_{size}$}
            \If {i==0}
                \State Initialize $X_{0}^{j}$ and $p_{0}^{j}$ according to \eqref{eq:primalSMPinitialvalues};
            \Else
                \State Use \eqref{eq:primalSMPdiscretesystem} in order to obtain $X_{i}^{j}$ and $p_{i}^{j}$;
            \EndIf
            \State Calculate $\pi_{i}^{j}$ and $q_{i}^{j}$ by means of \eqref{eq:networkprimalSMP};
            \EndFor
            \State $loss2_{i}\gets \frac{1}{b_{size}} \sum_{j=1}^{b_{size}} \big(\pi_{i}^{j}\big)^{\intercal}\hspace{0.02cm}\sigma(\omega_{j},t_{i})\hspace{0.02cm}\big(p_{i}^{j}\hspace{0.05cm}\theta(\omega_{j},t_{i}) + q_{i}^{j} \big)$;
            \State Update $\theta_{i,\pi}$ with one step of an SGD algorithm w.r.t. $loss2_{i}$;
        \EndFor
    \end{algorithmic}
\end{algorithm}

\section{Numerical Experiments}
\label{section:NumericalExperiments}
In this section, we aim at providing numerical examples which allow us to compare the novel deep primal SMP algorithm with the deep controlled 2BSDE algorithm and the deep SMP algorithm from the literature~(see~\cite{zheng1}). The former method solves Markovian problems, as well as their dual problems, via a controlled F2BSDE system and a maximization/minimization condition for the classical Hamiltonian, whereas the latter algorithm tackles the dual problem via the optimality conditions given by Theorem \ref{theorem:DualnonMarkov}. The plausibility of this approach is guaranteed by Corollary \ref{corollary:strongduality}. All experiments have been performed in Python by means of the popular machine learning library TensorFlow (see Appendix \ref{section:Appendix} for the Python code used in Example \ref{example:highdimlogutconstr}). Moreover, we choose $b_{size}=64$ in the following. For both SMP-based algorithms, we select $N_{MC}=10^{5}$ and calculate the bounds every 200 training steps. Adam (see \cite{Adam}) will serve as our SGD algorithm. Like in~\cite{zheng1, jentzenbsde}, the neural networks for $t_{i}\neq 0$ consist of two hidden dense layers with 11 neurons each. Furthermore, we place a batch normalization (see \cite{batchnorm}) layer in front of every dense layer and ReLU serves as our activation function.
\subsection{Markovian Utility Maximization Problems}\label{subsection:markovcoeffexperiments}
At first, we study a problem where the processes $r$, $\mu$ and $\sigma$ are deterministic. Hence, the state processes are in particular of a controlled Markovian form. Therefore, all of the aforementioned algorithms can be applied without any restrictions.
\begin{remark}\label{remark:BNepsilon}
    During our numerical experiments, we observed that the batch normalization layers' parameter~$\varepsilon$ is an essential hyperparameter for the studied algorithms. As we shall see below, using the default value might lead to NaN-values, whereas choosing a larger value, e.g. $\varepsilon=1$, resolves this issue and allows the algorithm to converge astonishingly fast (see Example \ref{example:highdimlogutconstr}). Moreover, small values for $\varepsilon$ may also cause that the bounds of both SMP-based algorithms explode at the beginning of the training procedure (see Example~\ref{example:randomcoeffunconstrained} for details). Hence, tuning $\varepsilon$ for each problem is of paramount importance for the training success.
\end{remark}
\noindent In the following example, we consider a high-dimensional Markovian problem with logarithmic utility function. By introducing constraints, we increase the complexity of the problem even further.
\begin{example}\label{example:highdimlogutconstr}
    We consider a $30$-dimensional problem with $K=[-\kappa,\infty)^{30}$, where $\kappa:=1/m=1/30$. Hence, short selling is permitted, but only to a certain extent for each stock. This choice of $K$ implies in particular that the overall short selling volume of stocks must not surpass the portfolio value of an investor. One can easily see that $\widetilde{K}=(\mathbb{R}_{0}^{+})^{30}$ and $\delta_{K}(v)=\kappa\sum_{i=1}^{30} v_{i}$, $v\in\widetilde{K}$, hold. Moreover, the natural logarithm serves as our utility function. We recall from Example \ref{example:utilitydual} that $\widetilde{U}=-\log-1$ holds. Furthermore, we choose $x_{0}=10$ and $T=0.5$. For every $t\in [0,0.5]$ and $i\in\{1,\dots,30\}$, $r(t)$ and $\mu_{i}(t)$ are given by $0.06\hspace{0.02cm}\exp(t/2)$ and $0.07+0.02\hspace{0.02cm}\sin (4\pi t + \pi i/15)$, respectively. Furthermore, the diagonal elements $\sigma_{i,i}(t)$ are given by~$0.3\hspace{0.02cm}\big(1+\sqrt{t}\big)$, whereas the remaining entries are chosen to be $0.1$. As the logarithm satisfies some desirable properties, we can even find the exact dual value quite easily (cf. \cite{zheng2}). As the dual state process is a stochastic exponential and $\widetilde{U}=-\log-1$ holds, we obtain $y^{*}=x_{0}^{-1}=0.1$. Hence, we can conclude for $D_{2}:=\{v\in\mathcal{A}_{prog}\hspace{0.1cm} | \hspace{0.1cm} (1,v)\in\mathcal{D}\}$ due to the integrability assumptions in the definition of $\mathcal{D}$:
    \begin{equation}\label{eq:highdimlogutunconstr2}
        \widetilde{V}=\log(x_{0}) + \inf_{v\in\mathcal{D}_{2}}\mathbb{E}\bigg[\int_{0}^{0.5} \big(r(t)+\delta_{K}(v(t))+ \frac{1}{2} \big|\theta(t)+\sigma^{-1}(t)v(t)\big|^{2}\big) \hspace{0.05cm} dt\bigg].
    \end{equation}
    The infimum is attained by the deterministic process $v^{*}$ which is, for every $t\in [0,0.5]$, defined via
    \begin{equation}\label{eq:highdimlogutunconstr3}
        v^{*}(t):=\argmin_{v\in\widetilde{K}} \bigg\{\delta_{K}(v)+ \frac{1}{2} \big|\theta(t)+\sigma^{-1}(t)v\big|^{2} \bigg\}.
    \end{equation}
    Obviously, this argument is also applicable to every other closed, convex set $K$ which contains $0$. For the corresponding unconstrained problem, we obtain, as expected, $v^{*}\equiv 0$. We approximate the integral in \eqref{eq:highdimlogutunconstr2} using an equidistant time discretization consisting of $1000$ points. For each of these points, we solved the deterministic convex optimization problem \eqref{eq:highdimlogutunconstr3} by means of the CVXPY package in Python for our particular choice of $K$. After only a few seconds, we obtained $2.34335$ as our benchmark.\newline
    For both algorithms which tackle the primal problem directly, we implemented hard constraints by applying $x\mapsto (x^{2}-\kappa)$ componentwise to the outputs of $\mathcal{N}_{\theta_{i,\pi}}$ for every $i\in\{1,\dots,N-1\}$. Our numerical experiments have shown that it is favorable to use $x\mapsto \max\{-\kappa, x\}$ for $\pi_{0}$ instead. For the dual control process, we project onto $\widetilde{K}$ in a similar manner by means of the function $x\mapsto x^{2}$. Moreover, $h_{K}$ in the deep SMP algorithm is defined as the function which applies $x\mapsto \max\{-\kappa, x\}$ to every component. Our numerical experiments have shown that also $x\mapsto (x^{2}-\kappa)$ and $x\mapsto (|x|-\kappa)$ yield similar results. In the following, we select $N=10$ and apply the learning rate schedules as given in Table \ref{tab:tablehighdimlogutunconstrLR} in Appendix \ref{section:AppendixLRtables}. We quickly reduce the learning rate for the parameter $y$ since it converges at an extremely high speed to $y^{*}$ and its value significantly influences the implied value approximation. Potential issues with NaN-values were resolved by choosing $\varepsilon$ sufficiently large for all batch normalization layers. The Python code for the deep primal SMP algorithm adapted to this specific problem is provided in Appendix \ref{section:Appendix} below. \newline
    We ran each algorithm for 10000 training steps. The training progress is depicted in Figure \ref{fig:highdimlogutconstr} below. We observe fast convergence towards the theoretical benchmark, except for the upper bound $\widetilde{V}_{u}$ determined by the deep SMP algorithm. The reason for this behavior is that the obtained dual control was very close to the zero process. This illustrates a weakness of the deep SMP algorithm which has also been observed in Remark~4.3 of the first preprint version of \cite{zheng1} (see arXiv:2008.11757): Since $q_{i}$ is modeled without taking into account the potential dependency on $v_{i}$ (see \cite[Theorem 3.2]{zheng1} or \cite[Theorem 3.23]{KW2021deep}, which, in particular, uncover the structure of the desired integrand in the adjoint BSDE for Markovian problems), the zero control always minimizes the loss function for the dual control, which is motivated by \eqref{eq:DualnonMarkovcondition3}. Hence, the algorithm tends to mistake the zero control for the optimal dual control, which leads to an upper bound as in the unconstrained case, where the theoretical benchmark is given by $2.35058$. This might be the reason, why in \cite{zheng1} the deep SMP algorithm was applied exclusively to unconstrained problems or problems with positive market price of risk and $K=(\mathbb{R}_{0}^{+})^{m}$. The lower bound is not exposed to this issue as it does not depend explicitly on the dual control (see Equations (25), where \eqref{eq:DualnonMarkovcondition3} is applied to a discretized version of \eqref{eq:SDEsystemdualnonMarkov}, and (26) of \cite{zheng1}). Here, the constraints are incorporated by the function $h_{K}$. Hence, we can still use $\widetilde{V}_{l}$ as a reliable estimate. In this example, it is even more accurate than the value determined by the dual version of the deep controlled 2BSDE algorithm, as Figure \ref{fig:highdimlogutconstr} and \eqref{eq:highdimlogutconstr1} show. 
    \begin{figure}[ht]
        \centering
        \includegraphics[width=15cm]{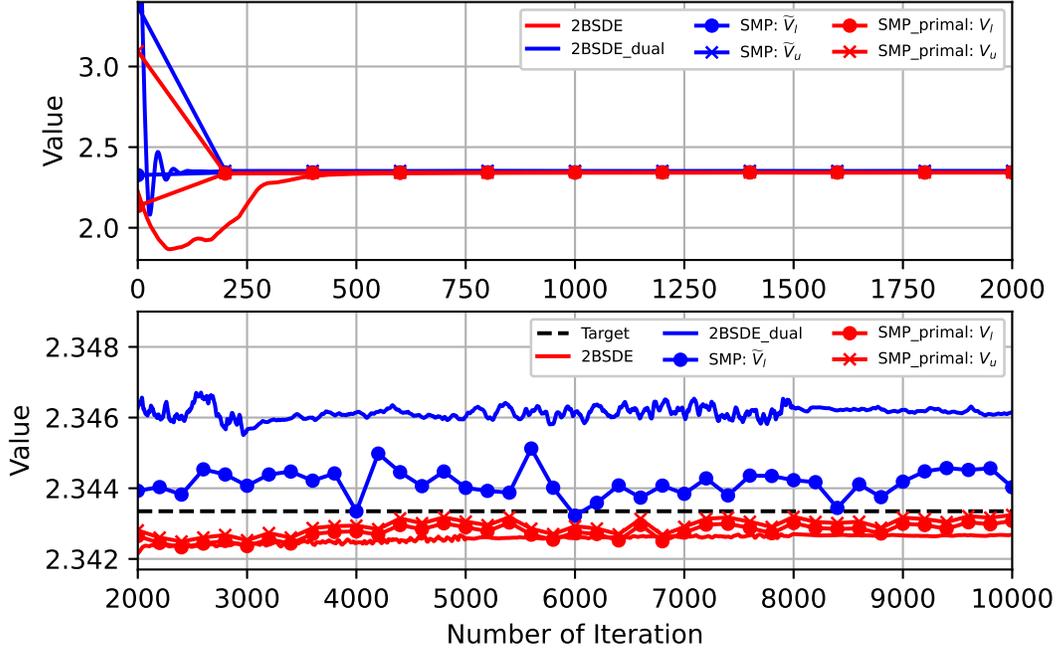}
        \caption{Value approximations for the constrained Markovian problem with logarithmic utility function studied in Example \ref{example:highdimlogutconstr} in the course of 10000 training steps.}
        \label{fig:highdimlogutconstr}
    \end{figure}
    \noindent At the end of the training procedure, we obtain
    \begin{equation}\label{eq:highdimlogutconstr1}
    \begin{aligned}
        &V = 2.34268,\hspace{0.3cm}\widetilde{V}= 2.34614,\hspace{0.3cm} \widetilde{V}_{l}= 2.34403,\hspace{0.3cm} \widetilde{V}_{u}= 2.35192, \\
        &V_{l}= 2.34308\hspace{0.3cm}\mbox{and}\hspace{0.3cm} V_{u}= 2.34324.
    \end{aligned}
    \end{equation}
    In conclusion, we gather from the above observations that both bounds implied by the deep primal SMP algorithm remain highly accurate for constrained problems, whereas the upper bound determined by the deep SMP algorithm tends to attain the value of the associated unconstrained problem. The deep primal SMP algorithm outperforms even both versions of the deep controlled 2BSDE algorithm quite significantly, even though the latter algorithm is specifically designed for Markovian problems.
\end{example}
\subsection{Non-Markovian Utility Maximization Problems: Path Dependent Coefficients}\label{subsection:trulynonmarkovcoeff}
In this subsection, we study a utility maximization problem with a coefficient process $\mu$ which depends on the paths of the stocks. Hence, we cannot apply the deep controlled 2BSDE algorithm. However, we can still use both SMP-based algorithms and compare their outcomes.
\begin{example}\label{example:randomcoeffunconstrained}
    We consider an unconstrained problem with five stocks, i.e. $K=\mathbb{R}^{5}$. Moreover, we select $U=2\sqrt{\cdot}$, $x_{0}=1$, $T=0.5$, $r\equiv 0.1$ and $\sigma$ as a constant matrix with diagonal elements equal to $0.2$ and the remaining entries are set to $0.05$. For every $i\in\{1,\dots, 5\}$ and $t\in (0,0.5]$, we define
    \begin{equation}\label{eq:murandomcoeff}
        \mu_{i}(t) := \begin{cases}
        \mu_{high} &\text{for $S_{i}(t)\ge \frac{1}{t} \int_{0}^{t} S_{i}(s)\hspace{0.05cm}ds$,}\\
        \mu_{low} &\text{else}.
        \end{cases}
    \end{equation}
    Note that the expression on the right-hand side is well-defined due to the continuity of $S_{i}$. Its limit for~$t\searrow 0$ exists almost surely and is given by $S_{i}(0)$, which follows again from the fact that $S_{i}$ has continuous paths. Hence, we can define $\mu_{i}(0):=\mu_{high}$. For our numerical experiments, we choose $\mu_{low}=0.08$ and $\mu_{high}=0.12$. As a motivation for \eqref{eq:murandomcoeff}, we assume that a significant share of investors tends to make trade decisions based on trends, e.g. if the number of retail investors in the market is exceptionally high. If a chart looks ``good'', i.e. the stock is rising, but not straight out of a historic low, then the stock is considered a buy due to its momentum. Hence, the demand rises which strengthens the upward trend. We assume that the opposite happens, if the stock price is lower than its historic mean. Moreover, we consider the historic mean over $[0,t]$ instead of a moving average for a fixed time length as $T=0.5$ is rather small. \newline
    For our numerical implementation, we take $N=10$ while applying the learning rate schedules given by Table \ref{tab:tablerandomcoeff} in Appendix \ref{section:AppendixLRtables}. Figure \ref{fig:randomcoeffunconstrained} shows the results of running both algorithms for 20000 steps, where we chose $\varepsilon=100$. This was necessary in order to prevent the lower bound provided by the deep primal SMP algorithm from becoming very large during the first half of the training procedure. 
    \begin{figure}[ht]
        \centering
        \includegraphics[width=15cm]{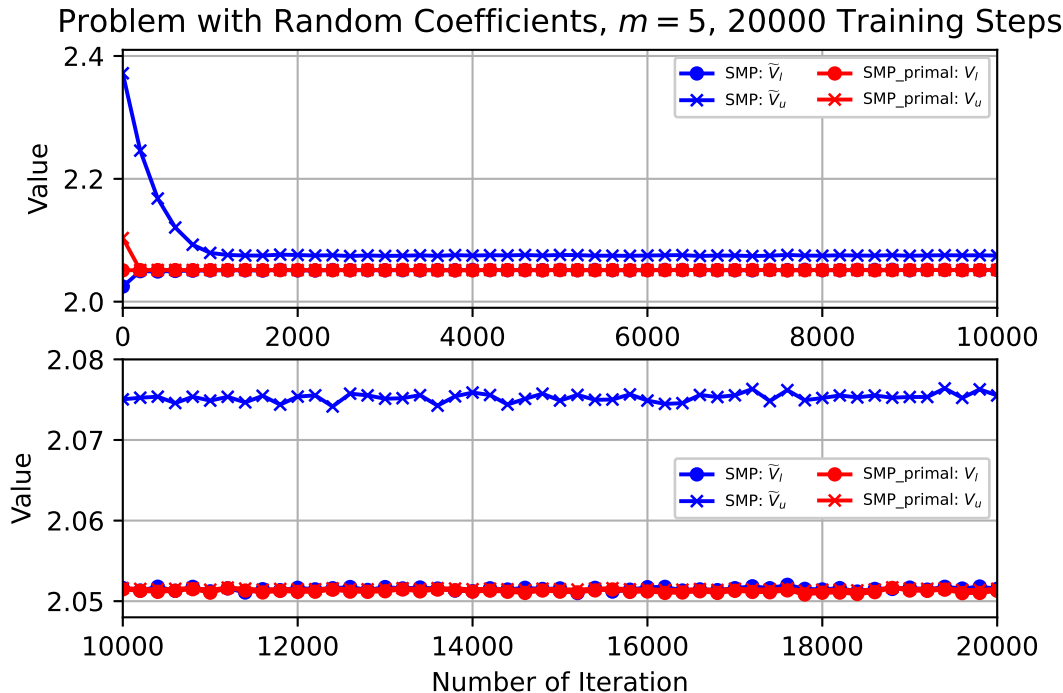}
        \caption{Value approximations for both SMP-based algorithms in the control problem with path dependent coefficient $\mu$ studied in Example \ref{example:randomcoeffunconstrained} in the course of 20000 training steps.}
        \label{fig:randomcoeffunconstrained}
    \end{figure}
    \noindent Since the non-Markovian structure of the problem complicates finding an explicit solution significantly, we content ourselves with comparing the bounds implied by the algorithms. As in the example above, we observe that most bounds only need a relatively low amount of training steps for arriving at reasonable estimates. This is in particular the case for both bounds implied by the deep primal SMP algorithm. After just $400$ training steps, we obtain $V_{l}=2.0513$ and $V_{u}=2.0516$, which almost agrees with the final result given in \eqref{eq:estimatesrandomcoeff} below. At the end of the training procedure, our four estimates are given by
    \begin{equation}\label{eq:estimatesrandomcoeff}
        \widetilde{V}_{l}= 2.0515,\hspace{0.3cm} \widetilde{V}_{u}= 2.0755,\hspace{0.3cm} V_{l}= 2.0513 \hspace{0.3cm}\mbox{and}\hspace{0.3cm} V_{u}= 2.0515.
    \end{equation}
     It is remarkable that there is a steady duality gap between $\widetilde{V}_{l}$ and $\widetilde{V}_{u}$. We investigated this result further by dividing the learning rates by 10 and running the deep SMP algorithm for 30000 additional steps. We obtained $\widetilde{V}_{l}=2.0514$ and $\widetilde{V}_{u}=2.0749$, i.e. approximately the same values as in \eqref{eq:estimatesrandomcoeff}, which supports the above observation. Moreover, we repeated the above experiment for $N=50$ and 20000 training steps, which resulted in $\widetilde{V}_{l}=2.0505$ and $\widetilde{V}_{u}=2.0758$, respectively, i.e. leading to the same result. Furthermore, working with $\mu_{high}=0.14$ instead increases the duality gap determined by the deep SMP algorithm to~$0.0532$, where the obtained bounds are given by $\widetilde{V}_{l}=2.0566$ and $\widetilde{V}_{u}=2.1098$. In contrast to this, the deep primal SMP algorithm yields $V_{l}=2.0565$ and $V_{u}=2.0568$, which agrees again with $\widetilde{V}_{l}$. Moreover, if we set $\mu_{low}=\mu_{high}$, the duality gap vanishes. Hence, the duality gap produced by the deep SMP algorithm might have a non-trivial connection with $\mu_{high}-\mu_{low}$. However, as the duality gap is negligible in the deep primal SMP algorithm, this behavior should not result from a potential violation of the strong duality property.
     In conclusion, we gather from \eqref{eq:estimatesrandomcoeff} and the above considerations that the value of our control problem should be approximately given by $2.0515$, as all bounds except $\widetilde{V}_{u}$ suggest. 
\end{example}
\subsection{Non-Markovian Utility Maximization Problems: Coefficients Satisfying Their Own SDEs}\label{subsection:nonmarkovcoeffwithSDE}
In this subsection, we consider problems with coefficients which satisfy their own SDEs. As in Subsection~\ref{subsection:trulynonmarkovcoeff}, this implies that the wealth process is not of a one-dimensional Markovian form. However, as we shall see below, we can still apply the deep controlled 2BSDE algorithm. The problem can be converted into a Markovian problem by considering higher-dimensional state processes, where the additional components are exactly the random coefficients. \vspace{0.3cm}\newline
At first, we consider Heston's stochastic volatility model with one traded, risky asset. Here the dynamics of the stock price is given by
\begin{equation}\label{eq:HestonstockSDE}
    dS(t)=S(t)\hspace{0.05cm}(r+A\hspace{0.02cm}\nu(t))\hspace{0.05cm}dt + S(t)\sqrt{\nu(t)}\hspace{0.05cm}dB^{1}(t), \hspace{0.3cm}t\in [0,T],
\end{equation}
where the parameters $r,A\in\mathbb{R}^{+}$ can be interpreted as the risk-free interest rate, which is assumed to be constant, and the market price of risk, respectively. Note that we have $\theta=A\sqrt{\nu}$ and $\sigma=\sqrt{\nu}$ following the notation from Section \ref{section:UMaxandDual}. The process $\nu$ satisfies
\begin{equation}\label{eq:HestonvolSDE}
    d\nu(t)=\kappa(\theta_{\nu}-\nu(t))\hspace{0.05cm}dt + \xi\sqrt{\nu(t)}\hspace{0.05cm}dB^{\nu}(t), \hspace{0.3cm}t\in [0,T],
\end{equation}
with initial condition $\nu(0)=\nu_{0}\in\mathbb{R}^{+}$. The Brownian motions $B^{\nu}$ and $B^{1}$ have correlation parameter $\rho\in [-1,1]$. Hence, we can write $B^{\nu}$ as $\rho \hspace{0.03cm}B^{1} + \sqrt{1-\rho^{2}}\hspace{0.03cm}B^{2}$, where $B^{2}$ is a standard Brownian motion which is independent from $B^{1}$. Therefore, $\big(B^{1},B^{2}\big)^{\intercal}$ can be viewed as a standard two-dimensional Brownian motion. We assume that the parameters $\kappa,\theta_{\nu},\xi\in\mathbb{R}^{+}$ satisfy the so-called Feller condition $2\kappa\theta_{\nu}>\xi^{2}$, which ensures that $\nu$ is strictly positive. Since $\nu$ is not necessarily progressively measurable with respect to the filtration generated by $B^{1}$, this market is not yet of the form discussed in Subsection \ref{subsection:marketmodel}. Hence, we introduce an artificial stock, whose local martingale part is driven by $B^{2}$:
\begin{equation}\label{eq:artificialstockheston}
    dS^{2}(t)=S^{2}(t)\hspace{0.05cm}r\hspace{0.05cm}dt + S^{2}(t)\hspace{0.05cm}dB^{2}(t), \hspace{0.3cm}t\in [0,T],
\end{equation}
which is excluded from trading by requiring that the admissible portfolio processes have to take values in~$K\times\{0\}$. Moreover, we have $\theta=(A\sqrt{\nu},0)^{\intercal}$ and $\sigma=\big(\begin{smallmatrix}
  \sqrt{\nu} & 0\\
  0 & 1
\end{smallmatrix}\big)$ in this two-dimensional setup. Since the second component of admissible controls necessarily has to be $0$, we obtain that $B^{2}$ influences the wealth process only through $\nu$. For notational convenience, we identify every control process with its first component. Hence, for each $\pi\in\mathcal{A}$, the wealth process $X^{\pi}$ satisfies 
\begin{equation}\label{eq:HestonwealthSDE}
    dX^{\pi}(t)=X^{\pi}(t)\hspace{0.05cm}(r+\pi(t)A\hspace{0.02cm}\nu(t))\hspace{0.05cm}dt + X^{\pi}(t)\hspace{0.05cm}\pi(t)\sqrt{\nu(t)}\hspace{0.05cm}dB^{1}(t), \hspace{0.3cm}t\in [0,T],
\end{equation}
where the initial wealth is given by a fixed number $x_{0}\in\mathbb{R}^{+}$. Note that the unboundedness of $\nu$ is not an issue here because its paths are continuous and its invertibility is ensured by the Feller condition. Hence, the deep primal SMP algorithm, i.e. Theorem \ref{theorem:PrimalnonMarkov} as its theoretical foundation, is applicable to this problem. \vspace{0.3cm}\newline
From the derivation of \eqref{eq:dualsde} and the above considerations, we conclude that for every $(y,v_{1},v_{2})\in\mathcal{D}$, $(v_{1},v_{2}):=v$, the associated dual state process satisfies $Y^{(y,v_{1},v_{2})}(0)=y$ and
\begin{equation}\label{eq:dualSDEHeston}
\begin{aligned}
    dY^{(y,v_{1},v_{2})}(t)=-Y^{(y,v_{1},v_{2})}(t)\big[&\big(r+\delta_{K}(v_{1}(t))\big)\hspace{0.05cm}dt + \big(A\sqrt{\nu(t)} + v_{1}(t)/\sqrt{\nu(t)}\big)\hspace{0.05cm}dB^{1}(t)\\
    &+v_{2}(t)\hspace{0.05cm}dB^{2}(t)\big], \hspace{0.3cm}t\in [0,T].
\end{aligned}
\end{equation}
In \cite{zheng1}, it is shown that also the deep SMP algorithm is applicable to the above problem, where essentially only a one-dimensional dual control process has to be learned. As already mentioned above, also the deep controlled 2BSDE algorithm can be applied here by increasing the dimension of the state processes. To be more precise, we consider the processes $\big(X^{\pi},\nu\big)^{\intercal}$ and $\big(Y^{(y,v_{1},v_{2})},\nu\big)^{\intercal}$ instead of just $X^{\pi}$ and $Y^{(y,v_{1},v_{2})}$, respectively. We refer to \cite{zheng1} for details. \vspace{0.3cm}\newline
In the following example, we study an unconstrained problem with power utility function for various terminal times $T$. This choice is insofar appealing, as explicit solutions are known in this case, which allows us to assess the quality of the obtained estimates.
\begin{example}\label{example:Hestonunconstr}
    We consider Heston's stochastic volatility model, as introduced above. Moreover, we assume that trading is not restricted, i.e. $K=\mathbb{R}$. Motivated by Example 5.1 from \cite{zheng3}, we choose $r=0.05$, $A=0.5$, $\kappa=10$, $\theta_{\nu}=0.05$, $\xi=0.5$, $\rho=-0.5$, $x_{0}=1$, $v_{0}=0.5$ and the power utility function with parameter $p=0.5$. Clearly, the Feller condition is satisfied. However, as the process can still become negative due to the SDE discretization error, we truncate $(\nu_{i})_{i\in\{0,\dots,N\}}$ at $\delta:=10^{-5}$. We consider three different values for $T$, namely $0.2$, $0.5$ and $1$. As indicated above, the value of an unconstrained power utility maximization problem in a Heston stochastic volatility setting can be found explicitly (see \cite{hestonkraft}). We refer also to \cite{zheng3} for details, where the HJB equation of the dual problem is solved by means of an appropriate ansatz which yields a system of two Riccati equations. Moreover, it is proved that the duality gap is indeed zero. We refer to Table \ref{tab:tablehestonunconstrresults} for the exact benchmark values for all three parameter configurations. \newline
    For our numerical experiments, we choose $N$ such that $T/N$ is constant for all three configurations (cf. Table~\ref{tab:tablehestonunconstrresults}). The learning rate schedules are given in Table \ref{tab:tablehestonunconstrLR} below. Note that we chose rather conservative learning rate schedules for both algorithms which tackle the primal problem directly, as they converged exceptionally fast. This holds except for the deep controlled 2BSDE algorithm and $T=1$, where too high learning rates even led to the divergence of the algorithm. Choosing values for the hyperparameter $\varepsilon$ which are larger than its default value was again highly favorable (see Remark \ref{remark:BNepsilon}). For both algorithms tackling the primal problem we selected $\varepsilon=1$ and for the remaining two methods $\varepsilon=100$. We ran each algorithm for $10000$ training steps. The final results are provided in the following table.
    \begin{table}[ht]
    \centering
    \begin{tabular}{|c||c||c||c|c||c|c||c|}
        \hline
         $(T,N)$ & $V$ & $\widetilde{V}$ & $\widetilde{V}_{l}$ & $\widetilde{V}_{u}$ & $V_{l}$ & $V_{u}$ & Benchmark \Tstrut\Bstrut\\
         \hline \hline
          $(0.2,6)$ & 2.02269 & 2.02278 & 2.02246 & 2.02285 & 2.02228 & 2.02296 & 2.02225 \\
          $(0.5,15)$ & 2.04237 & 2.04284 & 2.04249 & 2.04301 & 2.04283 & 2.04298 & 2.04268\\
          $(1,30)$ & 2.06609 & 2.07904 & 2.07431 & 2.07525 & 2.07474 & 2.07705 & 2.07484\\
          \hline
    \end{tabular}
    \caption{Value approximations for the studied unconstrained problem with stochastic volatility at the end of the training procedure for various pairs $(T,N)$ in comparison with the corresponding theoretical benchmarks.}
    \label{tab:tablehestonunconstrresults}
    \end{table}
    \noindent We observe that all value approximations, except~$V$ for $(T,N)=(1,30)$, are highly accurate. This outlier is most likely not a consequence of unfortunate initial guesses as repeating the experiment several times yielded approximately the same value. Interestingly enough, reducing $N$ to $5$ resulted in $V=2.07498$ which would correspond to the second best approximation in the last row of Table \ref{tab:tablehestonunconstrresults}. Hence, this result could be used for verifying the results of the other algorithms, if no theoretical benchmark was known. A trader can still use the control process determined by the deep primal SMP algorithm as it is a function of the wealth process, as opposed to the optimal primal control implied by the remaining two algorithms via \eqref{eq:fromdualtoprimal1}, which is a function of the non-observable dual state process. This illustrates the luxury of having now four independent methods for solving utility maximization problems, two of which even yielding practicable strategies. Moreover, we conclude from Table \ref{tab:tablehestonunconstrresults}, in particular from the last row, that increasing $T$ while keeping $T/N$ constant usually increases the approximation error. In \cite{zheng1}, a similar observation is made for the deep controlled 2BSDE algorithm and the classical unconstrained Merton problem. Motivated by Remark \ref{remark:newqSMP} below, we refined both SMP-based algorithms by also using the current volatility as an input for the neural networks $\mathcal{N}_{\theta_{i,q}}$, $i\in\{1,\dots,N-1\}$. Table \ref{tab:tablehestonunconstrresults} suggests that this is a worthwhile refinement, as $V_{l}$ corresponds to the best value approximation for all of the considered pairs $(T,N)$. Figure \ref{fig:hestonunconstr} depicts the evolution of the value approximations during the training procedure for $(T,N)=(0.5,15)$.
    \begin{figure}[ht]
        \centering
        \includegraphics[width=15cm]{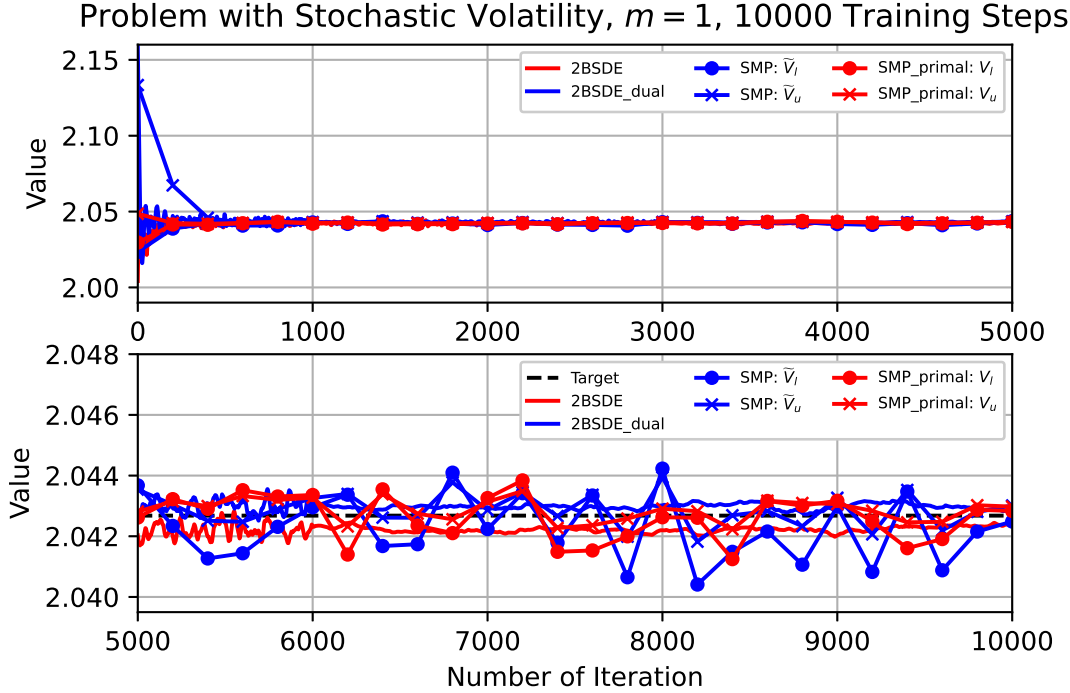}
        \caption{Value approximations for the unconstrained problem with stochastic volatility studied in Example \ref{example:Hestonunconstr} in the course of 10000 training steps for $(T,N)=(0.5,15)$.}
        \label{fig:hestonunconstr}
    \end{figure}
    \noindent As in both of the previous examples, we observe that it only takes $1000$ training steps for arriving at reasonable estimates. Finally, we intend to assess the control approximation quality for both algorithms tackling the primal problem directly. We choose again $(T,N)=(0.5,15)$ as the deep controlled 2BSDE algorithm did not even yield a satisfying value approximation for $(T,N)=(1,30)$. In contrast to this, we observed that the control process determined by the deep primal SMP algorithm boasts a similar accuracy as the one obtained for the problem with $(T,N)=(0.5,15)$.\newline
    Figure \ref{fig:hestonunconstrcontrol} illustrates the control processes determined by the deep controlled 2BSDE algorithm and the deep primal SMP algorithm, respectively. For this purpose, the neural networks $\mathcal{N}_{\theta_{i,\pi}}$ are depicted on $[0.7,1.4]\times [0,0.25]$ for several indices $i$. We recall that the neural networks $\mathcal{N}_{\theta_{i,\pi}}$ take only a one-dimensional input in the deep primal SMP algorithm. Hence, the corresponding surfaces in Figure \ref{fig:hestonunconstrcontrol} are constant in $\nu$. For each subplot and corresponding $t_{i}$, the domain of the color map is centered around the deterministic random variable $\pi^{*}(t_{i})$. We refer to \cite{hestonkraft, zheng3} for a derivation of the optimal control process $\pi^{*}$, which proves in particular that $\pi^{*}(t_{i})$ is indeed constant. For our parameter choice, $\big\{\pi^{*}(t_{i})\hspace{0.05cm}|\hspace{0.05cm}i\in\{0,\dots,14\}\big\}$ lies in~$[0.9938,0.9983]$. Hence, all the surfaces should lie in a similar region, as we shall see below. 
    \begin{figure}[ht]
        \centering
        \includegraphics[width=15cm]{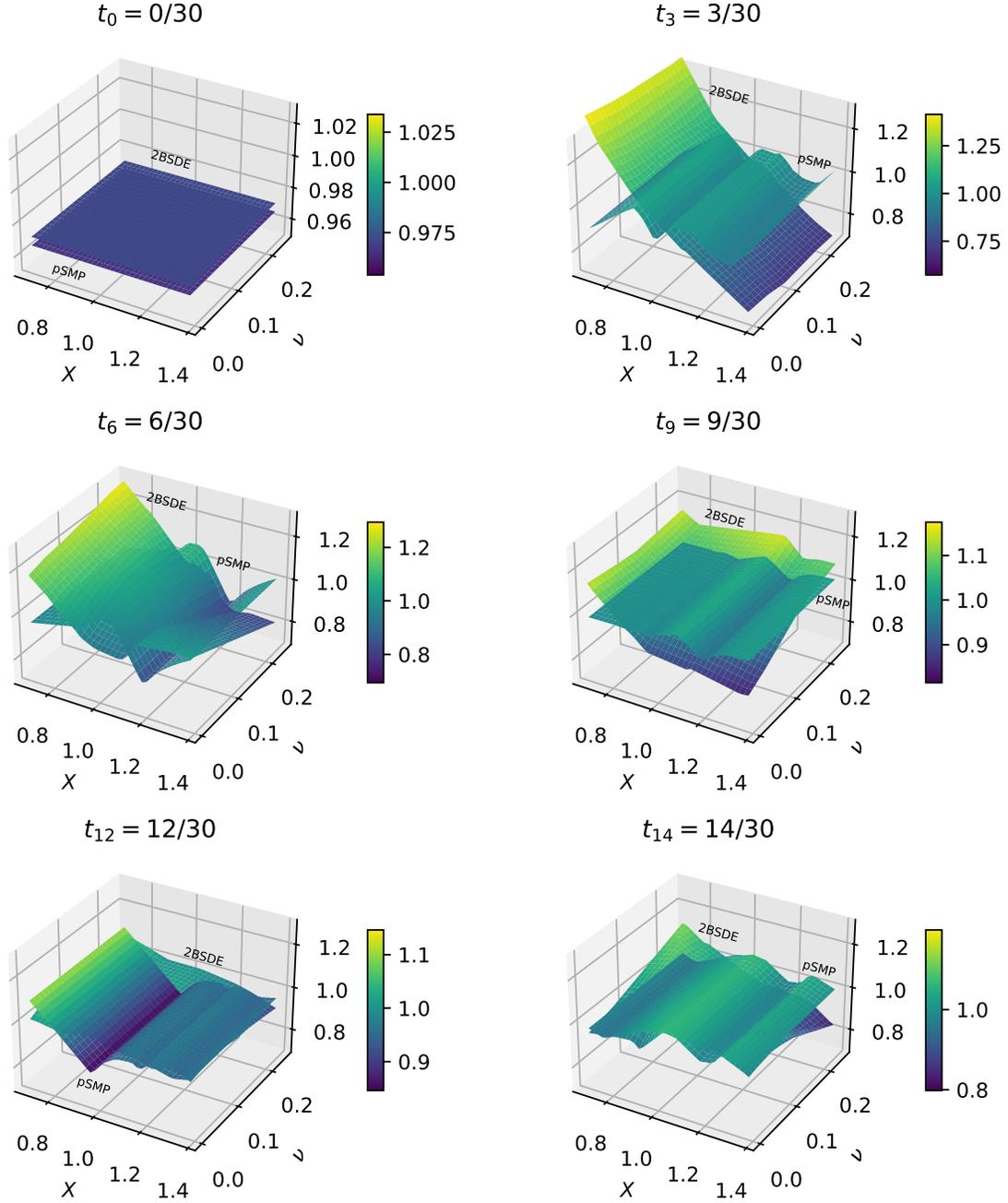}
        \caption{Neural networks $\mathcal{N}_{\theta_{i,\pi}}$ for the unconstrained problem with stochastic volatility studied in Example~\ref{example:Hestonunconstr} with $(T,N)=(0.5,15)$ for several indices $i$ and both algorithms tackling the primal problem directly.}
        \label{fig:hestonunconstrcontrol}
    \end{figure}
    \noindent We observe that the approximation quality is good for both algorithms. This holds in particular for large indices~$i$. The relatively large approximation errors for small indices and pairs whose first component is farther away from $1$ have to be assessed considering the fact that the wealth at time $t_{i}$ is close to $x_{0}=1$ with high probability. Note, however, that the deep primal SMP algorithm achieves a significantly better approximation for these indices. Finally, in Table \ref{tab:tablehestonunconstrcontrol} we compare the maximal absolute errors with respect to the exact solution on~$[0.7,1.4]\times [0,0.25]$.
    \begin{table}[ht]
    \centering
    \addtolength{\leftskip}{-2cm}
    \addtolength{\rightskip}{-2cm}
    \begin{tabular}{|c||c|c|c|c|c|c|}
        \hline
         Algorithm & $i=0$ & $i=3$ & $i=6$ & $i=9$ & $i=12$ & $i=14$ \Tstrut\Bstrut\\
         \hline \hline
          2BSDE & 0.02232 & 0.41483 & 0.28734 & 0.15769 & 0.08373 & 0.16496  \\
          SMP{\_}primal & 0.02771 & 0.09208 & 0.14719 & 0.03757 & 0.14283 & 0.06164 \\
          \hline
    \end{tabular}
    \caption{Maximal deviations of the neural networks $\mathcal{N}_{\theta_{i,\pi}}$ from the deterministic optimal control process at $t_{i}$ on $[0.7,1.4]\times [0,0.25]$, i.e. $\| \pi^{*}(t_{i})-\mathcal{N}_{\theta_{i,\pi}}\|_{\infty}$, for the same indices $i$ as in Figure \ref{fig:hestonunconstrcontrol}.}
    \label{tab:tablehestonunconstrcontrol}
    \end{table}
    \noindent As in Figure \ref{fig:hestonunconstrcontrol}, we observe that the deep primal SMP algorithm yields a more accurate control process than the deep controlled 2BSDE algorithm.
\end{example}
\begin{remark}\label{remark:newqSMP}
    Extending the input vectors of the neural networks $\mathcal{N}_{\theta_{i,q}}$ in both SMP-based algorithms by~$\sqrt{\nu_{i}}$ is motivated by \cite[Theorem 3.2]{zheng1} and \cite[Theorem 3.23]{KW2021deep}, which uncover in particular the structure of the desired integrand in the adjoint BSDE for this two-dimensional Markovian problem. This refinement improved the results in all of our numerical experiments. Moreover, this observation can be seen as a motivation for adapting the algorithms. Potential refinements of the deep primal SMP algorithm are discussed in Section \ref{section:refining}. 
\end{remark}
\noindent Finally, we consider our general utility maximization setup, where $\mu$ and $\sigma$ are deterministic and the short rate $r$ is, as in the Vasicek model, given by an Ornstein-Uhlenbeck process. Hence, the dynamics of $r$ is given~by
\begin{equation}\label{eq:VasicekrSDE}
    dr(t)=\alpha(\beta-r(t))\hspace{0.05cm}dt + \gamma\hspace{0.05cm}dB^{m+1}(t), \hspace{0.3cm}t\in [0,T],
\end{equation}
with initial condition $r(0)=r_{0}\in\mathbb{R}$ and (positive) real-valued parameters $\alpha$, $\beta$ and $\gamma$. Moreover, we assume in Example \ref{example:vasicekconstrhighdim} that the Brownian motions $B^{m+1}$ and $B$ are independent. Hence, $\overline{B}:=(B,B^{m+1})^{\intercal}$ can be viewed as a standard $(m+1)$-dimensional Brownian motion. Since $r$ is not progressively measurable with respect to the filtration generated by $B$, we introduce an additional stock (cf. \eqref{eq:artificialstockheston} and the market completion procedure for the Heston model), whose local martingale part is driven by $B^{m+1}$:
\begin{equation}\label{eq:artificialstockvasicek}
    dS^{m+1}(t)=S^{m+1}(t)\hspace{0.05cm}r(t)\hspace{0.05cm}dt + S^{m+1}(t)\hspace{0.05cm}dB^{m+1}(t), \hspace{0.3cm}t\in [0,T].
\end{equation}
As in our considerations for the Heston model, this stock cannot be traded by the investor, which ensures that the value of the problem remains unchanged under this generalization. Hence, we have a market of the form as introduced in Section \ref{section:UMaxandDual} which is driven by $\overline{B}$. Therefore, both SMP-based algorithms are applicable, where also the dual control process is again essentially $m$-dimensional. For notational convenience we, therefore, identify the processes $\pi$ and $v$ with their first $m$ components in the following. Like in the stochastic volatility setting, also the deep controlled 2BSDE algorithm can be applied here by increasing the dimension of the state processes, i.e. we consider $\big(X^{\pi},r\big)^{\intercal}$ and $\big(Y^{(y,v)},r\big)^{\intercal}$ instead of $X^{\pi}$ and $Y^{(y,v)}$, respectively.
\begin{example}\label{example:vasicekconstrhighdim}
    We consider a $30$-dimensional problem, where short selling is not permitted, i.e. $K=(\mathbb{R}_{0}^{+})^{30}$. Moreover, we choose $U=\log$, $x_{0}=10$ and $T=0.5$. The parameters for the short rate process are selected as $r_{0}=0.05$, $\alpha=5$, $\beta=0.05$ and $\gamma=0.05$. The processes $\mu$ and $\sigma$ are assumed to be deterministic. For every~$t\in [0,0.5]$ and $i,j\in\{1,\dots,30\}$, we choose $\mu_{i}(t)=0.06+0.01\hspace{0.01cm}\sin(4\pi t + \pi i/15)$, $\sigma_{i,i}(t)=0.3/(1+t)$ and $\sigma_{i,j}(t)=0.05$, if $j\neq i$ holds. For our numerical experiments, we select $N=20$. We ensure that the control processes map to $K$ by applying the function $x\mapsto x^{2}$ componentwise to the outputs of the final dense layers of the corresponding neural networks. Moreover, the same function serves as projection $h_{K}$ in the deep SMP algorithm. Since $\widetilde{K}=K$ holds for this particular problem, we choose the same projection function for the dual control processes. The learning rate schedules are given in Table~\ref{tab:tablevasicekconstrLR} below. As in many of our previous examples, choosing larger values for $\varepsilon$, e.g. $\varepsilon=1$ or $\varepsilon=100$, is highly favorable with regards to the convergence of the studied algorithms. We ran each algorithm for 10000 training steps. The evolution of the implied value approximations in the course of this procedure is depicted in Figure~\ref{fig:vasicekconstrvalueapprox} below.
    \begin{figure}[ht]
        \centering
        \includegraphics[width=15cm]{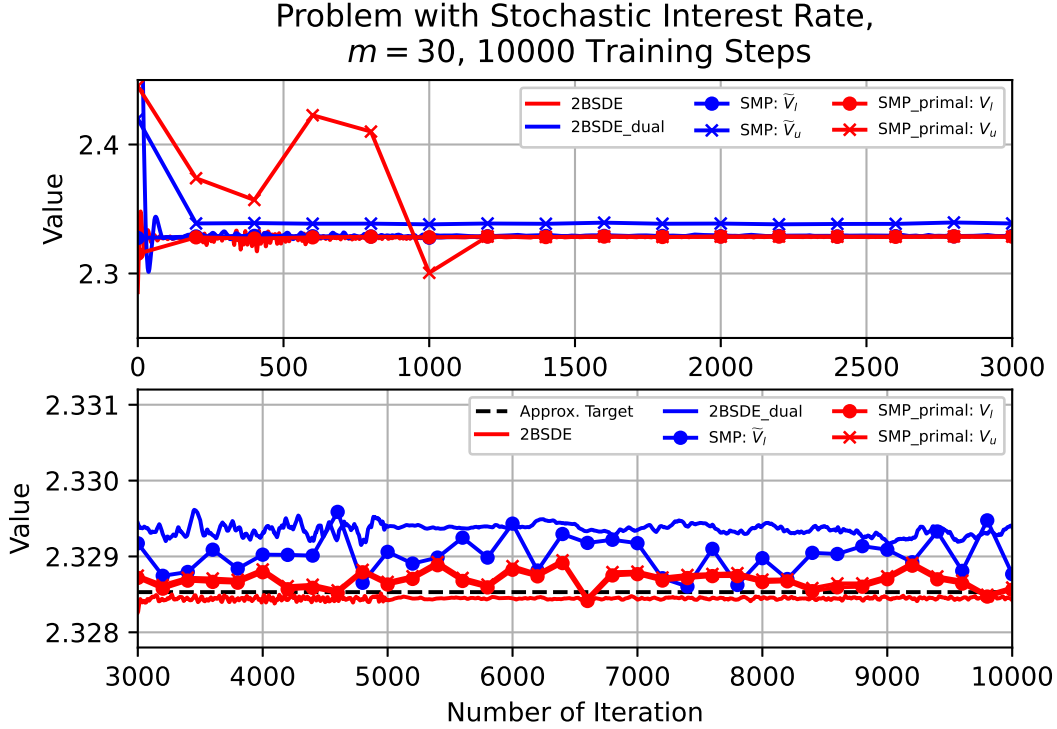}
        \caption{Value approximations for the high-dimensional, constrained problem with stochastic interest rate studied in Example \ref{example:vasicekconstrhighdim} in the course of 10000 training steps.}
        \label{fig:vasicekconstrvalueapprox}
    \end{figure}
    \noindent We observe that most bounds converge astonishingly fast. The initial variability of $V_{u}$ can be explained by the fact that $-p$ is used as an estimator for the parameter $y$ from the dual problem (cf.~\eqref{eq:boundsprimaldeepSMP}). Since $p$ is optimized with respect to $\mathcal{L}_{BSDE}$, its initial learning rate is $10^{-2}$ according to Table \ref{tab:tablevasicekconstrLR}, whereas the learning rate schedule for $y$ starts at $10^{-3}$ for both algorithms tackling the dual problem. Once the learning rate is reduced at training step $1000$, $V_{u}$ closes the duality gap with respect to $V_{l}$. In contrast to this, the deep SMP algorithm yields a steady duality gap which does not decrease as the training procedure progresses. While $\widetilde{V}_{l}$ remains close to the other value approximations, $\widetilde{V}_{u}$ yields a significantly higher value. Interestingly enough, applying the other algorithms to the corresponding unconstrained problem yields values which are very close to $\widetilde{V}_{u}$. Hence, this is another example, where the deep SMP algorithm mistakes $v\equiv 0$ for the optimal dual control. At the end of the training procedure, we obtain the following estimates for our constrained problem:
    \begin{equation}\label{eq:vasicekconstr1}
    \begin{aligned}
        &V = 2.32842,\hspace{0.3cm}\widetilde{V}= 2.32939,\hspace{0.3cm} \widetilde{V}_{l}= 2.32877,\hspace{0.3cm} \widetilde{V}_{u}= 2.33868, \\
        &V_{l}= 2.32856\hspace{0.3cm}\mbox{and}\hspace{0.3cm} V_{u}= 2.32859.
    \end{aligned}
    \end{equation}
    As all values except $\widetilde{V}_{u}$ lie within an interval of length $10^{-3}$, we are highly confident that also the true value of the problem lies within this region. Moreover, we solved the corresponding problem with deterministic short rate $r\equiv r_{0}$ by means of the machinery described in Example \ref{example:highdimlogutconstr}, which led to $2.32853$. While this is certainly not an exact benchmark, it emphasizes that the values in \eqref{eq:vasicekconstr1} are plausible.
\end{example}

\section{Two Possible Ways of Improving the Deep Primal SMP Algorithm}
\label{section:refining}
Modeling each $q_{1}(t_{i})$ as a function of the current wealth (cf. \eqref{eq:networkprimalSMP}) is motivated by \cite[Theorem 3.2]{zheng1} and \cite[Theorem 3.23]{KW2021deep}, which are formulated for Markovian problems and uncover, in particular, the structure of the solution to the adjoint BSDE. The time dependence is insofar incorporated as we consider a different neural network for each $t_{i}$. Moreover, using a similar ansatz for the control process leading to a Markovian control is inspired by \cite[Theorem 11.2.3]{oksendal}, which states (under additional assumptions) that Markovian controls cannot be outperformed by more general progressively measurable processes. However, since the problems solved by the deep primal SMP algorithm are not necessarily Markovian, there is still room for improvements as far as the network architectures are concerned. For example, we have already witnessed in Example \ref{example:Hestonunconstr} that adding $\sqrt{\nu_{i}}$ as an input for the neural network $\mathcal{N}_{\theta_{i,q}}$ for each $i\in\{1,\dots,N-1\}$ improves the obtained value approximations. Clearly, this phenomenon is not surprising since the Heston model has a two-dimensional Markovian structure (see also Remark \ref{remark:newqSMP}). \vspace{0.3cm}\newline
We start this section by first discussing an alternative approach for generating and using the training data. Motivated by the above considerations, we will then also discuss more general network architectures by introducing a semi-recurrent structure connecting each neural network to the one associated with the previous point in time. We assess the success of these changes by reconsidering Examples \ref{example:highdimlogutconstr} and \ref{example:Hestonunconstr}, for which exact theoretical benchmarks are known. 
\subsection{Introducing a Learning Procedure Based on Epochs}
\label{subsection:epochenverallg}
Since the algorithms defined in \cite{zheng1, jentzenbsde} served as inspiration for the formulation of the deep primal SMP algorithm, we inherited their approach for generating and managing the training data. This essentially means that in each training step a mini-batch is sampled, which will then be discarded. Hence, the algorithm works with each Brownian sample path almost surely at most once. However, it is a commonly used concept in deep learning that the algorithm learns an a priori fixed training data set several times in the course of the training procedure. The amount of iterations of the entire dataset corresponds to the so-called number of epochs. In addition to that, \eqref{eq:PrimalnonMarkovcondition1} is a pointwise maximization condition, which favors the second approach as each path is considered multiple times during the procedure.  Hence, this subsection is devoted to answering, whether it is more favorable to study as many paths as possible or certain paths more intensively. \vspace{0.3cm}\newline
For each problem, we apply 10000 training steps, which means that $640000$ paths can be considered as we have chosen $b_{size}=64$. In contrast to the previous sections, we sample the entire training data, i.e. the Brownian sample paths, at the beginning of the procedure. For example, if we wish that each path is learned 10 times, then our dataset has to consist of 64000 paths. Moreover, we ensure that the sample paths' order is the same for each epoch. Hence, as 10000 is a multiple of the number of epochs, also the batch composition is epoch-invariant. This is highly beneficial as it mitigates the simplification of converting the pathwise condition \eqref{eq:PrimalnonMarkovcondition1} into the sample-mean-based loss function \eqref{eq:primalSMPcontrolloss}. In order to assess the effect of the above innovation, we reconsider the problems discussed in Examples \ref{example:highdimlogutconstr} and \ref{example:Hestonunconstr} while varying the number of epochs. We apply again the learning rate schedules given in Tables \ref{tab:tablehighdimlogutunconstrLR} and~\ref{tab:tablehestonunconstrLR}. The results are summarized in Tables \ref{tab:tablehighdimlogutepochs} and \ref{tab:tableHestonepochs} below.
\begin{table}[ht]
    \centering
    \begin{tabular}{|c||c|c||c|c||c|}
        \hline
         No. Epochs & $V_{l}$ & $|V_{l}-V|$ & $V_{u}$ & $|V_{u}-V|$ & $\max\{|V_{u}-V|, |V_{l}-V|\}/V$ \Tstrut\Bstrut\\
         \hline \hline
          $1$ & 2.34308 & 0.00027 & 2.34324 & 0.00011 & 0.00012\\
          $5$ & 2.34309 & 0.00026 & 2.34326 & 0.00009 & 0.00011\\
          $10$ & 2.34317 & 0.00018 & 2.34333 & 0.00002 & 0.00008\\
          $50$ & 2.34329 & 0.00006 & 2.34349 & 0.00014 & 0.00006\\
          $100$ & 2.34312 & 0.00023 & 2.34332 & 0.00003 & 0.00010\\
          $500$ & 2.34321 & 0.00014 & 2.34338 & 0.00003 & 0.00006\\
          $1000$ & 2.34326 & 0.00009 & 2.34345 & 0.00010 & 0.00004\\
          \hline
    \end{tabular}
    \caption{Value approximations, absolute and relative errors at the end of the training procedure for the problem studied in Example \ref{example:highdimlogutconstr} in comparison with the theoretical benchmark $V=2.34335$ for various choices of the number of epochs.}
    \label{tab:tablehighdimlogutepochs}
    \end{table}
    \begin{table}[ht]
    \centering
    \begin{tabular}{|c||c|c||c|c||c|}
        \hline
         No. Epochs & $V_{l}$ & $|V_{l}-V|$ & $V_{u}$ & $|V_{u}-V|$ & $\max\{|V_{u}-V|, |V_{l}-V|\}/V$ \Tstrut\Bstrut\\
         \hline \hline
          $1$ & 2.04283 & 0.00015 & 2.04298 & 0.00030 & 0.00015\\
          $5$ & 2.04258 & 0.00010 & 2.04286 & 0.00018 & 0.00009\\
          $10$ & 2.04238 & 0.00030 & 2.04283 & 0.00015 & 0.00015\\
          $50$ & 2.04263 & 0.00005 & 2.04295 & 0.00027 & 0.00013\\
          $100$ & 2.04253 & 0.00015 & 2.04287 & 0.00019 & 0.00009\\
          $500$ & 2.04264 & 0.00004 & 2.04298 & 0.00030 & 0.00015\\
          $1000$ & 2.04248 & 0.00020 & 2.04309 & 0.00041 & 0.00020\\
          \hline
    \end{tabular}
    \caption{Value approximations, absolute and relative errors at the end of the training procedure for the problem studied in Example \ref{example:Hestonunconstr} with $(T, N)=(0.5, 15)$ in comparison with the theoretical benchmark $V=2.04268$ for various choices of the number of epochs.}
    \label{tab:tableHestonepochs}
    \end{table}
    \noindent The number of training steps per epoch can be easily obtained by dividing 10000 by the number of epochs. The quality of the results is measured by the maximum of the relative errors of both bounds as it would be unclear which bound is the most accurate, if no exact benchmark was known. Both tables show that learning the same paths multiple times can improve the approximation quality even further. However, in both cases we observed that the error is not inversely proportional to the number of epochs meaning that considering the same mini-batch for each training step even leads to the divergence of the algorithm. Hence, learning only from 64 different paths (i.e. No. Epochs = 10000) during the entire training procedure is certainly too extreme, whereas studying 640 paths might be worthwhile. In Table \ref{tab:tablehighdimlogutepochs}, this configuration even yields the smallest error. On the other hand, this choice produces the worst bounds in Table \ref{tab:tableHestonepochs}. Hence, the optimal choice of the number of epochs depends certainly on the specific problem, which is why we suggest treating the number of epochs as a hyperparameter, starting with a medium value and trying out several choices. As we have seen above, this approach turns out to be a powerful generalization of the classical version of the algorithm.
    
\subsection{Modeling the Control Process and the BSDE Integrand via a Semi-Recurrent Network Architecture}
\label{subsection:rekurrentverallg}
In the original formulation of the deep primal SMP algorithm, we chose to model the time-discretizations of the processes $\pi$ and $q_{1}$ by means of $N$ independent neural networks each. In the following, we aim at interconnecting each neural network with the corresponding network of the previous point in time. This semi-recurrent network architecture is motivated by \cite{deephedging}, where this structure yields promising results in the context of hedging a portfolio of derivatives. We are going to test such an architecture for each of the processes $\pi$ and $q_{1}$. Since the quotient $T/N$ is usually small, whence $|X_{i}-X_{i-1}|$ is small with high probability, and neural networks are continuous functions, the most important data points provided by the neural networks $\mathcal{N}_{\theta_{i-1,\pi}}$ and $\mathcal{N}_{\theta_{i-1,q}}$ for the corresponding networks at $t_{i}$ are the vectors $\pi_{i-1}$ and $q_{i-1}$, respectively. We prefer this approach over passing significantly more points of the network graphs to the next level as this would, especially for large $m$, significantly increase the dimension of the input spaces of the neural networks. In the context of the semi-recurrent architecture for the control process, our choice of examples is particularly interesting as $\pi^{*}$ is ``almost time-homogeneous'' in Example \ref{example:Hestonunconstr}, whereas the optimal control in Example \ref{example:highdimlogutconstr} is highly time-variant. \vspace{0.3cm}\newline
Clearly, the construction of $\pi_{0}$ and $q_{0}$ at $t_{0}=0$ remains the same as in the classical formulation. For each remaining $t_{i}$, we substitute the construction of $\pi_{i}$ and $q_{i}$ according to \eqref{eq:networkprimalSMP} with
\begin{equation}\label{eq:networkprimalSMPrecurrent}
    \pi_{i}:= \mathcal{N}_{\theta_{i,\pi}}(X_{i}, \pi_{i-1}) \hspace{0.3cm}\mbox{and}\hspace{0.3cm} q_{i}:= \mathcal{N}_{\theta_{i,q}}(X_{i}, q_{i-1}).
\end{equation}
In the following, we will also consider both variants where only one of the two processes is modeled via the semi-recurrent architecture, whereas the other one is built by means of the classical scheme. For our numerical experiments, we apply again the learning rate schedules given in Tables \ref{tab:tablehighdimlogutunconstrLR} and \ref{tab:tablehestonunconstrLR}. The results are provided in Tables \ref{tab:tablehighdimlogutrecurrent} and \ref{tab:tableHestonrecurrent}.
\begin{table}[ht]
    \centering
    \begin{tabular}{|c||c|c||c|c||c|}
        \hline
         $(\pi\mbox{-Arch.},q\mbox{-Arch.})$ & $V_{l}$ & $|V_{l}-V|$ & $V_{u}$ & $|V_{u}-V|$ & $\max\{|V_{u}-V|, |V_{l}-V|\}/V$ \Tstrut\Bstrut\\
         \hline \hline
          $(c,c)$ & 2.34308 & 0.00027 & 2.34324 & 0.00011 & 0.00012\\
          $(r,c)$ & 2.34320 & 0.00015 & 2.34337 & 0.00002 & 0.00006\\
          $(c,r)$ & 2.34284 & 0.00051 & 2.34301 & 0.00034 & 0.00022\\
          $(r,r)$ & 2.34295 & 0.00040 & 2.34312 & 0.00023 & 0.00017\\
          \hline
    \end{tabular}
    \caption{Value approximations, absolute and relative errors at the end of the training procedure for the problem studied in Example \ref{example:highdimlogutconstr} in comparison with the theoretical benchmark $V=2.34335$ for various combinations of the classical $(c)$ and the semi-recurrent $(r)$ network architectures.}
    \label{tab:tablehighdimlogutrecurrent}
    \end{table}
    \begin{table}[ht]
    \centering
    \begin{tabular}{|c||c|c||c|c||c|}
        \hline
         $(\pi\mbox{-Arch.},q\mbox{-Arch.})$ & $V_{l}$ & $|V_{l}-V|$ & $V_{u}$ & $|V_{u}-V|$ & $\max\{|V_{u}-V|, |V_{l}-V|\}/V$ \Tstrut\Bstrut\\
         \hline \hline
          $(c,c)$ & 2.04283 & 0.00015 & 2.04298 & 0.00030 & 0.00015\\
          $(r,c)$ & 2.04243 & 0.00025 & 2.04278 & 0.00010 & 0.00012\\
          $(c,r)$ & 2.04273 & 0.00005 & 2.04295 & 0.00027 & 0.00013\\
          $(r,r)$ & 2.04262 & 0.00006 & 2.04296 & 0.00028 & 0.00014\\
          \hline
    \end{tabular}
    \caption{Value approximations, absolute and relative errors at the end of the training procedure for the problem studied in Example \ref{example:Hestonunconstr} with $(T, N)=(0.5, 15)$ in comparison with the theoretical benchmark $V=2.04268$ for various  combinations of the classical $(c)$ and the semi-recurrent $(r)$ network architectures.}
    \label{tab:tableHestonrecurrent}
    \end{table}
    \noindent In contrast to the consequences of introducing an epoch-based learning procedure, the effect of semi-recurrent architectures is rather inhomogeneous depending on the configuration. In Table \ref{tab:tablehighdimlogutrecurrent}, the configurations $(c,r)$ and $(r,r)$ even yield noticeably larger errors. Moreover, the improvements for our stochastic volatility problem are not significant. However, using a semi-recurrent architecture for $\pi$ and the classical one for the BSDE integrand yields for both problems the best value approximation, whence this choice can certainly be regarded as a valuable refinement of the classical version of the algorithm. As we have only considered two problems where already the classical version performs exceptionally well, also both choices $(c,r)$ and $(r,r)$ might be valuable adjustments for other problems. The outstanding performance of the semi-recurrent architecture for $\pi$ while modeling the BSDE integrand classically is in two different ways remarkable: On the one hand, since the optimal control is highly time-variant in Example \ref{example:highdimlogutconstr}, whereas it is ``almost time-homogeneous'' in Example~\ref{example:Hestonunconstr} and on the other hand, as $\pi^{*}(t)$ is a constant random variable for every $t\in [0,T]$ and both problems. As the dimension of the input space of each $\mathcal{N}_{\theta_{i,\pi}}$ is increased (for our high-dimensional example even by a factor of $30$), it should, in theory, become more difficult to learn the constants $\pi^{*}(t_{i})$, which, however, seems to cause no problems for the algorithm since the approximation quality became even better.

\section{Conclusion}\label{section:conclusion}

In the course of this paper, we formulated a deep learning based algorithmic solver for tackling the utility maximization problem with convex constraints in its full generality. For this purpose, we derived a stochastic maximum principle which is also applicable to power utility functions and even more general ones, if certain families of random variables either satisfy a uniform integrability property or there exists a lower bound with integrable negative part. Like the deep SMP algorithm and in contrast to the deep controlled 2BSDE algorithm, our novel algorithm can also handle path dependent random coefficients. Moreover, we conclude from the results of the studied problems with non-trivial constraints that also the upper bound implied by the deep primal SMP algorithm serves as a reasonable value approximation. This is an essential advantage over the deep SMP algorithm as the latter method tends to produce the value of the corresponding unconstrained problem as the upper bound (see Example \ref{example:highdimlogutconstr} for a rationale). Combining this with the fact that the deep primal SMP algorithm outperformed its dual counterpart in Subsection \ref{subsection:trulynonmarkovcoeff} and that it produced more accurate results than both versions of the deep controlled 2BSDE algorithm in our high-dimensional example with deterministic coefficients (see Example \ref{example:highdimlogutconstr}) illustrates and underscores the power of our novel algorithm. Its performance can be improved even further by introducing an epoch-based learning procedure as this setup complies more with the pathwise nature of the maximization condition \eqref{eq:PrimalnonMarkovcondition1}. Moreover, choosing a semi-recurrent network architecture for the control process as well as the BSDE integrand can also be a valuable extension of the algorithm's original formulation. Our experiments suggest that modeling the control process this way while choosing the classical architecture for the BSDE integrand is the most powerful configuration. \vspace{0.3cm}\newline
If one is interested in the control process determined by the algorithm, e.g. in order to apply the strategy in practice, one simply has to consider the neural networks $\mathcal{N}_{\theta_{i,\pi}}$, $i\in\{0,\dots,N-1\}$. This is a decisive advantage over the deep SMP algorithm and the dual version of the deep controlled 2BSDE algorithm from~\cite{zheng1} as we cannot use the primal control process implied by a solution to the dual problem (see \eqref{eq:fromdualtoprimal1}) for this purpose. The reason for this is that, in contrast to the wealth process, the dual state process, i.e. $-1$ times the primal adjoint process according to Theorem \ref{theorem:fromdualtoprimal}, is not directly observable in the market.

\section*{Acknowledgements}
The author thanks Josef Teichmann for his valuable remarks and fruitful discussions in the course of the author's stay at ETH Zurich and the associated master's thesis project.

\begin{appendices}
\section{Python Implementation of the Deep Primal SMP Algorithm as Used in Example \ref{example:highdimlogutconstr}}
\label{section:Appendix}

In the following, we provide our Python code for the deep primal SMP algorithm adapted to the specific problem discussed in Example \ref{example:highdimlogutconstr}. The source code can be quite easily adjusted to random coefficient problems by adding the dynamics, which are required for determining the coefficients, to the loops used for the forward simulation of the processes. For example, we included the stock prices in Example \ref{example:randomcoeffunconstrained}, whereas we added the instantaneous variance process $\nu$ in Example \ref{example:Hestonunconstr}. Hence, these processes are simulated by means of the same Euler-Maruyama scheme. Moreover, adapting Code \ref{code:CodeprimalSMP} to the refinements proposed in Section \ref{section:refining} is also fairly straightforward. 
At the end of training, it suffices in most cases to save just the neural networks $\mathcal{N}_{\theta_{i,\pi}}, i\in\{0,\dots,N-1\}$, as these are the most crucial results for a trader wishing to apply the obtained strategy in practice. This can be done by means of the inherited save-method of the PartNetwork class.\vspace{0.3cm}

\definecolor{mygreen}{rgb}{0,0.6,0}
\definecolor{mygray}{rgb}{0.5,0.5,0.5}
\definecolor{mymauve}{rgb}{0.58,0,0.82}
\definecolor{myred}{rgb}{0.6,0,0}
\definecolor{myorange}{rgb}{0.93, 0.53, 0.18}
\definecolor{darkcoral}{rgb}{0.8, 0.36, 0.27}
\lstset{ 
  backgroundcolor=\color{white},   
  basicstyle=\footnotesize,        
  breakatwhitespace=false,         
  breaklines=true,                 
  captionpos=b,                    
  commentstyle=\color{mygreen},    
  escapeinside={<*}{*>},
  extendedchars=true,              
  firstnumber=1,                
  frame=single,	                   
  keepspaces=true,                 
  keywordstyle=\color{blue},       
  language=Python,                 
  morekeywords={*, with, as},
  numbers=left,                    
  columns=flexible,
  classoffset=1,
  morekeywords={True, False, None},
  deletekeywords={pow},
  keywordstyle=\color{myorange},
  numbersep=8pt,                   
  numberstyle=\tiny\color{mygray}, 
  rulecolor=\color{black},         
  stepnumber=1,                    
  stringstyle=\color{darkcoral}, 
  showstringspaces=false,
  showtabs=false,
  tabsize=2,	                   
}

\lstset{caption={Python code for the deep primal SMP algorithm in the setting of Example \ref{example:highdimlogutconstr}.}}
\begin{lstlisting}[label=code:CodeprimalSMP]
import tensorflow as tf
from tensorflow import keras
import numpy as np
import time
from <*scipy.stats*> import multivariate_normal

class DeepSMP_primal(keras.Model):
    
    def __init__(<*self,*> **kwargs):
        super(DeepSMP_primal, self).__init__(**kwargs)
        <*self.*>m = 30
        <*self.*>T = 0.5
        <*self.*>N = 10
        <*self.*>dt = <*self.*>T/<*self.*>N
        <*self.*>layers_num = 4 # number of hidden + 2
        <*self.*>nodes = [11, 11]  # nodes in hidden layers
        <*self.*>batch_size = 64
        <*self.*>x0 = 10
        <*self.*>schedule1 = <*keras.optimizers.schedules.*>PiecewiseConstantDecay([1000, 3000, 8000], 
                                                 [1e-2, 1e-3, 1e-4, 1e-5])
        <*self.*>schedule2 = [<*keras.optimizers.schedules.*>PiecewiseConstantDecay([1000, 3000, 8000], 
                                                 [1e-3, 1e-4, 1e-5, 1e-6]) for _ in range(<*self.*>N)]
        <*self.*>optimizer1 = <*keras.optimizers.*>Adam(learning_rate = <*self.*>schedule1)
        <*self.*>optimizer2 = [<*keras.optimizers.*>Adam(learning_rate = <*self.*>schedule2[i]) 
                          for i in range(<*self.*>N)]  
        <*self.*>p0 = <*tf.Variable*>(np.random.uniform(low=-0.4, high=-0.2), trainable=True)
        <*self.*>pi0 = <*tf.Variable*>(np.random.uniform(size=(1, <*self.*>m), low=0, high=0.2), 
                            trainable=True, dtype=<*tf.*>float32, 
                            constraint=lambda x: <*tf.*>where(x<-1/<*self.*>m, -1/<*self.*>m, x))
        <*self.*>Q0 = <*tf.Variable*>(np.random.uniform(size=(1, <*self.*>m), low=-0.1, high=0.1), 
                           trainable=True, dtype=<*tf.*>float32)
        <*self.*>ModelQ = [<*self.*>Q0] + [PartNetwork(<*self.*>m, <*self.*>layers_num, <*self.*>nodes, isQ=True) 
                                 for _ in range(<*self.*>N-1)]
        <*self.varforloss1*> = [<*self.*>p0] + <*self.*>ModelQ
        <*self.*>Modelpi = [<*self.*>pi0] + [PartNetwork(<*self.*>m, <*self.*>layers_num, <*self.*>nodes, isQ=False) 
                                  for _ in range(<*self.*>N-1)]
        <*self.*>history_time = []
        <*self.*>history_p0 = []
        <*self.*>mc_size = 100000
        <*self.*>history_bound_u = []
        <*self.*>history_bound_l = []
    
    def build(self):
        for i in range(<*self.*>N-1):
            <*self.*>Modelpi[i+1](<*tf.*>zeros(shape=(1, 1)), training = False)
            <*self.*>ModelQ[i+1](<*tf.*>zeros(shape=(1, 1)), training = False)
        return
    
    def sigma(<*self,*> t):
        sigma = (0.3*(1+<*tf.*>sqrt(t))-0.1)*<*tf.*>eye(<*self.*>m) + 0.1*<*tf.*>ones(shape=(<*self.*>m, <*self.*>m))
        return sigma
    
    def sigma_inv(<*self,*> t):
        sigma_inv = <*tf.linalg.inv*>(<*self.*>sigma(t))
        return sigma_inv
    
    def mu(<*self, t, size*>):
        helper = <*tf.*>expand_dims(<*\textcolor{black}{tf.range}*>(1, limit=<*self.*>m+1, dtype = <*tf.*>float32), axis=0)
        mu = 0.07 + 0.02*<*tf.*>sin(4*np.pi*t+2*np.pi*helper/<*self.*>m)
        return mu*tf.ones(shape=(size, <*self.*>m))
    
    def r(<*self, t, size*>):
        r = 0.06*<*tf.*>exp(0.5*t)
        return r*<*tf.*>ones(shape=(size, 1))
    
    def theta(<*self, t, size*>):
        theta = <*tf.*>transpose(<*tf.*>matmul(<*self.*>sigma_inv(t), <*self.*>mu(t, size)-<*self.*>r(t, size), 
                                     transpose_b = True))
        return theta
    
    def U_transform(<*self,*> x):
        x = <*tf.*>expand_dims(x, axis = -1)
        return <*tf.*>where(x>0, -1-<*tf.*>math.log(<*tf.*>where(x>0, x, 1)), 0)
    
    def g(<*self,*> x):
        x = <*tf.*>expand_dims(x, axis = -1)
        return <*tf.*>where(x>0, <*tf.*>math.log(<*tf.*>where(x>0, x, 1)), 0)
    
    def gx(<*self,*> x):
        x = <*tf.*>expand_dims(x, axis = -1)
        return <*tf.*>where(x>0, <*tf.*>pow(<*tf.*>where(x>0, x, 1), -1), 0)
    
    def loss1(<*self,*> x, p):
        return <*tf.*>reduce_mean(<*tf.*>square(<*tf.*>expand_dims(p, axis = -1) + <*self.*>gx(x)))
    
    def loss2(<*self,*> t, p, pi, q, size):
        return (<*tf.*>reduce_mean(<*tf.*>reduce_sum(pi * <*tf.*>transpose(<*tf.*>matmul(<*self.*>sigma(t), 
                <*tf.*>expand_dims(p, axis = -1) * <*self.*>theta(t, size) + q, transpose_b = True)), 
                axis = 1, keepdims = True)))
    
    def bounds(self):
        dW = <*tf.constant(multivariate\_normal.*>rvs(size=[<*self.*>mc_size, <*self.*>m, <*self.*>N]), 
                        dtype=<*tf.*>float32)
        if <*self.*>m == 1:
            dW = <*tf.*>expand_dims(dW, axis = 1)
        X, P = <*self.*>simulate1(dW, size = <*self.*>mc_size, training = False)
        
        bound_l = <*tf.*>reduce_mean(<*self.*>g(X))
        bound_u = <*tf.*>reduce_mean(<*self.*>U_transform(-P)) - <*self.*>p0*<*self.*>x0
            
        return [bound_l, bound_u]
    
    <*\textcolor{red}{$@$ tf.autograph.experimental.do$\_$not$\_$convert}*>
    def simulate1(<*self,*> dW, size, training):
        
        X = <*self.*>x0 * <*tf.*>ones(size)
        P = <*self.*>p0 * <*tf.*>ones(size)
        
        pi = <*self.*>pi0 * <*tf.*>ones(shape=(size, <*self.*>m))
        Q = <*self.*>Q0 * <*tf.*>ones(shape=(size, <*self.*>m))
        
        P = P - <*tf.*>squeeze(((<*self.r*>(0.0, size) + <*tf.*>reduce_sum(<*tf.*>matmul(pi, <*self.*>sigma(0.0))
                * <*self.theta*>(0.0, size), axis = 1, keepdims = True)) * <*tf.*>expand_dims(P, axis = -1)
                + <*tf.*>reduce_sum(Q * <*tf.*>matmul(pi, <*self.*>sigma(0.0)), axis = 1, keepdims = True)) 
                * <*self.*>dt, axis = 1) \
                + <*tf.*>reduce_sum(dW[:, :, 0] * Q * <*tf.*>sqrt(<*self.*>dt), axis = 1)
        X = X + <*tf.*>squeeze(<*tf.*>expand_dims(X, axis = -1) * (<*self.r*>(0.0, size) + <*tf.*>reduce_sum(
                <*tf.*>matmul(pi, <*self.*>sigma(0.0)) * <*self.theta*>(0.0, size), axis = 1, keepdims = True)) 
                * <*self.*>dt, axis = 1) \
                + <*tf.*>reduce_sum(dW[:, :, 0] * <*tf.*>expand_dims(X, axis = -1) 
                * <*tf.*>matmul(pi, <*self.*>sigma(0.0)) * <*tf.*>sqrt(<*self.*>dt), axis = 1)
        
        for i in range(<*self.*>N-1):
            
            pi = <*self.*>Modelpi[i+1](<*tf.*>expand_dims(X, axis = -1), <*training)*>
            Q = <*self.*>ModelQ[i+1](<*tf.*>expand_dims(X, axis = -1), <*training)*>
            
            P = P - <*tf.*>squeeze(((<*self.*>r(<*self.*>dt*(i+1), size) + <*tf.*>reduce_sum(<*tf.*>matmul(pi, 
                    <*self.*>sigma(<*self.*>dt*(i+1))) * <*self.*>theta(<*self.*>dt*(i+1), size), axis = 1, 
                    keepdims = True)) * <*tf.*>expand_dims(P, axis = -1)
                    + <*tf.*>reduce_sum(Q * <*tf.*>matmul(pi, <*self.*>sigma(<*self.*>dt*(i+1))), axis = 1, 
                    keepdims = True)) * <*self.*>dt, axis = 1) \
                    + <*tf.*>reduce_sum(dW[:, :, i+1] * Q * <*tf.*>sqrt(<*self.*>dt), axis = 1)
            X = X + <*tf.*>squeeze(<*tf.*>expand_dims(X, axis = -1) * (<*self.*>r(<*self.*>dt*(i+1), size) 
                    + <*tf.*>reduce_sum(<*tf.*>matmul(pi, <*self.*>sigma(<*self.*>dt*(i+1))) * <*self.*>theta(
                    <*self.*>dt*(i+1), size), axis = 1, keepdims = True)) * <*self.*>dt, axis = 1) \
                    + <*tf.*>reduce_sum(dW[:, :, i+1] * <*tf.*>expand_dims(X, axis = -1) 
                    * <*tf.*>matmul(pi, <*self.*>sigma(<*self.*>dt*(i+1))) * <*tf.*>sqrt(<*self.*>dt), axis = 1)
        
        return X, P
    
    <*\textcolor{red}{$@$ tf.function}*>
    def optimize1(<*self,*> dW):
        with <*tf.*>GradientTape(watch_accessed_variables = False) as tape:
            <*tape.*>watch(<*self.varforloss1.trainable$\_$variables*>)
            
            X, P = <*self.*>simulate1(dW, size = <*self.*>batch_size, training = True)
            loss1 = <*self.*>loss1(X, P)
            
        grad1 = <*tape.*>gradient(loss1, <*self.varforloss1.trainable$\_$variables*>)
        <*self.optimizer1.*>apply_gradients(zip(grad1, <*self.varforloss1.trainable$\_$variables*>))
        return
    
    <*\textcolor{red}{$@$ tf.function}*>
    def optimize2_0(<*self,*> optimizer):
        with <*tf.*>GradientTape(watch_accessed_variables = False) as tape:
            <*tape.*>watch(<*self.*>pi0)
            
            # X = <*\textcolor{mygreen}{self.}*>x0 * <*\textcolor{mygreen}{tf.}*>ones(<*\textcolor{mygreen}{self.}*>batch_size)
            P = <*self.*>p0 * <*tf.*>ones(<*self.*>batch_size)
            
            pi = <*self.*>pi0 * <*tf.*>ones(shape=(<*self.*>batch_size, <*self.*>m))
            Q = <*self.*>Q0 * <*tf.*>ones(shape=(<*self.*>batch_size, <*self.*>m))
            
            loss2 = <*self.*>loss2(0.0, P, pi, Q, size = <*self.*>batch_size)
            
        grad = <*tape.*>gradient(loss2, <*self.*>pi0)
        <*optimizer.*>apply_gradients([(grad, <*self.*>pi0)])
        return
    
    <*\textcolor{red}{$@$ tf.autograph.experimental.do$\_$not$\_$convert}*>
    def simulate2(<*self,*> dW, N, X, P, pi, Q):
        
        P = P - <*tf.*>squeeze(((<*self.*>r(<*self.*>dt*(N-1), <*self.*>batch_size) + <*tf.*>reduce_sum(<*tf.*>matmul(pi, 
                <*self.*>sigma(<*self.*>dt*(N-1))) * <*self.*>theta(<*self.*>dt*(N-1), <*self.*>batch_size), axis = 1, 
                keepdims = True)) * <*tf.*>expand_dims(P, axis = -1)
                + <*tf.*>reduce_sum(Q * <*tf.*>matmul(pi, <*self.*>sigma(<*self.*>dt*(N-1))), axis = 1, 
                keepdims = True)) * <*self.*>dt, axis = 1) \
                + <*tf.*>reduce_sum(dW[:, :, N-1] * Q * <*tf.*>sqrt(<*self.*>dt), axis = 1)
        X = X + <*tf.*>squeeze(<*tf.*>expand_dims(X, axis = -1) * (<*self.*>r(<*self.*>dt*(N-1), <*self.*>batch_size) 
                + <*tf.*>reduce_sum(<*tf.*>matmul(pi, <*self.*>sigma(<*self.*>dt*(N-1))) * <*self.*>theta(
                <*self.*>dt*(N-1), <*self.*>batch_size), axis = 1, keepdims = True)) * <*self.*>dt, axis = 1) \
                + <*tf.*>reduce_sum(dW[:, :, N-1] * <*tf.*>expand_dims(X, axis = -1)
                * <*tf.*>matmul(pi, <*self.*>sigma(<*self.*>dt*(N-1))) * <*tf.*>sqrt(<*self.*>dt), axis = 1)
            
        pi = <*self.*>Modelpi[N](<*tf.*>expand_dims(X, axis = -1), training = True)
        Q = <*self.*>ModelQ[N](<*tf.*>expand_dims(X, axis = -1), training = True)
        
        return X, P, pi, Q
        
    def optimize2(<*self,*> dW, N, X, P, pi, Q, optimizer, Model):
        with <*tf.*>GradientTape(watch_accessed_variables = False) as tape:
            <*tape.*>watch(<*Model.trainable\_variables*>)
            
            X, P, pi, Q = <*self.*>simulate2(dW, N, X, P, pi, Q)
            loss = <*self.*>loss2(N*<*self.*>dt, P, pi, Q, size = <*self.*>batch_size)
            
        grad = <*tape.*>gradient(loss, <*Model.trainable\_variables*>)
        <*optimizer.*>apply_gradients(zip(grad, <*Model.trainable\_variables*>))
        return X, P, pi, Q
    
    def optimize2prepare(self):
        <*self.optimizecontrol*> = [<*tf.*>function(<*self.*>optimize2<*).get\_concrete\_function(*>
            <*tf.*>TensorSpec(shape=[<*self.*>batch_size, <*self.*>m, <*self.*>N], dtype=<*tf.*>float32), i+1, 
            <*tf.*>TensorSpec(shape=[<*self.batch\_size],*> dtype=<*tf.*>float32), 
            <*tf.*>TensorSpec(shape=[<*self.batch\_size],*> dtype=<*tf.*>float32), 
            <*tf.*>TensorSpec(shape=[<*self.*>batch_size, <*self.*>m], dtype=<*tf.*>float32), 
            <*tf.*>TensorSpec(shape=[<*self.*>batch_size, <*self.*>m], dtype=<*tf.*>float32), 
            <*self.*>optimizer2[i+1], <*self.*>Modelpi[i+1]) for i in range(<*self.*>N-1)]
        return
    
    def train(<*self,*> steps):
        time_start = time.time()
        for k in range(steps):   
            # dW does not have the desired shape, if BBdim==1
            dW = <*tf.*>constant(<*multivariate\_normal.rvs*>(size=[<*self.*>batch_size, <*self.*>m, <*self.*>N]), 
                            dtype=<*tf.*>float32)
            if <*self.*>m == 1:
                dW = <*tf.*>expand_dims(dW, axis = 1)
            
            <*self.*>optimize1(dW)
                
            optimizer = <*self.*>optimizer2[0]
            <*self.*>optimize2_0(optimizer)
            
            X = <*self.*>x0 * <*tf.*>ones(<*self.*>batch_size)
            P = <*self.*>p0 * <*tf.*>ones(<*self.*>batch_size)
            
            pi = <*self.*>pi0 * <*tf.*>ones(shape=(<*self.*>batch_size, <*self.*>m))
            Q = <*self.*>Q0 * <*tf.*>ones(shape=(<*self.*>batch_size, <*self.*>m))
                
            for i in range(<*self.*>N-1):
                X, P, pi, Q = <*self.optimizecontrol*>[i](dW, X, P, pi, Q)
            
            <*self.history\_time.*>append(time.time()-time_start)
            <*self.history\_p0.*>append(<*self.*>p0.numpy())
            
            if k%50 == 0:
                  print("Step: %d, Time: %.2f, p_0: %.4f"
                        % (k, <*self.*>history_time[-1], <*self.*>history_p0[-1]))

            if k%200 == 0:
                helper = <*self.*>bounds()
                <*self.history\_bound\_l.*>append(helper[0].numpy())
                <*self.history\_bound\_u.*>append(helper[1].numpy())
                print("Step: %d, Bound_l: %.4f, Bound_u: %.4f"
                       % (k, <*self.*>history_bound_l[-1], <*self.*>history_bound_u[-1]))
        return 

        
class PartNetwork(keras.Model):
    
    def __init__(<*self,*> m, layers_num, nodes, isQ, **kwargs):
        super(PartNetwork, self<*).*>__init__(**kwargs)
        <*self.*>m = m
        <*self.*>layers_num = layers_num
        <*self.*>nodes = nodes
        <*self.*>isQ = isQ
        <*self.*>outdim = <*self.*>m

        <*self.*>bnorm_layers = [<*keras.layers.*>BatchNormalization(epsilon=100) 
                            for _ in range(<*self.*>layers_num-1)]
        <*self.*>dense_layers = [<*keras.layers.*>Dense(nodes[i], use_bias=False, activation=None) 
                            for i in range(<*self.*>layers_num-2)]
        <*self.dense\_layers.append*>(<*keras.layers.*>Dense(<*self.*>outdim, activation=None))
        
    def call(<*self,*> x, <*training):*>
        x = <*self.*>bnorm_layers[0](x, <*training)*>
        for i in range(<*self.*>layers_num-2):
            x = <*self.*>dense_layers[i](x)
            x = <*self.*>bnorm_layers[i+1](x, <*training)*>
            x = <*tf.nn.relu*>(x)
        x = <*self.*>dense_layers[<*self.*>layers_num-2](x)
        if <*self.*>isQ == False:
            x = x*x-1/<*self.*>m
        return x
    
    
Model = DeepSMP_primal()
Model.build()
Model.optimize2prepare()
Model.train(10000)
\end{lstlisting}

\section{Learning Rate Schedules as Applied in Our Numerical Experiments}
\label{section:AppendixLRtables}
In the final part of our appendices, we gather the learning rate schedules which were applied in the course of the numerical experiments from the previous sections. As these schedules can be viewed as ``hyperparameters'', adapting them depending on the specific problem and algorithm is highly recommendable. The notation $x\stackrel{y}{\rightarrow}$ means that the learning rate level $x$ will be applied throughout the following $y$ steps. If there is only one schedule given in the final column for the rows labeled as 2BSDE{\_}dual and SMP and the associated problem is constrained, then the learning rate schedules for the dual control and the dual initial value agree. Finally, we introduce two symbols due to the space constraints in Table \ref{tab:tablehighdimlogutunconstrLR}. $\overrightarrow{\alpha}$ stands for the schedule of the corresponding (2)BSDE part, whereas in the case of $\overrightarrow{\beta}$ it is divided by 10.
\begin{table}[H]
\centering
\begin{tabular}{|c||c|c|}
    \hline
     Algorithm & LR $\mathcal{L}_{BSDE}$  & LR $\mathcal{L}_{control}^{i}$ \big/ LR $\mathcal{L}_{dual}$\Tstrut\Bstrut\\
     \hline \hline
      SMP & $10^{-2}\stackrel{300}{\rightarrow} 10^{-3}\stackrel{1700}{\rightarrow}10^{-4} $ & $10^{-2}\stackrel{300}{\rightarrow} 10^{-3}\stackrel{1700}{\rightarrow}10^{-4}$\Tstrut \\
      SMP{\_}primal & $10^{-2}\stackrel{300}{\rightarrow} 10^{-3}\stackrel{1700}{\rightarrow}10^{-4} $ & $10^{-3}\stackrel{300}{\rightarrow} 10^{-4}\stackrel{1700}{\rightarrow}10^{-5} $ \\
      \hline
\end{tabular}
\caption{Piecewise constant learning rate schedules applied in Example \ref{example:randomcoeffunconstrained}.}
\label{tab:tablerandomcoeff}
\end{table}
\begin{table}[H]
    \centering
    \begin{tabular}{|c||c|c|c|}
        \hline
         Algorithm & LR $\mathcal{L}_{(2)BSDE}$  & LR $\mathcal{L}_{control}^{i}$ \big/ LR $\mathcal{L}_{dual}$\Tstrut\Bstrut\\
         \hline \hline
          2BSDE & $10^{-2}\stackrel{2000}{\rightarrow} 10^{-3}\stackrel{3000}{\rightarrow}10^{-4} \stackrel{3000}{\rightarrow}10^{-5}$ & $10^{-3}\stackrel{2000}{\rightarrow} 10^{-4}\stackrel{3000}{\rightarrow}10^{-5} \stackrel{3000}{\rightarrow}10^{-6}$ \Tstrut\\
          2BSDE{\_}dual & $10^{-2}\stackrel{1000}{\rightarrow} 10^{-3}\stackrel{2000}{\rightarrow}10^{-4} \stackrel{3000}{\rightarrow}10^{-5}$ & $\overrightarrow{\beta}$ \big/ $10^{-2}\stackrel{200}{\rightarrow} 10^{-4}\stackrel{800}{\rightarrow}10^{-5} \stackrel{7000}{\rightarrow}10^{-6}$ \\
          SMP & $10^{-2}\stackrel{2000}{\rightarrow} 10^{-3}\stackrel{3000}{\rightarrow}10^{-4} \stackrel{3000}{\rightarrow}10^{-5}$ & $\overrightarrow{\alpha}$ \big/ $10^{-2}\stackrel{200}{\rightarrow} 10^{-4}\stackrel{800}{\rightarrow}10^{-5} \stackrel{7000}{\rightarrow}10^{-6}$ \\
          SMP{\_}primal & $10^{-2}\stackrel{1000}{\rightarrow} 10^{-3}\stackrel{2000}{\rightarrow}10^{-4} \stackrel{5000}{\rightarrow}10^{-5}$ & $10^{-3}\stackrel{1000}{\rightarrow} 10^{-4}\stackrel{2000}{\rightarrow}10^{-5} \stackrel{5000}{\rightarrow}10^{-6}$ \\
          \hline
    \end{tabular}
    \caption{Piecewise constant learning rate schedules applied in Example \ref{example:highdimlogutconstr}.}
    \label{tab:tablehighdimlogutunconstrLR}
\end{table}
\begin{table}[H]
    \centering
    \begin{tabular}{|c||c|c|c|}
        \hline
         Algorithm & LR $\mathcal{L}_{(2)BSDE}$  & LR $\mathcal{L}_{control}^{i}$ \big/ LR $\mathcal{L}_{dual}$\Tstrut\Bstrut\\
         \hline \hline
          2BSDE & $10^{-2}\stackrel{500}{\rightarrow} 10^{-3}\stackrel{1500}{\rightarrow}10^{-4} \stackrel{4000}{\rightarrow}10^{-5}$ & $10^{-3}\stackrel{500}{\rightarrow} 10^{-4}\stackrel{1500}{\rightarrow}10^{-5} \stackrel{4000}{\rightarrow}10^{-6}$ \Tstrut\\
          2BSDE{\_}dual & $10^{-2}\stackrel{1000}{\rightarrow} 10^{-3}\stackrel{2000}{\rightarrow}10^{-4} \stackrel{3000}{\rightarrow}10^{-5}$ & $10^{-3}\stackrel{1000}{\rightarrow} 10^{-4}\stackrel{2000}{\rightarrow}10^{-5} \stackrel{3000}{\rightarrow}10^{-6}$ \\
          SMP & $10^{-2}\stackrel{1000}{\rightarrow} 10^{-3}\stackrel{2000}{\rightarrow}10^{-4} \stackrel{3000}{\rightarrow}10^{-5}$ & $10^{-3}\stackrel{1000}{\rightarrow} 10^{-4}\stackrel{2000}{\rightarrow}10^{-5} \stackrel{3000}{\rightarrow}10^{-6}$ \\
          SMP{\_}primal & $10^{-2}\stackrel{200}{\rightarrow} 10^{-3}\stackrel{800}{\rightarrow}10^{-4} \stackrel{4000}{\rightarrow}10^{-5}$ & $10^{-3}\stackrel{200}{\rightarrow} 10^{-4}\stackrel{800}{\rightarrow}10^{-5} \stackrel{4000}{\rightarrow}10^{-6}$ \\
          \hline
    \end{tabular}
    \caption{Piecewise constant learning rate schedules applied in Example \ref{example:Hestonunconstr}.}
    \label{tab:tablehestonunconstrLR}
\end{table}
\begin{table}[H]
    \centering
    \begin{tabular}{|c||c|c|c|}
        \hline
         Algorithm & LR $\mathcal{L}_{(2)BSDE}$  & LR $\mathcal{L}_{control}^{i}$ \big/ LR $\mathcal{L}_{dual}$\Tstrut\Bstrut\\
         \hline \hline
          2BSDE & $10^{-2}\stackrel{1000}{\rightarrow} 10^{-3}\stackrel{2000}{\rightarrow}10^{-4} \stackrel{2000}{\rightarrow}10^{-5}$ & $10^{-3}\stackrel{1000}{\rightarrow} 10^{-4}\stackrel{2000}{\rightarrow}10^{-5} \stackrel{2000}{\rightarrow}10^{-6}$ \Tstrut\\
          2BSDE{\_}dual & $10^{-2}\stackrel{1000}{\rightarrow} 10^{-3}\stackrel{2000}{\rightarrow}10^{-4} \stackrel{2000}{\rightarrow}10^{-5}$ & $10^{-3}\stackrel{1000}{\rightarrow} 10^{-4}\stackrel{2000}{\rightarrow}10^{-5} \stackrel{2000}{\rightarrow}10^{-6}$ \\
          SMP & $10^{-2}\stackrel{1000}{\rightarrow} 10^{-3}\stackrel{2000}{\rightarrow}10^{-4} \stackrel{2000}{\rightarrow}10^{-5}$ & $10^{-2}\stackrel{1000}{\rightarrow} 10^{-3}\stackrel{2000}{\rightarrow}10^{-4} \stackrel{2000}{\rightarrow}10^{-5}$ \\
          SMP{\_}primal & $10^{-2}\stackrel{1000}{\rightarrow} 10^{-3}\stackrel{2000}{\rightarrow}10^{-4} \stackrel{2000}{\rightarrow}10^{-5}$ & $10^{-3}\stackrel{1000}{\rightarrow} 10^{-4}\stackrel{2000}{\rightarrow}10^{-5} \stackrel{2000}{\rightarrow}10^{-6}$ \\
          \hline
    \end{tabular}
    \caption{Piecewise constant learning rate schedules applied in Example \ref{example:vasicekconstrhighdim}.}
    \label{tab:tablevasicekconstrLR}
\end{table}
    
\end{appendices}

\bibliographystyle{plain}
\bibliography{references}  

\end{document}